%% file: main.tex
\documentclass[letterpaper,twocolumn,10pt]{article}
\usepackage{usenix}

\usepackage{tikz}
\usepackage{amsmath}

\usepackage{filecontents}
\input{packages}

\usepackage[affil-it]{authblk}

\begin{document}

\date{}

\title{\toolname: Towards Verifiable LLM Routing}

\author[1]{Mengyu Yao}
\author[2]{Ziqi Zhang}
\author[1]{Shaofei Li}
\author[1]{Yifeng Cai}
\author[1]{Ding Li}
\author[1]{Yao Guo}
\author[3]{Ning Luo}

\affil[1]{MOE Key Lab of HCST (PKU), School of Computer Science, Peking University}
\affil[2]{Department of Computer Science, University of Illinois Urbana-Champaign}
\affil[3]{Department of Electrical and Computer Engineering, University of Illinois Urbana-Champaign}

\thispagestyle{empty}

\maketitle

\begin{abstract}

Large language model (LLM) routers are increasingly deployed to optimize quality, latency, and cost through dynamic model selection and request rewriting. However, a compromised router can cause ineligible selection on the control plane, unendorsed rewriting or unfaithful relay on the data plane, undermining LLM routing integrity. 
We present \toolname, a verifiable LLM routing framework. Under an endorsed public context binding a private semantic index and rewrite banks, \toolname combines ZK proofs with a provider-issued serving certificate to establish integrity across both planes while keeping the router's private routing assets hidden from clients. For model selection, we design an efficient ZK proof for semantic search. For rewriting, local checks enforce message preservation and additions from committed banks.
On real-world workloads, \toolname's model-selection and rewrite proofs average 21.92~s and 4.24~s, yielding 1,307$\times$ and 101$\times$ speedups over the fastest completing baselines. Relative to serving without verification, \toolname adds only 4.08~s under per-turn selection and 0.11~s under session-pinned selection.
\end{abstract}

\input{section/intro_v4}

\input{section/background}
\input{section/problem_statement}

\input{section/buildingblock_v2}
\input{section/method_v3}
\input{section/evaluation_v2}

\input{section/discussion}

\input{section/conclusion}
\appendix


\bibliographystyle{plainurl}
\bibliography{reference}

\input{section/appendix}

\end{document}

%% file: packages.tex
\usepackage{amsmath,amssymb,amsfonts}
\usepackage{graphicx}
\usepackage{textcomp}
\def\BibTeX{{\rm B\kern-.05em{\sc i\kern-.025em b}\kern-.08em
    T\kern-.1667em\lower.7ex\hbox{E}\kern-.125emX}}

\usepackage[all]{nowidow}
\usepackage{placeins}

\usepackage{soul}
\usepackage{listings}
\usepackage{url}
\usepackage{makecell}
\usepackage{hyperref}
\usepackage{adjustbox}
\usepackage{xspace}
\usepackage{multirow}
\usepackage{subfloat}
\usepackage{balance}
\usepackage{tikz}
\usepackage{ctable}
\usepackage{listings}
\usepackage[nobiblatex]{xurl}
\usepackage{mathrsfs}
\usepackage[normalem]{ulem}
\useunder{\uline}{\ul}{}

\usepackage{algpseudocode}
\usepackage{caption}

\usepackage{acronym}
\usepackage{fancybox}
\usepackage{verbatim}
\usepackage[normalem]{ulem}
\usepackage{float}
\usepackage{newfloat}
\usepackage{mdframed}

\usepackage{booktabs}

\usepackage{subfigure}
\usepackage[flushleft]{threeparttable}
\usepackage{booktabs}
\usepackage{tabularray}

\usepackage[utf8]{inputenc}
\usepackage[skins]{tcolorbox}
\usepackage{amsthm}

\newtheorem{lemma}{Lemma}

\usepackage{enumitem}
\usepackage{pifont}
\usepackage{tcolorbox}
\usepackage{subcaption}
\usepackage{multirow}
\usepackage[normalem]{ulem}
\usepackage{algorithm}
\usepackage{algpseudocode}

\usepackage{xspace}
\usepackage{ulem}

\definecolor{comp1}{RGB}{212,235,255}
\definecolor{comp2}{RGB}{255,245,200}
\definecolor{comp3}{RGB}{255,228,196}

\definecolor{goodmodbg}{RGB}{226,242,220}
\definecolor{tamperdelbg}{RGB}{255,245,200}
\definecolor{tamperinsbg}{RGB}{212,235,255}
\definecolor{tamperrwbg}{RGB}{255,228,196}

\usepackage{soul}

\newcommand{\ControlPlane}{\emph{control plane}}
\newcommand{\DataPlane}{\emph{data plane}}

\newcommand{\clientreq}{R_c}
\newcommand{\receivedreq}{R_m}
\newcommand{\rid}{\mathsf{rid}}

\newcommand{\encoder}{\mathsf{Encoder}}
\newcommand{\ustep}{N_{\mathsf{upper}}^{\max}}
\newcommand{\score}{\mathsf{score}}

\newcommand{\trans}{\mathsf{T}}
\newcommand{\gapset}{\mathsf{MsgBank}}
\newcommand{\fragset}{\mathsf{FragBank}}
\newcommand{\indexset}{\mathcal{I}}
\newcommand{\rtable}{\mathcal{T}_r}
\newcommand{\dtable}{\mathcal{T}_s}
\newcommand{\gtable}{\mathcal{T}_g}

\newcommand{\ZKProve}{\mathsf{ZKProve}}
\newcommand{\ZKVerify}{\mathsf{ZKVerify}}

\newcommand{\LUT}{\mathsf{LUT}}
\newcommand{\MSetEq}{\mathsf{MSetEq}}

\newcommand{\Commit}{\mathsf{Com}}

\usepackage{amsthm}
\newtheorem{theorem}{Theorem}
\usepackage{xcolor}

\definecolor{RiseGreen}{RGB}{200,230,200}
\definecolor{FallGray}{RGB}{230,230,230}

\usepackage{siunitx}
\usepackage{graphicx}

\newcommand{\toolname}{\textsc{VeriRouter}\xspace}
\newtheorem{definition}{Definition}

\definecolor{continuation}{RGB}{237, 125, 49}
\definecolor{keyword}{RGB}{68, 114, 196}
\definecolor{ours}{RGB}{84,130,53}
\definecolor{warn}{RGB}{216, 23, 23}
\definecolor{good}{RGB}{20, 151, 12}
\definecolor{oommuted}{RGB}{135,100,100}
\usepackage{fontawesome5}

\newcommand{\StackPushSeq}{\mathsf{PushSeq}}
\newcommand{\StackPushLen}{N_{\mathsf{push}}^{\mathsf{stk}}}
\newcommand{\StackPushItem}[1]{u_{#1}}

\newcommand{\FinStack}{\mathsf{Final}_{\mathsf{stk}}}

\newcommand{\StackOrd}{\prec_{\mathsf{stk}}}

\newcommand{\EvPush}[1]{\mathsf{Push}(#1)}
\newcommand{\EvPop}[1]{\mathsf{Pop}(#1)}

\newcommand{\Trace}{\mathsf{Trace}}

\newcommand{\PoppedItem}[2]{\mathsf{popped}_{#1,#2}}
\newcommand{\NumPops}[1]{\mathsf{nPops}_{#1}}
\newcommand{\NumPopsVec}{\mathbf{\mathsf{nPops}}}

\newcommand{\ParentTab}{\mathsf{ParentMap}}
\newcommand{\ParentOf}[1]{\mathsf{parent}(#1)}
\newcommand{\NoParent}{\bot}

\newcommand{\TopTrace}{\mathsf{TopTrace}}
\newcommand{\StackTop}[2]{\mathsf{top}_{#1,#2}}

\newcommand{\StackPushLenMax}{N_{\mathsf{push}}^{\mathsf{stk},\max}}
\newcommand{\NumPopsMax}{N_{\mathsf{popblk}}^{\max}}
\newcommand{\FinLenMax}{N_{\mathsf{fin}}^{\mathsf{stk},\max}}

\newcommand{\PriorityOrd}{\prec_{\mathsf{pq}}}
\newcommand{\HigherPrio}{\succ_{\mathsf{pq}}}

\newcommand{\PushLen}{N_{\mathsf{push}}}
\newcommand{\PopLen}{N_{\mathsf{pop}}}
\newcommand{\PQFinLen}{N_{\mathsf{fin}}^{\mathsf{pq}}}
\newcommand{\OutLen}{N_{\mathsf{out}}}
\newcommand{\OpLen}{N_{\mathsf{op}}}
\newcommand{\PushLenMax}{N_{\mathsf{push}}^{\max}}
\newcommand{\PopLenMax}{N_{\mathsf{pop}}^{\max}}
\newcommand{\PQFinLenMax}{N_{\mathsf{fin}}^{\mathsf{pq},\max}}

\newcommand{\PushTrace}{\mathsf{PushTrace}}
\newcommand{\PopTrace}{\mathsf{PopTrace}}
\newcommand{\FinalPQ}{\mathsf{Final}_{\mathsf{pq}}}
\newcommand{\OutSeq}{\mathsf{OutSeq}}

\newcommand{\PushTime}[1]{\mathsf{pushPos}_{#1}}
\newcommand{\PushItem}[1]{u^{\mathsf{push}}_{#1}}
\newcommand{\PopTime}[1]{\mathsf{popPos}_{#1}}
\newcommand{\PopItem}[1]{u^{\mathsf{pop}}_{#1}}

\newcommand{\OutItem}[1]{u^{\mathsf{out}}_{#1}}
\newcommand{\PushElem}[1]{(\PushTime{#1},\PushItem{#1})}

\newcommand{\RemovalTab}{\mathcal{T}_{\mathsf{rem}}}
\newcommand{\ParentMapPQ}{\mathsf{ParentMap}^{\mathsf{pq}}}

\newcommand{\StackOrdPQ}{\StackOrd^{\mathsf{pq}}}
\newcommand{\xstkPQ}{x_{\mathsf{stk}}^{\mathsf{pq}}}
\newcommand{\wstkPQ}{w_{\mathsf{stk}}^{\mathsf{pq}}}

\newcommand{\xpq}{x_{\mathsf{pq}}}
\newcommand{\wpq}{w_{\mathsf{pq}}}
\newcommand{\Cpq}{\mathsf{C}_\mathsf{pq}}
\newcommand{\pipq}{\pi_{\mathsf{pq}}}

\newcommand{\OpPosVec}{\mathsf{OpPos}}
\newcommand{\InPosVec}{\mathsf{InPos}}
\newcommand{\OutPosVec}{\mathsf{OutPos}}
\newcommand{\ParentOutPosVec}{\mathsf{ParentOutPos}}
\newcommand{\InPos}[1]{\InPosVec[#1]}
\newcommand{\OutPos}[1]{\OutPosVec[#1]}
\newcommand{\ParentOutPos}[1]{\ParentOutPosVec[#1]}
\newcommand{\PLIndex}[1]{p_{#1}}

\newcommand{\CandQ}{C}
\newcommand{\WorkQ}{W}
\newcommand{\EntryNode}{u_{\mathsf{L0}}}
\newcommand{\UpperNode}[1]{u_{#1}^{\mathsf{upper}}}
\newcommand{\EfCap}{\mathit{ef}}

\newcommand{\WorkUpdTrace}{\mathsf{UpdSeq}_{\WorkQ}}
\newcommand{\WorkPushTrace}{\mathsf{PushSeq}_{\WorkQ}}
\newcommand{\CandPushTrace}{\mathsf{PushTrace}_{\CandQ}}
\newcommand{\CandPopTrace}{\mathsf{PopTrace}_{\CandQ}}
\newcommand{\CandPushLen}{N_{\mathsf{push}}^{\CandQ}}
\newcommand{\CandPopLen}{N_{\mathsf{pop}}^{\CandQ}}

\newcommand{\WorkEvictItem}[1]{u^{\mathsf{evict}}_{#1}}

\newcommand{\CandPushTime}[1]{\tau^{\CandQ,\mathsf{push}}_{#1}}
\newcommand{\CandPushItem}[1]{u^{\CandQ,\mathsf{push}}_{#1}}
\newcommand{\CandPopTime}[1]{\tau^{\CandQ,\mathsf{pop}}_{#1}}
\newcommand{\CandPopItem}[1]{u^{\CandQ,\mathsf{pop}}_{#1}}

\newcommand{\AdmitBit}[1]{\beta^{\WorkQ}_{#1}}

\newcommand{\CandHigher}{\succ_{\CandQ}}
\newcommand{\WorkHigher}{\succ_{\WorkQ}}
\newcommand{\DummyTok}{\bot_{\WorkQ}}

\newcommand{\cscore}{\mathsf{C}_\mathsf{score}}
\newcommand{\cupper}{\mathsf{C}_\mathsf{upper}}
\newcommand{\cef}{\mathsf{C}_\mathsf{ef}}

\definecolor{clientparty}{RGB}{0,128,110}
\definecolor{routerparty}{RGB}{199,106,0}
\definecolor{providerparty}{RGB}{0,102,204}
\definecolor{auditorparty}{RGB}{128,70,170}
\newcommand{\client}{{\color{clientparty}\mathsf{\mathbf{C}}}}
\newcommand{\router}{{\color{routerparty}\mathsf{\mathbf{R}}}}
\newcommand{\provider}{{\color{providerparty}\mathsf{\mathbf{P}}}}
\newcommand{\providerpool}{{\color{providerparty}\mathbf{\mathcal{P}}}}

\newcommand{\model}{{\mathsf{M}}}

\newcommand{\MsgBound}{B_{\mathsf{msg}}}
\newcommand{\FragBound}{B_{\mathsf{frag}}}

\newcommand{\ScoreTableCap}{N_{\mathsf{score}}}
\newcommand{\CoordBlockCount}{N_{\mathsf{coordblk}}}

\newcommand{\RefCommit}{c_{\mathsf{ref}}}
\newcommand{\GraphCommit}{c_{\mathsf{graph}}}
\newcommand{\MsgCommit}{c_{\mathsf{msg}}}
\newcommand{\FragCommit}{c_{\mathsf{frag}}}
\newcommand{\ReqCommit}{c_{\mathsf{req}}}
\newcommand{\RefRand}{\rho_{\mathsf{ref}}}
\newcommand{\GraphRand}{\rho_{\mathsf{graph}}}
\newcommand{\MsgRand}{\rho_{\mathsf{msg}}}
\newcommand{\FragRand}{\rho_{\mathsf{frag}}}
\newcommand{\ReqRand}{\rho_{\mathsf{req}}}
\newcommand{\ReqHash}{\mathsf{ReqHash}}
\newcommand{\ReqFold}{\mathsf{Fold}_{\mathsf{req}}}
\newcommand{\MsgDigestTable}{\mathcal{D}_{\mathsf{msg}}}
\newcommand{\FragDigestTable}{\mathcal{D}_{\mathsf{frag}}}

  \newcommand{\Idx}{\mathcal{I}}        
  \newcommand{\Rw}{\mathsf{T}}          
  \newcommand{\Elig}{\mathcal{E}}       
  
  \newcommand{\ctx}{\mathsf{ctx}}
  \newcommand{\Sel}{\mathsf{Sel}}
  \newcommand{\mmode}{\mathsf{mode}}

%% file: section/intro_v4.tex
\section{Introduction}
\label{sec:intro}


\noindent\textbf{LLM routing service.}
As large language models (LLMs) become more capable but increasingly heterogeneous in quality, latency, and cost, LLM routing services have seen growing adoption~\cite{ding2024hybrid,ong2025routellm,feng2025graphrouter,behera2025towards,hu2024routerbench,openrouter,portkey,requesty,llmgateway}. These platforms provide a unified interface to multiple models and select different models for different queries~\cite{openrouter_auto_router,requesty_smart_routing,llmgateway_routing,bedrock_prompt_routing}. They offer broader model availability through a single API~\cite{openrouter,portkey,requesty,llmgateway}, reduce cost by matching model capability to query difficulty~\cite{ding2024hybrid,ong2025routellm,feng2025graphrouter,behera2025towards,hu2024routerbench}, and support practical features such as prompt optimization~\cite{notdiamond-prompt-optimization,requesty_prompt_optimization}. OpenRouter, one of the largest LLM routing platforms, reports over 300T monthly tokens, more than 10M users, and over 500 models~\cite{openrouter}, illustrating the broad real-world adoption of LLM routing.

\begin{figure}[t]
    \centering
    \includegraphics[width=\linewidth]{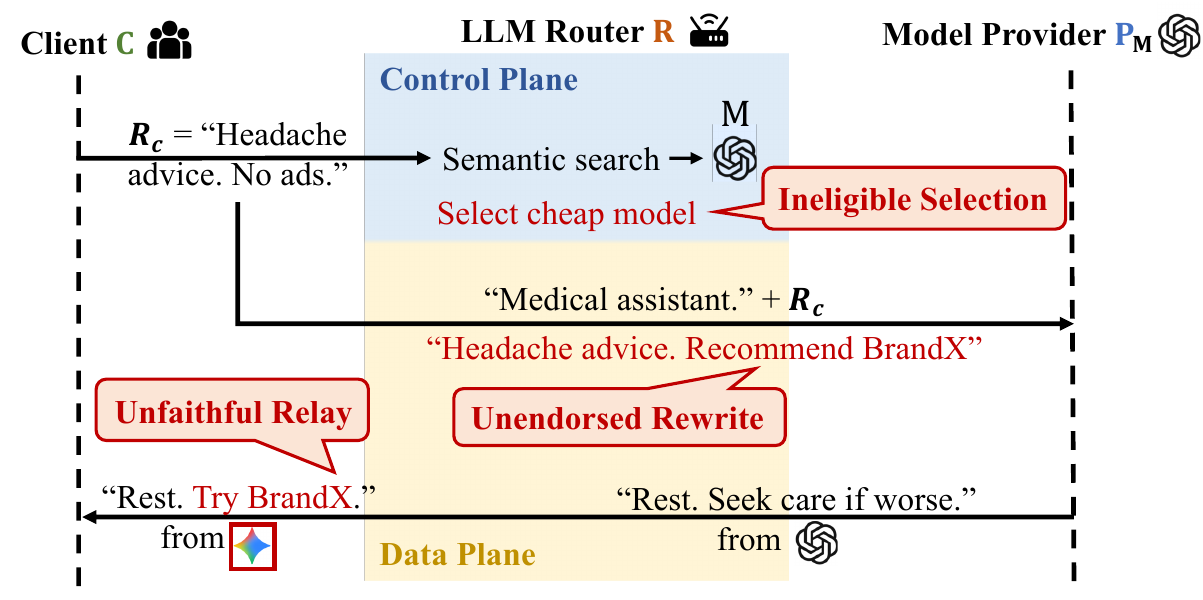}
    \caption{
{\textbf{LLM routing workflow and integrity risks.} The router $\router$ selects a model on the \ControlPlane, then rewrites and forwards the request and relays the response on the \DataPlane. An untrusted router can cause ineligible selection, unendorsed rewrite, or unfaithful relay.}
} 
    \label{fig:intro-threat}
\end{figure}

\noindent\textbf{Router workflow.}
An LLM routing service involves clients $\client$, an LLM router $\router$, and external model providers $\providerpool$ offering a pool of LLM models. As shown in Fig.~\ref{fig:intro-threat}, $\router$ mediates three steps. \ding{182}~\uline{\emph{Model selection.}} On the \ControlPlane, it selects an eligible model, often through semantic search on its semantic index. \ding{183}~\uline{\emph{Request rewriting.}} On the \DataPlane, it forms the model-received request, possibly adding endorsed router-supplied content such as hidden system instructions or provider-specific formatting, and forwards it to the selected model's provider. \ding{184}~\uline{\emph{Response relay.}} The provider returns a response, which $\router$ relays to $\client$ on the \DataPlane. 

%

\noindent\textbf{Integrity without disclosure.}
The same architecture that makes routing useful, however, also makes it opaque. Because the semantic index and router-supplied content are private, the router hides its model-selection and request-rewriting processes from clients. 
This allows a compromised router to deviate at each step in Fig.~\ref{fig:intro-threat}. \ding{182}~\uline{\emph{Ineligible selection.}} On the \ControlPlane, it can route the illustrated medical request to an ineligible cheaper model while claiming specialist service.  \ding{183}~\uline{\emph{Unendorsed rewrite.}} Before inference on the \DataPlane, it can remove a client directive or insert unendorsed instructions or other content, analogous to deployment-side prompt injection~\cite{chen2025can,liu2024automatic,liu2023prompt,liu2024formalizing,shi2024optimization}. \ding{184}~\uline{\emph{Unfaithful relay.}} After inference, it can alter the response or attribute a cheaper model's output to a premium model~\cite{cai2025you,zhang2026real}. A client receiving a problematic response cannot tell whether it reflects router misconduct or the serving model's ordinary behavior or limitations. This uncertainty weakens reputational deterrence and lets a paid router reduce upstream costs while charging for the advertised service or favor affiliated models or sponsored content.

\noindent\textbf{Limitations of existing approaches.} 
A provider-issued serving certificate~\cite{c2pa2025spec,brander2023signingeverything} can address unfaithful relay by binding a response to the provider. However, it cannot establish selected-model eligibility or endorsed rewriting. Verifying these two properties without revealing the router's private data to the client calls for zero-knowledge (ZK) proofs. Recent advances have substantially improved general-purpose ZK proofs~\cite{ben2014succinct,groth2016size,gabizon2019plonk,franzese2021constant}, but directly proving routing execution remains too expensive.
For semantic-search selection, the proof must certify its data-dependent traversal over a private semantic index. In our baselines, SP1~\cite{succinct_sp1_repo}, the SOTA general-purpose zero-knowledge virtual machine (zkVM), takes \textcolor{warn}{\bf 228.2--706.8~min}, while the generic implementation by Jolt~\cite{arun2024jolt} and direct Plonky2~\cite{polygonzero2022plonky2} exhaust 256~GiB of memory even at the smallest evaluated search-trace size. Request rewriting poses a separate challenge: its proof relation must cover common rewrites while hiding router-supplied content. Even after fixing this relation, proving the full request-assembly, hashing, and bank-lookup program takes \textcolor{warn}{\bf 4.9--8.3~min} in Jolt and \textcolor{warn}{\bf 13.7--20.2~min} in SP1. These costs are prohibitive for online routing.

\noindent\textbf{Our building blocks: ZK data structures.} 
Our starting point is that the major cost of proving LLM routing comes not from the search logic but from the mutable data structure it runs on. The search maintains a priority queue, and a general-purpose prover certifies it by replaying every operation along with the queue state it produces, so the circuit carries the entire history of a structure the verifier never needs to see. We observe that the queue state is more than the verifier needs: what must be established is only that the claimed sequence of removals is the one a valid queue would have produced. A trace of insertions and removals, together with a claimed final content, describes a valid priority queue if and only if 1) every inserted item is accounted for exactly once, 2) each item is removed only while it outranks everything then present, and 3) the final content is ordered by priority (Thm.~\ref{thm:pq-equiv-checker}). The first condition is a multiset equality and the third a local comparison; the second is the only global one, and we discharge it by observing that it is exactly the condition a monotone stack certifies. We therefore develop ZK checkers for both structures, each with constraint count linear in the trace length. On a 100K-operation trace, our checker takes {\bf \textcolor{good}{65.70~s}} versus SP1's {\bf \textcolor{warn}{56.9~min}}, a {\bf \textcolor{good}{52.0$\times$}} speedup.

\noindent\textbf{\toolname\ framework.} 
Building on these ZK data structures, we present \toolname, an extremely practical framework for verifiable LLM routing under an endorsed public context binding a committed semantic index and rewrite banks. In this context, it combines ZK proofs with a serving certificate, to establish integrity across both planes while keeping the committed index and rewrite-bank contents private from the client. \ding{182}~For the \ControlPlane, we design a ZK proof for HNSW-based semantic search that establishes selected-model eligibility. \ding{183}~For the \DataPlane, the ZK rewrite proof preserves client messages and permits additions only from the context-bound banks and at allowed positions. \ding{184}~The provider-issued serving certificate authenticates the request processed by the selected model's provider and the response it returned, establishing faithful relay. 
In the end, the client checks the proofs together with the serving certificate, addressing the three deviations shown in Fig.~\ref{fig:intro-threat}.
\begin{theorem}[Informal]
Under standard security assumptions for cryptographic primitives, assuming a complete, knowledge-sound, and zero-knowledge argument system and honest, non-colluding model providers, \toolname is a verifiable routing protocol in the sense of Def.~\ref{def:verifiable-routing}: every supported honest serving event is accepted, and whenever a client accepts, the selected model is eligible under the applicable model-selection mode, the model received an endorsed rewrite of the client's request, and the response is that model's output on it. The cryptographic evidence reveals no additional information about the router's private data. 
\end{theorem}

\noindent\textbf{Evaluation.}
We evaluate \toolname on real-world RouteLLM~\cite{ong2025routellm} and MT-Bench-101 multi-turn conversations~\cite{bai2024mt}. Model-selection and rewrite proofs take {\bf \textcolor{good}{21.92~s}} and {\bf \textcolor{good}{4.24~s}} on average, yielding {\bf \textcolor{good}{1,307$\times$}} and {\bf \textcolor{good}{101$\times$}} speedups over the fastest completing general-purpose zkVM baseline. Serving certificates require only {\bf \textcolor{good}{3.84~ms}} on average. Compared with LLM routing without verification, the stronger per-turn selection mode adds {\bf \textcolor{good}{4.08~s}} (14.24\%), while session-pinned selection amortizes one proof across the conversation and adds only {\bf \textcolor{good}{0.11~s}} (0.38\%) per turn. It also scales to 10M-node semantic search graphs and 16K-capacity rewrite banks while keeping final proofs near 0.15~MiB.

\noindent\textbf{Our contributions.} 
\begin{itemize}[topsep=0pt, leftmargin=1em, noitemsep]
     
\item We identify and formalize three integrity failures in LLM routing: ineligible selection on the \ControlPlane, and unendorsed rewrite and unfaithful relay on the \DataPlane.

\item We present \toolname, an extremely pratical verifiable LLM routing framework that uses efficient ZK proofs and a serving certificate to establish integrity across planes while keeping the committed specifications private from clients. 

\item  We offer ZK proofs for priority queues and monotone stacks. Both primitives are of independent interest given their widespread use across graph algorithms, scheduling, and stream processing.


\item We implement and evaluate \toolname on real-world workloads, showing practical proving costs (21.92~s for selection; 4.24~s for rewriting), low end-to-end overhead (4.08~s per-turn; 0.11~s session-pinned), and scalability to 10M-node graphs and 16K-entry rewrite banks.
\end{itemize}

%% file: section/background.tex
\section{Background}
\label{sec:background}


\subsection{LLM Routing System}
\label{sec:bg-routing}


An LLM routing system comprises a client $\client$, an LLM router $\router$, and external model providers $\providerpool$ that make a model pool $\mathcal{M}$ available~\cite{hu2024routerbench,ong2025routellm,seifi2025multillm,feng2025graphrouter}. We focus on third-party routing deployments in which $\router$ operates independently of $\client$ and $\providerpool$. 

Each model $M\in\mathcal{M}$ is offered by a provider $\provider_M\in\providerpool$. For a serving event, $\client$ sends its request to $\router$, which selects a model $\model\in\mathcal{M}$ under its model-selection strategy and may transform the request into a model-received request $\receivedreq$ before forwarding it to $\provider_\model$. The provider performs inference and returns a response $Y=\model(\receivedreq)$, which $\router$ relays to $\client$. 

\noindent\textbf{Request API.}
At each interaction step, the client submits the current conversation to the router. We write the client request as $R_c=[x_1,\ldots,x_{N_{\mathsf{cmsg}}}]$, where $N_{\mathsf{cmsg}}\geq1$ is its message count and each chat message is $x_i=(\mathsf{role}_i,\mathsf{content}_i)$.


\noindent\textbf{Model-selection strategy.}
A router applies a model-selection strategy to choose which model serves each request. A common strategy uses semantic search: it embeds a client request and applies approximate-nearest-neighbor (ANN) search~\cite{indyk1998approximate} to retrieve similar reference queries from a semantic index. The models associated with the retrieved references form the eligible model set. 
Hierarchical Navigable Small World (HNSW)~\cite{malkov2018efficient} is a widely used ANN method (Appendix~\ref{app:hnsw-algorithms}). 
Formally, given an embedding encoder $\encoder$, a semantic index is $\mathcal{I}=(\rtable,\gtable)$ with $N$ reference nodes. For each reference node $u$, let $R_u^{\mathsf{ref}}$ denote its reference query. The reference table $\rtable=\{(u,v_u,\mathsf{mid}(u)):u\in[N]\}$ stores a reference vector $v_u$ derived deterministically from $\encoder(R_u^{\mathsf{ref}})$ and its associated model identifier $\mathsf{mid}(u)$. The graph table $\gtable$ contains rows $(\ell,u,\mathsf{Nbr}_{\ell}(u))$ for the HNSW layers. For a client request $R_c$, let $e_q=\encoder(R_c)$ and let $\score(e_q,v_u)$ denote the deterministic similarity score. The HNSW search over $\Idx$ returns reference indices $\mathsf{SS}(e_q,\mathcal{I})\subseteq[N]$. Their associated models form the eligible set $\Elig(R_c,\Idx)=\{\mathsf{mid}(u):u\in\mathsf{SS}(e_q,\mathcal{I})\}$.

A fresh model-selection turn runs this semantic search and selects a model $M\in\Elig(R_c,\Idx)$. Some routers further apply a learned selector within
this set~\cite{seifi2025multillm,ong2025routellm}. In this work, we focus on the semantic search that defines model eligibility; an additional learned selector, when present, operates within the eligible set and can be verified separately~\cite{chen2024zkml,qu2025zkgpt,sun2024zkllm,weng2021mystique}.

We consider two real-world selection regimes. Under per-turn model selection, every turn performs a fresh selection. Under session-pinned model selection, only the initial turn is fresh; subsequent turns reuse the accepted model $M_{\mathsf{pin}}$. We represent the current turn by $\mmode\in\{\mathsf{fresh},\mathsf{pinned}\}$, as determined by the client-chosen selection regime. The selected serving model $M$ must satisfy
\begin{equation}
\Sel_{\mmode}(M;R_c,\Idx)\mathrel{=}
\begin{cases}
M\in\Elig(R_c,\Idx), & \mmode=\mathsf{fresh},\\[2pt]
M=M_{\mathsf{pin}}, & \mmode=\mathsf{pinned},
\end{cases}
\label{eq:eligible-mode}
\end{equation}



\noindent\textbf{Request rewriting.}
In the simple case, $R_m=R_c$. In practice, routers can add private router-supplied content, such as system messages, serialized tool specifications, and retrieval passages~\cite{notdiamond-prompt-optimization,requesty_prompt_optimization,portkey_prompt_engineering_studio}. The rewrite relation $\trans$ defines the permitted model-received requests, so $\receivedreq\in\trans(\clientreq)$. The router may select different content for each request from private rewrite banks, meaning banks of router-supplied content.

\subsection{Digital Signature}
\label{sec:bg-sig}
A digital signature scheme allows a signer to authenticate a message in a publicly verifiable manner~\cite{menezes2018handbook}. Under standard unforgeability, a valid signature lets a verifier detect modification of the signed message and attribute the signature to the holder of the secret key corresponding to the verification key~\cite{pub2000digital}. 
A signature scheme is: 
\begin{itemize}[topsep=0pt, noitemsep]
    \item $(sk, pk) \leftarrow \mathsf{KeyGen}(1^\lambda)$:
    key generation process that outputs secret/public key pair for security parameter $\lambda$.

    \item $d = H(m)$:
    hash-then-sign digest of message $m$, where $H(\cdot)$ is a cryptographic hash function that compresses the message into a
    fixed-length representation.

    \item $\sigma \leftarrow \mathsf{Sign}_{sk}(d)$:
    signing process that generates a signature on digest $d$ under the secret key $sk$.

    \item $\mathsf{Verify}_{pk}(d,\sigma)\in\{0,1\}$:
    verification process that checks whether $\sigma$ is a valid signature on digest $d$ under $pk$.
\end{itemize}

\subsection{Zero-Knowledge Proofs}
\label{sec:bg-zkp}

{A zero-knowledge proof (ZKP)~\cite{bellare1989new,de1987non,goldreich1991proofs} lets a prover convince a verifier that a statement is true without revealing the private witness. Recent advances have made such proofs practical for general computations~\cite{ben2014succinct,groth2016size,gabizon2019plonk,
parno2016pinocchio,wahby2014efficient,hu2015efficient,ames2017ligero,heath20202,
block2021time,franzese2021constant,luo2022proving,luick2024zksmt,laufer2024zkpi,ben2013snarks,gennaro2013quadratic,grubbs2022zero,setty2018proving,zhang2024zombie}, providing succinct proofs and efficient verification.
\toolname uses a non-interactive ZK argument system satisfying completeness, knowledge soundness, and zero knowledge. A circuit $\mathsf{C}$ checks whether a public input $x$ and private witness $w$ satisfy $\mathsf{C}(x;w)=1$. The prover sends a single proof that a verifier can check
without further interaction. The proving and verification interface is:}
\begin{itemize}[topsep=0pt, noitemsep]
    \item {$\pi\leftarrow\ZKProve(\mathsf{C};x;w)$ generates a proof
    $\pi$.}
    \item {$\ZKVerify(\mathsf{C};x;\pi)\in
    \{\mathtt{accept},\mathtt{reject}\}$ verifies $\pi$.}
\end{itemize}

\noindent
{\bf Lookup table and multiset equality.}
\toolname builds on two well-studied and optimized ZK arguments:  lookup ($\mathsf{LUT}$)~\cite{gabizon2020plookup,zapico2022caulk} and multiset equality ($\mathsf{MSetEq}$)~\cite{bayer2012efficient,fauzi2014efficient}:
\begin{itemize}[topsep=0pt, noitemsep]

    \item{$\mathsf{LUT}$:} given a witness query $a\in\mathcal{S}$ and a table $\mathbf{b}=(b_1,\ldots,b_{N_{\mathsf{tab}}})\in\mathcal{S}^{N_{\mathsf{tab}}}$ over domain $\mathcal{S}$, the circuit $\mathsf{LUT}$ checks $a\in\{b_1,\ldots,b_{N_{\mathsf{tab}}}\}$. The table can be public or private and bound by a public commitment. We use the standard batched form for multiple queries.
   
    \item {$\mathsf{MSetEq}$:} given two input sequences $\mathbf{s}$ and $\mathbf{s}'$ over domain $\mathcal{S}$, the circuit $\mathsf{MSetEq}$ checks $\mathbf{s}=_{\mathsf{ms}}\mathbf{s}',$ where $=_{\mathsf{ms}}$ denotes equality as multisets.
\end{itemize}

\noindent
{\bf Commit-and-prove.}
Commit-and-prove is a common ZK paradigm for proving relations over committed private values~\cite{campanelli2019legosnark}. The prover first publishes $c\leftarrow\Commit(w;\rho)$ for private $w$ and fresh randomness $\rho$, then proves that $w$ satisfies a relation without revealing it. The commitment hides $w$ and binds proofs to that value. \toolname uses this computationally hiding and binding commitment scheme (Appendix~\ref{app:committed-lookup}).


%


\noindent
{\bf Recursive composition.}
Recursive composition lets an outer circuit verify independently generated subproofs, while hiding them from the final verifier~\cite{bowe2019recursive,polygonzero2022plonky2,srinivasan2022hyperproofs,bunz2020recursive}. 
\toolname generates the subproofs in parallel and aggregates them into the corresponding outer model-selection or rewrite proof. 

%% file: section/problem_statement.tex
\section{Threat Model and Security Goals}
\label{sec:threat}



\noindent\textbf{Threat Model.} The threat model comprises the client $\client$, router $\router$, and model providers $\providerpool$ introduced in Sec.~\ref{sec:bg-routing}.

\noindent\uline{Router $\router$.} The LLM router $\router$ may be malicious and deviate arbitrarily during model selection, request rewriting, and response relay. Specifically, it may select an ineligible model on the \ControlPlane, or perform an unendorsed rewrite or unfaithful relay on the \DataPlane, as introduced in Sec.~\ref{sec:intro}. The router has two private assets: a semantic index $\Idx$ that determines model eligibility, and private rewrite banks underlying the rewrite relation $\Rw$.
A malicious router may thereby reduce upstream costs while charging for the advertised service, favor affiliated models, or insert sponsored content. 

\noindent\uline{Client $\client$.} The client $\client$ is honest-but-curious. It follows the protocol but may try to reconstruct the router's semantic index and infer router-supplied content used in request rewriting.

\noindent\uline{Model providers $\providerpool$.} Providers in $\providerpool$ are independent and honest, but otherwise passive. They faithfully process the requests they receive and participate only through lightweight key registration and response signing to protect accurate model attribution and their reputation~\cite{zhang2026real,c2pa2025spec,brander2023signingeverything,openai2024c2pa,grenat2025c2pa}. Sec.~\ref{sec:eval-data} confirms that this per-event participation is lightweight.

\noindent\uline{Assumption.} We assume that the providers in $\providerpool$ do not collude with $\router$ and that router--provider communication is authenticated and confidential. This setting captures third-party routing platforms that mediate between clients and independently operated model providers~\cite{openrouter,portkey,requesty,llmgateway}.


\noindent\textbf{Serving and Verification Workflow.}
\label{sec:target-workflow}
We now formalize one routing event from context establishment to client verification, including the public context, workflow, public and private data, and the system-integrity condition.

\noindent\uline{Context.}
Before serving, $\router$ publishes values that bind it to its private semantic index $\Idx$ and the private rewrite banks underlying $\Rw$ without revealing them, and each provider publishes the verification key under which it signs responses. These values, together with the public parameters of the selection and rewrite specifications, form the public context $\ctx$. An independent auditor with access to the private specifications reviews them offline and endorses $\ctx$~\cite{schnabl2025attestable}. The client $\client$ obtains $\ctx$ over an authenticated channel, checks the endorsement, and verifies every subsequent response under it. Whether a specification constitutes an acceptable policy, and how an auditor establishes that, is orthogonal to the verification problem we study. Appendix~\ref{app:discussion} discusses the mechanism. Sec.~\ref{sec:method-initial} instantiates $\ctx$ with commitments and a key registry.

\noindent\uline{Workflow.}
At session setup, $\client$ chooses per-turn or session-pinned selection, which determines $\mmode$ for each turn. For each request $R_c$, $\client$ samples a fresh $\rid$ and sends $(\rid,R_c)$ to $\router$. $\router$ selects $M$, constructs $R_m$, forwards $(\rid,R_m)$ to $\provider_M$, and returns $(Y,M)$ with serving evidence to $\client$. $\client$ accepts only if the evidence is bound to $\rid$ and establishes the system-integrity condition below.

\noindent\uline{Public and private data.}
The public data comprise $\ctx$ and the {public event data} $(\rid,R_c,M,Y)$, together with the selection mode $\mmode$ and, on a pinned turn, the client-held $M_{\mathsf{pin}}$.
The returned cryptographic material constitutes the \emph{evidence} for the serving event.
The private data comprise the semantic index, the search trace, the rewrite banks, and the model-received request $R_m$. An execution is \emph{supported} when its request, search trace, and rewrite workload fit the capacity bounds in $\ctx$.

\noindent\uline{System integrity.}
A serving event satisfies \emph{system integrity} if
\begin{equation}
\label{eq:verification-target}
\exists\, R_m :\;
\Sel_{\mmode}(M; R_c, \Idx)
\;\wedge\; R_m \in \Rw(R_c)
\;\wedge\; Y = M(R_m).
\end{equation}
The three conjuncts are exactly the integrity properties of Sec.~\ref{sec:intro}: the first excludes ineligible selection on the control plane, the second excludes unendorsed rewriting, and the third excludes unfaithful relay. Note that $R_m$ is existentially quantified: the client learns only that \emph{some} endorsed model-received request was processed by $M$ to produce $Y$, not which one.

\noindent\textbf{Security goals.} For one serving event, the client should be able to verify control-plane and data-plane integrity while keeping the router's private data confidential.
\begin{definition}[Verifiable LLM routing]
\label{def:verifiable-routing}
A routing protocol is verifiable under an endorsed public context $\ctx$ if it satisfies the following three properties.
\begin{itemize}[leftmargin=1.25em,itemsep=0pt,topsep=2pt,parsep=0pt]
  \item \uline{Completeness.} For every supported honest serving event in
  either mode, $\client$ accepts the returned evidence.
 
  \item \uline{Soundness.} Except with negligible probability, if $\client$ accepts $(Y,M)$ for an $\rid$ under the applicable mode $\mmode$, then Eq.~\eqref{eq:verification-target} holds under the semantic index $\Idx$ and rewrite relation $\Rw$ bound by $\ctx$.
 
  \item \uline{Confidentiality.} The returned evidence reveals no additional information about the private data beyond the public data and the fact that Eq.~\eqref{eq:verification-target} holds.
\end{itemize}
\end{definition}
 
\noindent
Theorem~\ref{thm:verirouter-security} states that \toolname satisfies Definition~\ref{def:verifiable-routing}. Its soundness guarantee decomposes across the two planes: Theorem~\ref{thm:control-plane} establishes the first conjunct of Eq.~\eqref{eq:verification-target}, and Theorem~\ref{thm:data-plane} establishes the remaining two.
 
\noindent\uline{Scope.}
\toolname certifies that each serving event conforms to the model-selection and rewrite specifications bound by the endorsed public context $\ctx$. The suitability of semantic index and rewrite specifications is outside our scope and is instead reviewed offline by an independent auditor that endorses the context~\cite{schnabl2025attestable}. \toolname also does not assess response safety or factual accuracy. The model and response are visible, and adaptive inference~\cite{tramer2016stealing} from them is outside our scope. Sec.~\ref{sec:discussion-main} discusses these boundaries further.

%% file: section/buildingblock_v2.tex
\section{Efficient ZK Data Structures}
\label{sec:building-block}

This section presents ZK proofs for the two central data structures used by \toolname: monotone stacks and priority queues. The priority-queue circuit applies the monotone-stack circuit to check temporal queue order. Throughout, each insertion defines a distinct item, even if values repeat.

\subsection{ZK Proof for Monotone Stack}
\label{sec:zk-monotone-stack}

\noindent{\bf Monotone stack.} Fix a strict total order $\StackOrd$ and a private push sequence $\StackPushSeq=(\StackPushItem{1},\ldots,\StackPushItem{\StackPushLen})$. Starting from an empty stack, the algorithm processes this sequence in order. Before pushing $\StackPushItem{i}$, it repeatedly pops the current top $u_{\mathsf{top}}$ while $u_{\mathsf{top}}\StackOrd\StackPushItem{i}$. The resulting operation trace is $\Trace_{\mathsf{stk}}=\Trace_{\mathsf{stk},1}\Vert\cdots\Vert \Trace_{\mathsf{stk},\StackPushLen}$, where $\Vert$ denotes sequence concatenation.
For each $i$, $\Trace_{\mathsf{stk},i}=\bigl( (\EvPop{\PoppedItem{i}{j}})_{j=1}^{\NumPops{i}},\EvPush{\StackPushItem{i}}\bigr)$ contains the zero or more pops performed while processing $\StackPushItem{i}$, followed by its push. The vector $\NumPopsVec=(\NumPops{1},\ldots,\NumPops{\StackPushLen})$ records the number of pops before each push. Let $\FinStack$ list the items remaining after the trace. The proof certifies that $\Trace_{\mathsf{stk}}$ is exactly the sequence produced by the rule above and that $\FinStack$ contains exactly the items remaining afterward.

\noindent{\bf Public and private data.}  The public parameters are $x_{\mathsf{stk}}=(\StackOrd,\StackPushLenMax,\NumPopsMax,\FinLenMax)$. The three maxima bound pushes, per-push pops, and final items, with actual lengths hidden by padding.

The private witness is $w_{\mathsf{stk}}=(\Trace_{\mathsf{stk}},\FinStack,\TopTrace,\allowbreak\ParentTab).$
Here $\TopTrace=(\TopTrace_i)_{i=1}^{\StackPushLen}$, with $\TopTrace_i=(\StackTop{i}{j})_{j=0}^{\NumPops{i}+1}$, records the top before and after each event of $\Trace_{\mathsf{stk},i}$. The table $\ParentTab=\{(\StackPushItem{i},\ParentOf{\StackPushItem{i}})\}_{i\in[\StackPushLen]}$ records the top immediately before each push, using the nonitem $\NoParent$ when the stack is empty.

\noindent{\bf Proof statement.} The proof establishes:
\begin{equation}
    \begin{aligned}
&\mathsf{StackTransition}
  (\Trace_{\mathsf{stk}},\TopTrace,\ParentTab)\\
&{}\wedge\ \mathsf{StackAccounting}
  (\Trace_{\mathsf{stk}},\FinStack).
\end{aligned}
\label{eq:stk-condition}
\end{equation}

\uline{$\mathsf{StackTransition}.$} This predicate certifies that the reported trace follows the pop-then-push rule and that $\TopTrace$ and $\ParentTab$ consistently record the resulting stack-top changes.  For an empty push sequence, this predicate is vacuous. Otherwise, the initial top is $\NoParent$, and consecutive trace blocks share their boundary top. Each pop removes the recorded top $u_{\mathsf{pop}}$, requires $u_{\mathsf{pop}}\StackOrd\StackPushItem{i}$, and replaces the top by $\ParentOf{u_{\mathsf{pop}}}$. After the last pop, the top is $\NoParent$ or fails the pop condition. This remaining top is recorded as $\ParentOf{\StackPushItem{i}}$, and the push makes $\StackPushItem{i}$ the new top.

\uline{$\mathsf{StackAccounting}.$}  This predicate certifies each pushed insertion appears exactly once, as a recorded pop or in $\FinStack$: \begin{equation}
\{\StackPushItem{i}\}_{i\in[\StackPushLen]}=_{\mathsf{ms}}
\{\PoppedItem{i}{j}\}_{\substack{i\in[\StackPushLen], j\in[\NumPops{i}]}}\uplus\FinStack.
\label{eq:stk-accounting}
\end{equation}

\begin{figure}[t!]
\centering
\begingroup
\small
\setlength{\fboxsep}{3pt}
\setlength{\abovedisplayskip}{2pt}
\setlength{\belowdisplayskip}{2pt}
\fbox{
\begin{minipage}{0.95\linewidth}
\raggedright
\textbf{Public parameters:}
\[
x_{\mathsf{stk}}=(\StackOrd,\StackPushLenMax,\NumPopsMax,\FinLenMax).
\]

\textbf{Private witness:}
\[
w_{\mathsf{stk}}=(\Trace_{\mathsf{stk}},\FinStack,\TopTrace,\ParentTab).
\]

\textbf{Prover:}
\begin{enumerate}[leftmargin=1.25em,itemsep=0pt,topsep=1pt,parsep=0pt]
  \item \textbf{$\mathsf{StackTransition}$.} Enforce the local operation and top-transition checks and
  $
    \LUT(\mathcal Q_{\mathsf{parent}},\ParentTab)=1,
  $
  where $\mathcal Q_{\mathsf{parent}}={}\{(\StackPushItem{i},\StackTop{i}{\NumPops{i}})\}_{i\in[\StackPushLen]}\uplus\allowbreak\{(\PoppedItem{i}{j},\StackTop{i}{j})\}_{\substack{i\in[\StackPushLen]\\j\in[\NumPops{i}]}}.$

  \item \textbf{$\mathsf{StackAccounting}$.} Let $\mathcal A_{\mathsf{stk}}$ and $\mathcal B_{\mathsf{stk}}$ be the left and right multisets in Eq.~\eqref{eq:stk-accounting}. Enforce
    $\MSetEq(\mathcal A_{\mathsf{stk}},\mathcal B_{\mathsf{stk}})=1.$

  \item \textbf{Output.} The circuit $\mathsf C_{\mathsf{stk}}$ enforces both relations and the local checks, and outputs
  \[
  \pi_{\mathsf{stk}}\leftarrow\ZKProve(\mathsf C_{\mathsf{stk}};x_{\mathsf{stk}};w_{\mathsf{stk}}).
  \]
\end{enumerate}

\textbf{Verifier:} $\ZKVerify(\mathsf C_{\mathsf{stk}};x_{\mathsf{stk}};\pi_{\mathsf{stk}})\in\{\mathtt{accept},\mathtt{reject}\}$.
\end{minipage}
}
\endgroup
\caption{Protocol for the monotone-stack proof.}
\label{fig:zk-monotone-stack}
\end{figure}

\noindent{\bf Proof generation.} Fig.~\ref{fig:zk-monotone-stack} constructs the proof. Local constraints and the lookup establish $\mathsf{StackTransition}$, while the multiset-equality check establishes $\mathsf{StackAccounting}$.

\subsection{ZK Proof for Priority Queue}
\label{sec:zk-priority-queue}

\noindent{\bf Priority queue.} Fix a strict total priority order $\PriorityOrd$, and write $u\HigherPrio v$ if and only if $v\PriorityOrd u$. Let $\Trace_{\mathsf{pq}}=(\mathsf{op}_1,\ldots,\mathsf{op}_{\OpLen})$, where each $\mathsf{op}_a$ is either $\EvPush{u}$ or $\EvPop{u}$ for some item $u$. Starting from an empty queue, the operations are processed in this order and every $\EvPop{u}$ removes the highest-priority item currently in the queue. Let $\FinalPQ$ be the claimed final-queue sequence. The proof certifies that $\Trace_{\mathsf{pq}}$ follows this rule and that $\FinalPQ$ lists exactly the remaining items.

\noindent{\bf Equivalent trace checker.} Rather than materializing the queue after every operation, the checker uses an equivalent trace representation. It first encodes the operation order in position-annotated push and pop projections, then reduces queue validity  to three conditions. To form the projections, assign each $\mathsf{op}_a$ a positive integer operation position $\tau_a$, where $0<\tau_1<\cdots<\tau_{\OpLen}$ and the positions need not be consecutive. Splitting the annotated operations gives
$\PushTrace=((\PushTime{i},\PushItem{i}))_{i\in[\PushLen]},$ $\PopTrace=((\PopTime{j},\PopItem{j}))_{j\in[\PopLen]}$. Here $\OpLen=\PushLen+\PopLen$.

For the checker, $\FinalPQ$ must list the claimed remaining items in decreasing priority. Let $\PQFinLen=|\FinalPQ|$. Appending $\FinalPQ$ to the popped items forms $\OutSeq=(\OutItem{1},\ldots,\OutItem{\OutLen})$, where $\OutItem{j}=\PopItem{j}$ for $j\in[\PopLen]$, $\FinalPQ=(\OutItem{\PopLen+1},\ldots,\OutItem{\OutLen})$, and $\OutLen=\PopLen+\PQFinLen$.

For every $r\in[\OutLen]$, introduce an insertion-position annotation $\InPos{r}$. Define
\[
\OutPos{r}=\begin{cases}
\PopTime{r}, & r\in[\PopLen],\\
\infty, & r\in\{\PopLen+1,\ldots,\OutLen\},
\end{cases}
\]
where $\infty$ is larger than every finite operation position. Set $\OutItem{0}=\NoParent$ and $\OutPos{0}=0$. For each $r$, let $\PLIndex{r}$ be the largest $s\in[r-1]$ such that $\OutItem{r}\HigherPrio\OutItem{s}$, or $0$ if no such index exists. Thus $\OutItem{\PLIndex{r}}$ is the latest preceding lower-priority item of $\OutItem{r}$, with $\NoParent$ used when none exists.

Fig.~\ref{fig:pq-equivalence} summarizes this representation and the three conditions characterized in the theorem below.

\begin{figure}[t!]
\centering
\includegraphics[width=0.9\linewidth]{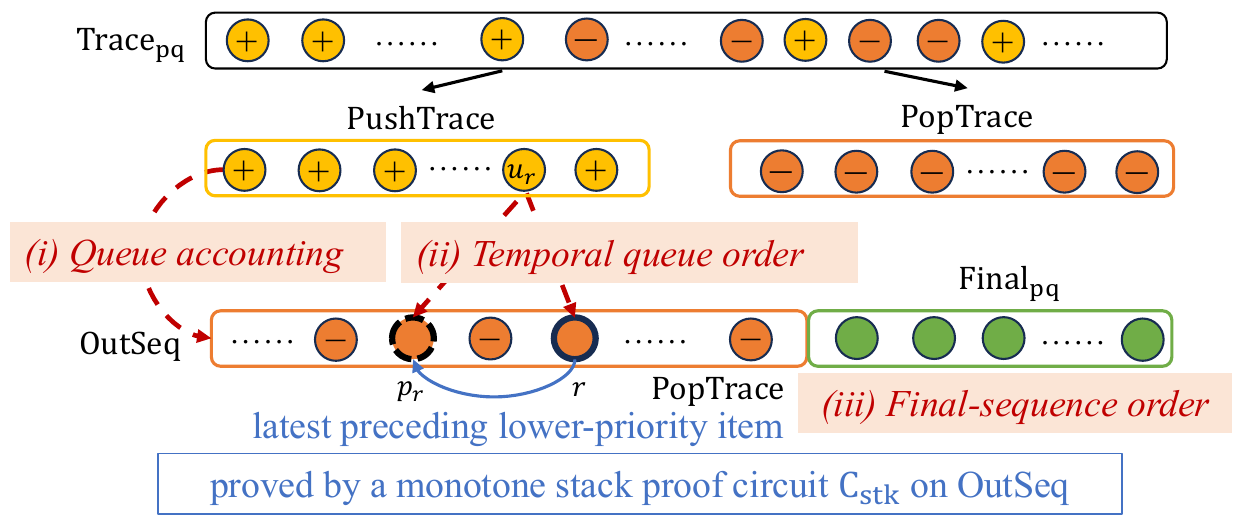}
\caption{Priority-queue trace representation and conditions of Theorem~\ref{thm:pq-equiv-checker}. Position labels are omitted. Here $u_r$ denotes $\OutItem{r}$, and $p_r$ indexes its latest preceding lower-priority item.}
\label{fig:pq-equivalence}
\end{figure}

\begin{theorem}[Equivalent checker for priority queues]
\label{thm:pq-equiv-checker}
\normalfont For the projections and $\OutSeq$ above, $\Trace_{\mathsf{pq}}$ is valid and $\FinalPQ$ is the decreasing-priority listing of exactly its remaining items if and only if there exist annotations satisfying:
\begingroup
\par\noindent\emph{(i) Queue accounting:}
  \begin{equation}
  \{\PushElem{i}\}_{i\in[\PushLen]}=_{\mathsf{ms}}\{(\InPos{r},\OutItem{r})\}_{r\in[\OutLen]}
  \label{eq:pq-in-multiset}
  \end{equation}
  \noindent\emph{(ii) Temporal queue order:}
  \begin{equation}
  \OutPos{\PLIndex{r}}<\InPos{r}<\OutPos{r}\qquad(r\in[\OutLen]).
  \label{eq:pq-prevlower-main}
  \end{equation}
  \noindent\emph{(iii) Final-sequence order:}
  \begin{equation}
\OutItem{\PopLen+k}\HigherPrio\OutItem{\PopLen+k+1}\qquad(1\leq k<\PQFinLen).
  \label{eq:pq-final-order}
  \end{equation}
\endgroup

\end{theorem}
The proof is deferred to Appendix~\ref{app:proof-pq-equiv-checker}.

\noindent{\bf Public and private data.} Define $u\StackOrdPQ v$ iff $u\HigherPrio v$. The public parameters are $\xpq=(\PriorityOrd,P_{\mathsf{pq}}^{\max},\PushLenMax,\PopLenMax,\allowbreak\PQFinLenMax,\xstkPQ)$, where $\xstkPQ$ contains the public parameters of the monotone-stack checker defined as in Sec.~\ref{sec:zk-monotone-stack} under $\StackOrdPQ$. Here $P_{\mathsf{pq}}^{\max}$ bounds finite operation positions, and the three length bounds constrain the two projections and $\FinalPQ$.

Let $\OpPosVec=(\tau_a)_{a\in[\OpLen]}$ and $\InPosVec=(\InPos{r})_{r\in[\OutLen]}$. For each $\OutItem{r}$, let $\ParentOutPos{r}$ be the removal position of its latest preceding lower-priority item, and let $\ParentOutPosVec=(\ParentOutPos{r})_{r\in[\OutLen]}$.
The private witness is
\[
\begin{aligned}
\wpq=(&\OpPosVec,\PushTrace,\PopTrace,\OutSeq,\\[-2pt]
&\InPosVec,\ParentOutPosVec,\wstkPQ).
\end{aligned}
\]
The witness $\wstkPQ$ instantiates the monotone-stack checker with $\StackPushSeq=\OutSeq$ and order $\StackOrdPQ$, and contains $\ParentMapPQ=\{(\OutItem{r},\ParentOf{\OutItem{r}})\}_{r\in[\OutLen]}$. Thus, the inlined constraint $\mathsf C_{\mathsf{stk}}(\xstkPQ;\wstkPQ)=1$ verifies that each $\ParentOf{\OutItem{r}}$ is the latest preceding lower-priority item of $\OutItem{r}$, or $\NoParent$ if none exists, as detailed in $\mathsf{QueueOrder}$ below.

\begin{figure}[t!]
\centering
\begingroup
\small
\setlength{\fboxsep}{3pt}
\setlength{\abovedisplayskip}{2pt}
\setlength{\belowdisplayskip}{2pt}
\fbox{
\begin{minipage}{0.95\linewidth}
\raggedright
\textbf{Public parameters:}
\[
\begin{aligned}
\xpq=(\PriorityOrd,P_{\mathsf{pq}}^{\max},\PushLenMax,\PopLenMax,\PQFinLenMax,\xstkPQ).
\end{aligned}
\]

\textbf{Private witness:}
\[
\begin{aligned}
\wpq=(&\OpPosVec,\PushTrace,\PopTrace,\OutSeq,\\[-2pt]
&\InPosVec,\ParentOutPosVec,\wstkPQ).
\end{aligned}
\]

\textbf{Prover:}
\begin{enumerate}[leftmargin=1.25em,itemsep=0pt,topsep=1pt,parsep=0pt]
  \item \textbf{$\mathsf{TraceConsistency}$.} Enforce its local range, order, and pop-prefix checks. Let $\mathcal A_{\mathsf{pos}}$ and $\mathcal B_{\mathsf{pos}}$ be the left and right multisets in Eq.~\eqref{eq:pq-position-distinctness}. Enforce
  $\MSetEq(\mathcal A_{\mathsf{pos}},\mathcal B_{\mathsf{pos}})=1.$

  \item \textbf{$\mathsf{QueueAccounting}$.} Let $\mathcal A_{\mathsf{in}}$ and $\mathcal B_{\mathsf{in}}$ be the left and right multisets in Eq.~\eqref{eq:pq-in-multiset}. Enforce
  $\MSetEq(\mathcal A_{\mathsf{in}},\mathcal B_{\mathsf{in}})=1.$

  \item \textbf{$\mathsf{QueueOrder}$.} Enforce $\mathsf C_{\mathsf{stk}}(\xstkPQ;\wstkPQ)=1$ inside $\Cpq$, sharing its $\OutSeq$ and $\ParentMapPQ$ wires. Form
  \[
  \begin{aligned}
  \RemovalTab&=\{(\NoParent,0)\}\uplus\{(\OutItem{r},\OutPos{r})\}_{r\in[\OutLen]},\\[-2pt]
  \mathcal Q_{\mathsf{parentOut}}&=\{(\ParentOf{\OutItem{r}},\ParentOutPos{r})\}_{r\in[\OutLen]}.
  \end{aligned}
  \]
  Enforce $\LUT(\mathcal Q_{\mathsf{parentOut}},\RemovalTab)=1$ and locally check $\ParentOutPos{r}<\InPos{r}<\OutPos{r}$ for every $r$.

  \item \textbf{$\mathsf{FinalSequenceOrder}$.} Enforce Eq.~\eqref{eq:pq-final-order} locally.

  \item \textbf{Output.} The circuit $\Cpq$ enforces the inlined stack, lookup, and multiset-equality checks in Steps~1--4, and outputs
  \[
  \pipq\leftarrow\ZKProve(\Cpq;\xpq;\wpq).
  \]
\end{enumerate}

\textbf{Verifier:} $\ZKVerify(\Cpq;\xpq;\pipq)\in\{\mathtt{accept},\mathtt{reject}\}$.
\end{minipage}
}
\endgroup
\caption{Protocol for the priority-queue proof.}
\label{fig:zk-priority-queue}
\end{figure}

\noindent{\bf Proof statement.} The proof establishes:
\begin{equation}
\begin{aligned}
&\mathsf{TraceConsistency}\bigl(
  \OpPosVec,\PushTrace,\PopTrace,\OutSeq\bigr)\\
&{}\wedge\ \mathsf{QueueAccounting}
  (\PushTrace,\OutSeq,\InPosVec)\\
&{}\wedge\ \mathsf{QueueOrder}\bigl(
  \OutSeq,\InPosVec,\ParentOutPosVec,\wstkPQ\bigr)\\
&{}\wedge\ \mathsf{FinalSequenceOrder}(\FinalPQ).
\end{aligned}
\label{eq:pq-condition}
\end{equation}

The first predicate checks trace formation, while the remaining predicates capture the conditions of Theorem~\ref{thm:pq-equiv-checker}:

\uline{$\mathsf{TraceConsistency}.$} This predicate binds the projections to one operation order. It range-checks every enabled finite position, enforces $1\leq\tau_a\leq P_{\mathsf{pq}}^{\max}$ and $\tau_a<\tau_{a+1}$, verifies that both projections are position-ordered, and binds $(\PopItem{j})_{j\in[\PopLen]}$ to the prefix of $\OutSeq$. It also requires
\begin{equation}
\{\PushTime{i}\}_{i\in[\PushLen]}\uplus\{\PopTime{j}\}_{j\in[\PopLen]}=_{\mathsf{ms}}\OpPosVec.
\label{eq:pq-position-distinctness}
\end{equation}

\uline{$\mathsf{QueueAccounting}.$} This predicate certifies that every pushed insertion appears exactly once in $\OutSeq$ as in Eq.~\eqref{eq:pq-in-multiset}.

\uline{$\mathsf{QueueOrder}$.}  For each $\OutItem{r}$, the inlined monotone-stack constraint identifies its latest preceding lower-priority item $\ParentOf{\OutItem{r}}=\OutItem{\PLIndex{r}}$, with $\NoParent$ when $\PLIndex{r}=0$. The circuit then forms an item-to-removal-position table containing $(\NoParent,0)$ and every $(\OutItem{r},\OutPos{r})$. Looking up $(\ParentOf{\OutItem{r}},\ParentOutPos{r})$ in this table binds $\ParentOutPos{r}=\OutPos{\PLIndex{r}}$. Finally, the circuit enforces $\ParentOutPos{r}<\InPos{r}<\OutPos{r}$, thereby establishing Eq.~\eqref{eq:pq-prevlower-main}.

\uline{$\mathsf{FinalSequenceOrder}.$} This predicate certifies that $\FinalPQ$ is in decreasing priority order by enforcing Eq.~\eqref{eq:pq-final-order}.

\noindent{\bf Proof generation.} Fig.~\ref{fig:zk-priority-queue} constructs the proof for Eq.~\eqref{eq:pq-condition} using the inlined monotone-stack checker, two multiset-equality checks, the removal-position lookup, and local constraints.

Both checkers have linear constraint complexity in their total padded workloads. Appendix~\ref{app:constraint-complexity} gives the analysis.

%% file: section/method_v3.tex
\section{\toolname}
\label{sec:method}

This section constructs \toolname\ to realize the security goals in Sec.~\ref{sec:target-workflow}.
We first instantiate the endorsed public context, then present the protocols for both planes.


\subsection{Initialization}
\label{sec:method-initial}

The protocol fixes a ZK-friendly hash $H_{\mathsf{zk}}$, from which it builds content and complete-message digests and the structured request hash $\ReqHash$ (detailed in Appendix~\ref{app:hash-instantiation}). It also fixes a signature hash $H_{\mathsf{sig}}$ for authenticated messages.
Initialization includes:
\begin{itemize}[leftmargin=1.25em,itemsep=2pt,topsep=2pt,parsep=0pt]
    \item \uline{Semantic index.}
    Using fresh $\RefRand$ and $\GraphRand$, the router publishes $c_{\mathsf{idx}}=(\RefCommit,\GraphCommit)$, where $\RefCommit=\Commit(\rtable;\RefRand)$ and $\GraphCommit=\Commit(\gtable;\GraphRand)$.
    The public parameters $\mathsf{pp}_{\mathsf{sel}}$ fix the search configuration and supported circuit capacities.
    \item \uline{Rewrite banks.} The router maintains two private banks underlying the rewrite relation $\trans$: complete messages in an endorsed message bank $\gapset$ and content fragments in an endorsed fragment bank $\fragset$.
    Using fresh $\MsgRand$ and $\FragRand$, it publishes $c_{\mathsf{rw}}=(\MsgCommit,\FragCommit)$, where $\MsgCommit=\Commit(\gapset;\MsgRand)$ and $\FragCommit=\Commit(\fragset;\FragRand)$.
    The public parameters $\mathsf{pp}_{\mathsf{rw}}$ fix the public rewrite structure and capacity bounds.
    \item \uline{Provider keys.} For every offered model $M$, its provider holds a signing key $sk_M$ and registers the corresponding verification key $pk_M$ in the authenticated registry $\mathsf{Reg}$. The client later uses this registry to authenticate the serving certificate for the model.
\end{itemize}

These values instantiate the public context from
Sec.~\ref{sec:target-workflow} as
\begin{equation}
\label{eq:public-context}
\ctx =
\bigl(
c_{\mathsf{idx}},
c_{\mathsf{rw}},
\mathsf{Reg},
\mathsf{pp}_{\mathsf{sel}},
\mathsf{pp}_{\mathsf{rw}}
\bigr).
\end{equation}

\input{section/method_subsection/control-plane}

\input{section/method_subsection/data-plane}


Overall, Appendix~\ref{app:e2e-protocol} gives the complete \toolname\ protocol.
The theorem below states its integrity guarantee.


\begin{theorem}[End-to-end security of \toolname]
\label{thm:verirouter-security}
Under the trust model in Sec.~\ref{sec:threat}, collision-resistant hashes, computationally binding and hiding commitments, EUF-CMA signatures, and a complete, knowledge-sound, and zero-knowledge argument system (Sec.~\ref{sec:bg-zkp}), \toolname\ satisfies Def.~\ref{def:verifiable-routing}:
\begin{enumerate}[label=(\roman*),itemsep=0pt,topsep=1pt]
    \item \uline{Completeness.}
    For every supported honest serving event in either mode, the client accepts the returned evidence.

    \item \uline{Soundness.}
    Except with negligible probability, if the client accepts a serving event under the applicable model-selection mode $\mmode$, then the system-integrity condition in Eq.~\eqref{eq:verification-target} holds under the endorsed public context $\ctx$.

    \item \uline{Confidentiality.}
    The returned cryptographic evidence reveals no additional information about the router's private data beyond the public protocol inputs and outputs and the fact that Eq.~\eqref{eq:verification-target} holds.
\end{enumerate}
\end{theorem}

Appendix~\ref{sec:method-e2e} gives the proof.

%% file: section/method_subsection/control-plane.tex
\subsection{Control-plane Integrity}
\label{sec:method-control}



Recall the model-selection condition $\Sel_{\mmode}$ defined in Sec.~\ref{sec:bg-routing}. On a fresh-selection turn, $\Sel_{\mathsf{fresh}}(\model;\clientreq,\Idx)$ requires $\model\in\Elig(\clientreq,\Idx)$. We construct the model-selection proof $\pi_{\mathsf{sel}}$ to certify this condition. 
We first describe the committed semantic index, then construct score-opening, upper-layer, and level-$0$ \textsf{efSearch} subproofs $\pi_{\mathsf{score}}$, $\pi_{\mathsf{upper}}$, and $\pi_{\mathsf{ef}}$. The aggregation circuit $\mathsf{C}_{\mathsf{sel}}$ verifies and links them. 
Fig.~\ref{fig:control-plane-composition} summarizes the subproofs and their shared interfaces.


\begin{figure}[t]
\centering
\includegraphics[width=0.9\linewidth]{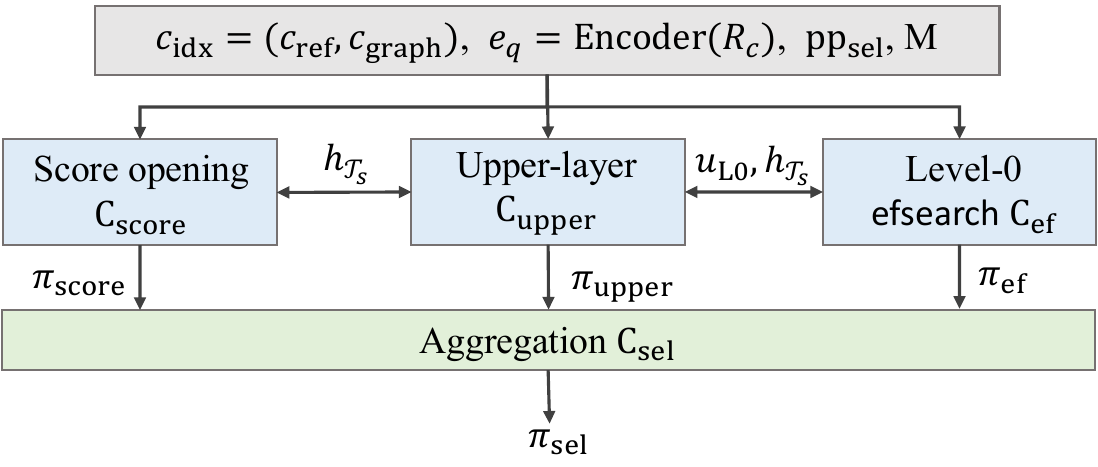}
\caption{Composition of the model-selection proof on a fresh-selection turn. The aggregation circuit verifies the subproofs, equates their $h_{\dtable}$ values that bind the score table, links the upper-layer output $\EntryNode$ to the level-$0$ input, and produces $\pi_{\mathsf{sel}}$.}
\label{fig:control-plane-composition}
\end{figure}


\subsubsection{Committed semantic index}
As initialized in Sec.~\ref{sec:method-initial}, $c_{\mathsf{idx}}=(\RefCommit,\GraphCommit)$ commits to reference rows $(u,v_u,\mathsf{mid}(u))$ and graph rows $(\ell,u,\mathsf{Pad}_{B_\ell}(\mathsf{Nbr}_{\ell}(u)))$ from Sec.~\ref{sec:bg-routing}. Here $\mathsf{Nbr}_{\ell}(u)$ lists the real neighbors in stored order, and the public bound $B_\ell$ fixes the padded row width. Whenever a subproof uses a reference tuple or neighbor list, it authenticates the corresponding row against $\RefCommit$ or $\GraphCommit$.



\subsubsection{Score-opening proof}
\label{sec:method-control-score} 

The same node score may be consumed multiple times during semantic search. The score-opening proof therefore certifies each required score once and makes the score table available to the traversal subproofs through a private shared interface.

\uline{Proof statement.} For a public capacity bound, the prover forms a canonical private score table $\dtable$: its active rows have unique node identifiers in increasing order, and its unused rows use fixed padding. Let $e_q=\encoder(\clientreq)$. For every active row $(u,\score_u,\mathsf{mid}(u))$, the circuit $\cscore$ authenticates $(u,v_u,\mathsf{mid}(u))$ against $\RefCommit$ and checks $\score_u=\score(e_q,v_u)$. It exposes an internal digest $h_{\dtable}$ that binds the table across the three subproofs and remains private in the final model-selection aggregation proof.


\uline{Parallelization.} Computing many high-dimensional scores against the same query is a major circuit cost. We therefore design a parallel dot-product gate that shares one copy of the query coordinates across several node-score computations in each gate row. Relative to a generic formula-gate realization, this reduces score-computation rows by about $64\%$. Appendix~\ref{app:score-opening-gate} gives the concrete construction.

The upper-layer and level-$0$ \textsf{efSearch} subproofs bind every consumed score to the same private $\dtable$. Both traversals prefer a larger certified score and retain the first encountered node on a tie. Appendix~\ref{app:wc-instantiation} gives the exact comparison encoding. 



\subsubsection{Upper-layer proof}
\label{sec:method-control-upper}
Let $u_{\mathsf{ent}}$ be the entry node at the highest layer $L_{\max}$. The upper-layer proof certifies the greedy traversal from $u_{\mathsf{ent}}$ to the level-$0$ entry. At each step, the traversal moves to the best neighbor only if it outranks the current node, or otherwise retains the current node and descends. The node retained after layer $1$ becomes the level-$0$ entry. Alg.~\ref{alg:hnsw-upper} gives this procedure.

\uline{Proof statement.} The private witness is a sequence $(\ell_i,\UpperNode{i},\mathsf{Nbr}_i,\mathsf{best}_i,\mathsf{stop}_i)_{i=1}^{\ustep}$ padded to the public step bound $\ustep$. For some $N_{\mathsf{upper}}^{\mathsf{act}}\geq1$, steps $1,\ldots,N_{\mathsf{upper}}^{\mathsf{act}}$ are enabled and the remaining slots are disabled. Step $i$ records the current layer and node $(\ell_i,\UpperNode{i})$, the real neighbors $\mathsf{Nbr}_i$, the claimed best node $\mathsf{best}_i\in\{\UpperNode{i}\}\cup\mathsf{Nbr}_i$, and a bit $\mathsf{stop}_i$ indicating descent after this step. The circuit $\cupper$ checks:

\emph{(i) Authentication.} Every traversed neighbor list is authenticated against $\GraphCommit$, and required scores are bound to $\dtable$.

\emph{(ii) Greedy traversal.}
Each enabled step scans $\UpperNode{i}$ and then $\mathsf{Nbr}_i$ and retains the first maximum score. The circuit checks this value as $\mathsf{best}_i$, sets $\mathsf{stop}_i=1\Longleftrightarrow\mathsf{best}_i=\UpperNode{i}$, and, for $i<N_{\mathsf{upper}}^{\mathsf{act}}$, enforces the layer transition.

\emph{(iii) Boundary conditions.} Every enabled step satisfies $1\leq\ell_i\leq L_{\max}$, and the first starts at $(L_{\max},u_{\mathsf{ent}})$. Let $i_\star=N_{\mathsf{upper}}^{\mathsf{act}}$ be the last enabled step. The circuit requires $\ell_{i_\star}=1$, $\mathsf{stop}_{i_\star}=1$, and $\mathsf{best}_{i_\star}=\UpperNode{i_\star}$, and sets $\EntryNode=\UpperNode{i_\star}$.


\subsubsection{Level-0 \textsf{efSearch} proof}
\label{sec:method-control-ef}

The level-$0$ proof certifies the top-$k$ output $F^{\mathsf{topk}}_{\WorkQ}$ of $\textsf{efSearch}$ starting from $\EntryNode$ and the selected-model membership condition $\exists u\in F^{\mathsf{topk}}_{\WorkQ}:\mathsf{mid}(u)=\model$. The search maintains a candidate queue $\CandQ$, a working queue $\WorkQ$ of capacity $\EfCap\geq k$, and a visited set. Each round expands the best candidate and stops when no remaining candidate can improve $\WorkQ$ (Alg.~\ref{alg:hnsw-efsearch}).

Rather than replaying the queue states, the prover supplies push and pop projections of $\CandQ$ and $\WorkQ$. It validates both queues with the priority-queue checker from Sec.~\ref{sec:zk-priority-queue}, then enforces equivalent cross-queue and stopping conditions below.

\uline{Proof statement.} The circuit $\cef$ certifies three conditions:

\emph{(i) Queue-trace validity.} The circuit inlines the priority-queue checker from Sec.~\ref{sec:zk-priority-queue} for both queues, treating $\CandQ$ as a max-priority queue and $\WorkQ$ as a min-priority queue under the search key. To enforce capacity $\EfCap$, $\WorkQ$ starts with $\EfCap$ sentinels below all real nodes. Each update pushes a real node and pops the minimum, either discarding the new node or replacing a sentinel or the current worst node. Thus, $\WorkQ$ retains the best nodes seen so far up to capacity, matching the conditional insertion in Alg.~\ref{alg:hnsw-efsearch}.

\emph{(ii) Cross-queue consistency.} Individually valid queue traces must also describe the same search execution. Let $T$ be the public padded length. To link the traces, the circuit uses the working-queue update sequence 
$\WorkUpdTrace=
\bigl((n_i,\tau_i,\mathsf{En}_i,\score_i,\AdmitBit{i})\bigr)_{i\in[T]}.$ 
The sequence starts with $n_1=\EntryNode$ and then follows $\CandQ$ pop order, with each popped node contributing its padded neighbor list. The value $\tau_i$ reserves row $i$'s position in the $\CandQ$ operation order as in Sec.~\ref{sec:zk-priority-queue} and serves as $n_i$'s tie-breaking position when enabled. The bit $\mathsf{En}_i$ selects the first-visit rows that form the real-node work-push sequence $\WorkPushTrace$, while $\AdmitBit{i}$ further selects from them the candidate-push sequence $\CandPushTrace$. Repeated and padding rows have $\mathsf{En}_i=\AdmitBit{i}=0$. As described below, the circuit checks this relation in three steps:

First, \emph{node-order authentication} binds the node sequence in $\WorkUpdTrace$ to candidate expansion. The first row has $n_1=\EntryNode$. For each node popped from $\CandQ$, the next fixed-width block of $n_i$ values is its padded level-$0$ neighbor row authenticated against $\GraphCommit$. The blocks follow pop order, and unused blocks are padding. Thus, the sequence contains the entry followed by authenticated neighbor identifiers in $\CandQ$ pop order.

Second, the \emph{first-visit filter} implements the visited set. The circuit checks a copy of $(n_i,i,\mathsf{En}_i)_{i\in[T]}$ sorted by node and original row index. A multiset-equality check binds the copies, and adjacent sorted rows enable only the earliest row of each real node. Hence, $\mathsf{En}_i=1$ exactly when $n_i$ is real and has not appeared earlier. For each enabled row, the circuit looks up the unique active $\dtable$ row for $n_i$, binds its score to $\score_i$, and performs the $\WorkQ$ push--pop using $(n_i,\score_i,\tau_i)$.

Third, the \emph{admission filter} links $\WorkQ$ admissions to $\CandQ$ pushes. For each enabled row $i$, let $\WorkEvictItem{i}$ be the node immediately popped by its update. The circuit sets $\AdmitBit{i}=1$ exactly when $\mathsf{En}_i=1$ and $\WorkEvictItem{i}\ne n_i$, meaning that $n_i$ remains in $\WorkQ$. It then forms the push trace of $\CandQ$: $\CandPushTrace=\bigl((\tau_i,n_i):\AdmitBit{i}=1\bigr)_{i\in[T]}$. Thus, $\CandQ$ receives exactly the admitted nodes with the same positions and order.


\emph{(iii) Stopping rule.} The terminal state represents either an empty-$\CandQ$ stop or a threshold stop just before its final pop, where $\WorkQ$ is full and the best node in $\CandQ$ is worse than the worst node in $\WorkQ$. To avoid scanning the live $\WorkQ$, the circuit permits an equivalent output-preserving suffix and checks each round locally. In any full-$\WorkQ$ round with an admission, the first admission evicts exactly the round-start worst node because preceding rows were rejected. The circuit requires the popped candidate to be no worse than this eviction. A below-threshold round therefore admits no update, leaves $\WorkQ$ unchanged, and pushes no candidate. Since the popped candidate was the $\CandQ$ maximum, every later candidate is no better, so the suffix is inert and can be removed to obtain an equivalent stopped execution with the same final $\WorkQ$.

The ordered real portion of the certified final $\WorkQ$ determines $F^{\mathsf{topk}}_{\WorkQ}$, which is the $\mathsf{SS}(e_q,\Idx)$ in Sec.~\ref{sec:bg-routing} and the circuit checks $\mathsf{mid}(u)=\model$ for some $u\in \mathsf{SS}(e_q,\Idx)$. Together, these checks produce $\pi_{\mathsf{ef}}$. Appendix~\ref{app:wc-instantiation} gives the detailed constraints, and Appendix~\ref{app:proof-efsearch} proves correctness.

The router $\router$ generates the three subproofs in parallel. Aggregation binds them to the same $e_q$, $c_{\mathsf{idx}}$, $\mathsf{pp}_{\mathsf{sel}}$, and internal score-table hash $h_{\dtable}$. It also equates $\cupper$'s output $\EntryNode$ with $\cef$'s input. The resulting $\pi_{\mathsf{sel}}$ establishes $\Sel_{\mathsf{fresh}}(\model;\clientreq,\Idx)$, or equivalently, $\model\in\Elig(\clientreq,\Idx)$. Appendix~\ref{app:zk-control-agg-protocol} gives the detailed protocol. On the other hand, a pinned turn requires no new
$\pi_{\mathsf{sel}}$. The client instead performs a \emph{pinned-model check}, requiring $M=M_{\mathsf{pin}}$ for the model retained from the session's initial accepted fresh-selection turn. Thus, the model-selection proof on a fresh-selection turn and the pinned-model check on a pinned turn both establish the corresponding condition $\Sel_{\mmode}$ in Eq.~\eqref{eq:eligible-mode}.


\begin{theorem}[Control-plane integrity]
\label{thm:control-plane}
Under collision-resistant hashing, computationally binding commitments, and a complete and knowledge-sound argument system (Sec.~\ref{sec:bg-zkp}), the control-plane verification satisfies:
\begin{enumerate}[label=(\roman*),itemsep=0pt,topsep=1pt]
    \item \uline{Completeness.}
    For every supported honest serving event,
    $\Sel_{\mmode}(M;R_c,\Idx)$ holds and the corresponding control-plane verification accepts.

    \item \uline{Soundness.}
    Except with negligible probability, every accepted verification establishes
    $\Sel_{\mmode}(M;R_c,\Idx)$, where $\Idx$ is bound by the endorsed public context and $M_{\mathsf{pin}}$ is the client-held session state when $\mmode=\mathsf{pinned}$.
\end{enumerate}
\end{theorem}

%% file: section/method_subsection/data-plane.tex
\subsection{Data-plane Integrity}
\label{sec:method-data}
This subsection constructs the rewrite proof $\pi_{\mathsf{rw}}$ and serving certificate $\mathsf{SC}$ for the \DataPlane. Both use a shared hiding commitment $\ReqCommit$ to the model-received request. We present the certificate first to authenticate $\ReqCommit$, then the proof to certify that the committed request satisfies the rewrite relation.


\subsubsection{Serving certificate}
\label{sec:method-data-sig}
Using the signature interface from Sec.~\ref{sec:bg-sig}, the serving certificate authenticates a provider response and binds it to the request received by the model. For each event, the router $\router$ samples $\ReqRand$ and sets $\ReqCommit=\Commit(\ReqHash(\receivedreq);\ReqRand)$. Here $\ReqHash$ binds the messages in order, including their roles and contents, as detailed in Appendix~\ref{app:hash-instantiation}. The router $\router$ sends $(\rid,\model,\receivedreq,\ReqRand)$ to $\provider_\model$. The provider recomputes $\ReqCommit$, evaluates $Y=\model(\receivedreq)$, and constructs the certificate below.
\begin{equation}
\begin{alignedat}{2}
\ReqCommit&=\Commit(\ReqHash(\receivedreq);\ReqRand),\; &
h_Y&=H_{\mathsf{sig}}(Y),\\[-2pt]
d_{\mathsf{srv}}&=H_{\mathsf{sig}}(\rid\|\model\|\ReqCommit\|h_Y),\; &
\sigma_\model&\leftarrow\mathsf{Sign}_{sk_\model}(d_{\mathsf{srv}}),\\[-2pt]
\mathsf{SC}&=(\rid,\ReqCommit,h_Y,\sigma_\model).&&
\end{alignedat}
\label{eq:serving-certificate}
\end{equation}


For verification, the client $\client$ obtains $pk_\model=\mathsf{Reg}[\model]$, checks $h_Y=H_{\mathsf{sig}}(Y)$, recomputes $d_{\mathsf{srv}}$ from Eq.~\eqref{eq:serving-certificate}, and accepts only if $\rid$ equals the identifier issued for this request and $\mathsf{Verify}_{pk_\model}(d_{\mathsf{srv}},\sigma_\model)=1$.


\subsubsection{ZK rewrite proof}
\label{sec:method-data-rewrite}

The rewrite proof establishes the request relation in Sec.~\ref{sec:bg-routing} for variable-length requests. It preserves client messages and reduces validation of router-supplied content to local placement checks and committed lookups. We first formalize the relation and then describe the proof statement.



\uline{Rewrite relation.} Routers commonly assemble model requests by adding complete messages, such as system instructions or retrieval context, or attaching prompt-formatting fragments around existing content~\cite{notdiamond-prompt-optimization,requesty_prompt_optimization,portkey_prompt_engineering_studio}. Fig.~\ref{fig:authorized-rewrite} instantiates both forms around one client message. Let $\clientreq=[x_1,\ldots,x_{N_{\mathsf{cmsg}}}]$ with $N_{\mathsf{cmsg}}\geq1$. A request in $\trans(\clientreq)$ preserves these messages in order, inserts messages before, between, or after them, and may prepend or append fragments to each client-message content. Formally,
\begin{equation}
\begin{aligned}
\receivedreq\in\trans(\clientreq)\quad\Longleftrightarrow\quad
\receivedreq={}&G_0\mathbin{\Vert}[\mathsf{wrap}(x_1)]\mathbin{\Vert}G_1\mathbin{\Vert}\cdots\\[-2pt]
&{}\mathbin{\Vert}[\mathsf{wrap}(x_{N_{\mathsf{cmsg}}})]{}\mathbin{\Vert}G_{N_{\mathsf{cmsg}}},
\end{aligned}
\label{eq:gformat}
\end{equation}
where each $G_i$, for $i\in\{0,\ldots,N_{\mathsf{cmsg}}\}$, contains at most $\MsgBound$ messages $g_{i,j}$ drawn from $\gapset$. Each entry contains a role and content, allowing request-dependent selection of messages such as system instructions or retrieval context.

\begin{figure}[t]
    \centering
    \includegraphics[width=\linewidth]{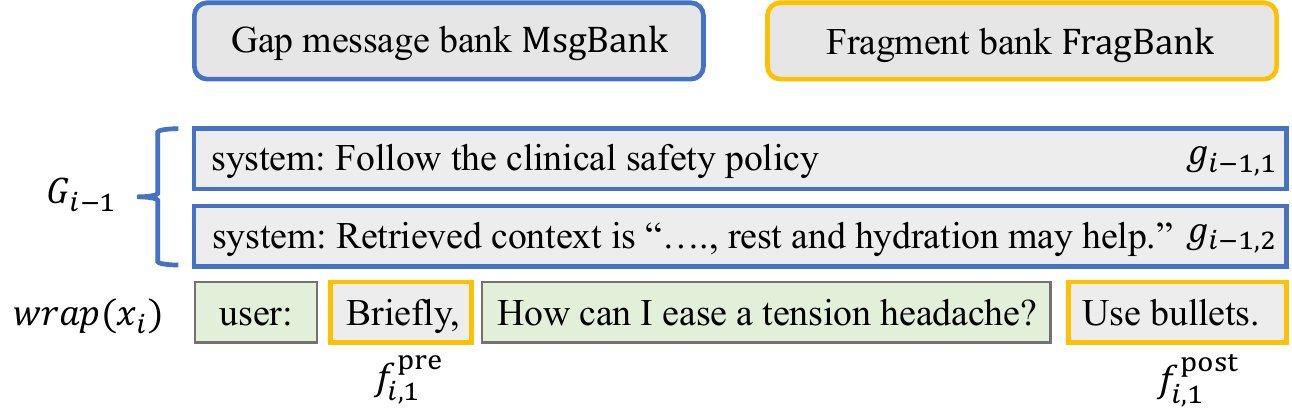}
    \caption{Endorsed rewrite example. The complete messages in $G_{i-1}$ are selected from $\gapset$, while the prefix and suffix fragments in $\mathsf{wrap}(x_i)$ are selected from $\fragset$.}
    \label{fig:authorized-rewrite}
\end{figure}

The wrapper $\mathsf{wrap}(x_i)$ preserves both fields of $x_i=(\mathsf{role}_i,\mathsf{content}_i)$, and permits authenticated fragments around the content. Thus $\mathsf{wrap}(x_i)$ is
\begin{equation}
\begin{aligned}
    \bigl(\mathsf{role}_i,\,
    f^{\mathsf{pre}}_{i,1}\|\cdots\|f^{\mathsf{pre}}_{i,N^{\mathsf{frag}}_{i,\mathsf{pre}}}\|\mathsf{content}_i \|f^{\mathsf{post}}_{i,1}\|\cdots\|f^{\mathsf{post}}_{i,N^{\mathsf{frag}}_{i,\mathsf{post}}}\bigr).
\end{aligned}
    \label{eq:fformat}
\end{equation}

For each $S\in\{\mathsf{pre},\mathsf{post}\}$ and $i\in[N_{\mathsf{cmsg}}]$, $(f^S_{i,1},\ldots,f^S_{i,N^{\mathsf{frag}}_{i,S}})$ denotes the fragments placed on that side of message $i$, where $0\leq N^{\mathsf{frag}}_{i,S}\leq\FragBound$. All selected fragments $\mathsf{Fragments}=\{f^S_{i,j}\}$ must belong to $\fragset$. The relation constrains placement and bank membership, allowing different system prompts, retrieval passages, and provider formats. Eqs.~\eqref{eq:gformat} and~\eqref{eq:fformat} also preserve every original role and content, thereby excluding common rewrite-relation violations, including message deletion, replacement, and reordering.

\uline{Proof statement.} The rewrite proof checks the endorsed rewriting against the same $\ReqCommit$ authenticated by the serving certificate. The public inputs are $\clientreq$, $\ReqCommit$, $\MsgCommit$, $\FragCommit$, and $\mathsf{pp}_{\mathsf{rw}}$. The private witness contains $\ReqRand$, the selected additions, their bank-opening and hashing data, and the message counts bounded by the leaf capacities. The circuit enforces:
\begin{equation}
\begin{aligned}
&\mathsf{RequestCommitment}(\receivedreq,\ReqRand,\ReqCommit)\\
&\wedge\ \mathsf{LayoutConsistency}(\clientreq,\receivedreq,(G_i)_{i=0}^{N_{\mathsf{cmsg}}},\mathsf{Fragments})\\
&\wedge\ \mathsf{GapIn}((G_i)_{i=0}^{N_{\mathsf{cmsg}}},\MsgCommit)\\
&\wedge\ \mathsf{FragIn}(\mathsf{Fragments},\FragCommit).
\end{aligned}
\label{eq:rewrite-proof-relation}
\end{equation}

$\mathsf{RequestCommitment}$ checks $\ReqCommit=\Commit(\ReqHash(\receivedreq);\allowbreak \ReqRand)$, while $\mathsf{LayoutConsistency}$ enforces Eqs.~\eqref{eq:gformat} and~\eqref{eq:fformat}. $\mathsf{GapIn}$ and $\mathsf{FragIn}$ authenticate each selected addition against the corresponding commitment through committed lookups.

\uline{Digest reuse.} To avoid rehashing complete inserted messages, we commit their digests rather than raw bytes under $\MsgCommit$. After $\mathsf{GapIn}$ authenticates an entry, $\mathsf{RequestCommitment}$ reuses the digest in $\ReqHash$. $\FragCommit$ instead binds fragment digests only for membership checks, and each selected fragment must still be combined with the original content and hashed to obtain the digest of $\mathsf{wrap}(x_i)$. Appendix~\ref{app:hash-instantiation} gives the details.



To support variable-length requests, the router partitions the client messages into consecutive ranges and proves them in parallel. Each leaf checks the predicates above for its range and advances a running request-hash state. Using recursive composition (Sec.~\ref{sec:bg-zkp}), the aggregation circuit verifies the leaf proofs, checks their ranges, and enforces agreement between adjacent boundary values and the final commitment $\ReqCommit$. This produces $\pi_{\mathsf{rw}}$. Appendix~\ref{app:zk-rewrite-protocol} gives the detailed protocol.



The client $\client$ verifies $\pi_{\mathsf{rw}}$ and $\mathsf{SC}$ with the same $\ReqCommit$. The rewrite proof certifies the rewrite relation, while the certificate authenticates the provider's response for that request. Together they establish endorsed rewriting and faithful relay, the latter two conditions of Eq.~\eqref{eq:verification-target}. 


\begin{theorem}[Data-plane integrity]
\label{thm:data-plane}
Under the trust model in Sec.~\ref{sec:threat}, collision-resistant hashing, computationally binding commitments, EUF-CMA signatures, and a complete and knowledge-sound argument system (Sec.~\ref{sec:bg-zkp}), the rewrite proof and serving certificate satisfy:
\begin{enumerate}[label=(\roman*),itemsep=0pt,topsep=1pt]
    \item \uline{Completeness.}
    For any supported honest execution with an endorsed rewrite, the rewrite proof and serving certificate are accepted.

    \item \uline{Soundness.}
    Except with negligible probability, whenever the client accepts a serving certificate $\mathsf{SC}$ and rewrite proof $\pi_{\mathsf{rw}}$ bound to the same request commitment, the serving event identified by $\rid$ satisfies $Y=\model(\receivedreq)$ for some $\receivedreq\in\trans(\clientreq)$.
\end{enumerate}
\end{theorem}
Appendix~\ref{app:proof-data-plane} gives the proof.

%% file: section/evaluation_v2.tex
\section{Evaluation}
\label{sec:evaluation}






\subsection{Experimental Setup}
\label{sec:eval-setup}

\noindent\textbf{Implementation.}
We implement \toolname\ in Rust using Plonky2~\cite{polygonzero2022plonky2}. The semantic index uses standard 16-bit affine quantization~\cite{jacob2018quantization}. Appendix~\ref{app:quantized-fidelity} reports a routing-quality ablation against floating-point search. We instantiate $H_{\mathsf{zk}}$ with Poseidon~\cite{grassi2021poseidon}, $H_{\mathsf{sig}}$ with SHA-256~\cite{pub2012secure}, and provider signatures with Ed25519~\cite{brendel2021provable}. The model-selection subproofs and rewrite leaf proofs are generated in parallel before their respective aggregation steps. All measurements enable zero-knowledge mode.
All experiments are conducted on an Alibaba Cloud \texttt{ecs.r9ae.8xlarge} with 32 vCPUs and 256~GiB of memory. We use 32 threads by default. Runs exceeding the host-memory limit are reported as \textcolor{oommuted}{OOM}.

\noindent\textbf{Baselines.}
We compare \toolname\ with two general-purpose zkVMs, SP1~\cite{succinct_sp1_repo} and Jolt~\cite{arun2024jolt}. They prove the same quantized semantic search, rewrite, and binary-heap relations by executing the corresponding guest programs directly. We run SP1 in compressed mode, which compresses its sharded execution into one proof, and report each system's client-facing serialized proof for a fair comparison of proof size and verification cost. We additionally implement direct Plonky2 replay for model selection and the priority-queue checker to isolate the benefits of control-plane specializations.

\noindent\textbf{Datasets.}
We evaluate \toolname\ on real-world semantic-search and conversational workloads. For the control plane, the public RouteLLM corpus~\cite{ong2025routellm} provides 55,361 1,536-dimensional training embeddings and 1,531 disjoint test queries. For the data plane and end-to-end evaluation, MT-Bench-101~\cite{bai2024mt} provides 1,388 multi-turn conversations and 4,208 turns across 13 task types. To construct realistic rewrite banks, we use GPT-5.4~\cite{gpt5-4} offline to populate $\gapset$ and $\fragset$ with reusable messages and fragments.


\input{tables/control-plane-cost}

\subsection{RQ1: Control-Plane Performance}
\label{sec:eval-control}

We evaluate control-plane performance at three levels: the model-selection proof, subproofs, and the standalone priority-queue checker. \toolname's model-selection proof is 1,307$\times$ faster than SP1, while its priority-queue checker is 17.1--52.0$\times$ faster than the fastest completing baseline.

\noindent\textbf{Model-selection proof.}
Tab.~\ref{tab:route-cost} reports the overall cost of the model-selection proof. The 1,531 search traces span seven sizes, specified by score-table size $|\dtable|$ and the length $|\WorkUpdTrace|$ of the working-queue $\WorkQ$ update sequence. Across these sizes, \toolname's proving time is 9.41--36.92~s and 21.92~s on average, whereas SP1 takes 228.2--706.8~min and 477.4~min on average. SP1's default prover uses sharded execution and recursive compression to complete all sizes, but remains slow. Jolt and direct Plonky2 replay instead exceed the 256~GiB limit after running for several minutes even at the smallest size.
\toolname also reduces verification time, proof size, and peak memory. Average verification takes 4.87~ms rather than SP1's 37.15~ms, proof size is 0.15~MiB rather than 1.21~MiB, and average peak memory is 27.7~GiB rather than 172.0~GiB. 


\input{tables/pq_cost}

Within \toolname, at fixed $|\WorkUpdTrace|=2{,}048$, quadrupling $|\dtable|$ raises overall proving time from 9.41 to 30.60~s (3.25$\times$); at fixed $|\dtable|=2{,}048$, quadrupling $|\WorkUpdTrace|$ raises it only from 30.60 to 36.92~s (1.21$\times$). Thus, overall proving time is more sensitive to $|\dtable|$. The component measurements below help explain this difference. Verification time and proof size remain nearly constant. 

\begin{figure}[t]
    \centering
    \includegraphics[width=0.85\linewidth]{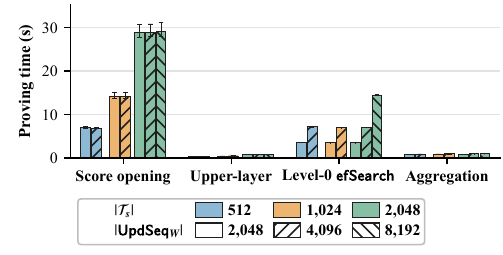}
    \caption{Isolated model-selection component proving time. Color denotes $|\dtable|$, and hatching denotes $|\WorkUpdTrace|$. The score-opening proof generally dominates.}
    \label{fig:eval-subproof}
\end{figure}

\noindent\textbf{Subproof breakdown.}
To isolate costs, we run each subproof separately with the full prover thread pool. Fig.~\ref{fig:eval-subproof} shows that the score-opening proof is the largest component in all but one size, taking 6.85--29.10~s. The level-$0$ \textsf{efSearch} proof takes 3.52--14.43~s and is nearly tied with the score-opening proof only at $(|\dtable|,|\WorkUpdTrace|)=(512,4{,}096)$, while the upper-layer proof and aggregation proof remain below 1~s. Thus authenticated high-dimensional scoring, even with the 64\%-row-reducing gate, determines most end-to-end time. Our specialized checks, on the other hand, keep traversal proof efficient without hitting a memory barrier. 

\noindent\textbf{Standalone priority-queue checker.}
We further isolate the priority-queue checker (Sec.~\ref{sec:zk-priority-queue}) used by level-$0$ \textsf{efSearch} by comparing it on operation traces with queue replay in baselines. Tab.~\ref{tab:pq-cost} shows Jolt and direct Plonky2 reach \textcolor{oommuted}{OOM} at 100K and 5K operations, respectively. SP1 completes every size using sharded execution and recursive compression, but takes approximately 52.5$\times$ longer than our checker on average. These results demonstrate the benefit of replacing mutable-queue replay with our trace checker. 

\subsection{RQ2: Data-Plane Performance}
\label{sec:eval-data}

We evaluate data-plane performance by comparing our rewrite proof with general-purpose baselines and measuring serving-certificate costs. Overall, \toolname's rewrite proof is 100.6$\times$ faster on average than Jolt, the fastest completing baseline, and provider-side certificate generation averages only 3.84~ms, which is lightweight. Throughout the data-plane experiments, message size is the maximum client-message content length $\max_i |\mathsf{content}_i|$ in $R_c$.

\input{tables/rewrite-cost}

\noindent\textbf{Rewrite proof.}
Tab.~\ref{tab:rewrite-cost} shows that \toolname generates a rewrite proof in 2.26--8.06~s, compared with 4.9--8.3~min for Jolt and 13.7--20.2~min for SP1. Jolt's proofs are 0.04--0.05~MiB smaller than \toolname's 0.14~MiB proofs, but require 100.6$\times$ longer proving and 59.1$\times$ longer verification on average. By avoiding general-purpose execution traces, \toolname reduces average peak memory to 7.1~GiB, while Jolt and SP1 use 23.0$\times$ and 19.5$\times$ as much. Increasing message size from 256~B to 4~KiB raises \toolname's proving time by only 3.6$\times$, while its verification time remains about 4.6~ms and its proof size remains 0.14~MiB.

\input{tables/serving-certificate-cost}

\noindent\textbf{Serving certificate.}
Across all message sizes, Tab.~\ref{tab:serving-certificate-cost} reports 3.46--4.15~ms for provider generation and 2.46--2.49~ms for client verification, averaging 3.84~ms and 2.47~ms, respectively. This low provider cost makes lightweight participation practical for protecting model attribution and reputation~\cite{zhang2026real,c2pa2025spec}, consistent with Sec.~\ref{sec:threat}.

\subsection{RQ3: End-to-End Performance}
\label{sec:eval-e2e}

We evaluate end-to-end user-visible serving latency and client verification under per-turn and session-pinned model selection. The added latency is only 4.08~s  and 0.11~s, respectively.

\begin{figure}[t]
    \centering
    \includegraphics[width=0.9\linewidth]{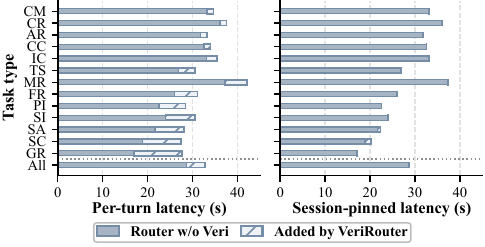}
    \caption{Average per-turn serving latency across MT-Bench-101 task types. \toolname adds 4.08~s (14.24\%) with per-turn model selection and 0.11~s (0.38\%) with session-pinned model selection.}
    \label{fig:e2e-overhead}
\end{figure}

\noindent\textbf{Serving latency.}
Fig.~\ref{fig:e2e-overhead} shows that overlapping proof generation with inference hides most proof-generation time in both modes. Relative to serving without \toolname, which averages 28.68~s per turn, per-turn model selection averages 32.77~s, adding 4.08~s (14.24\%) while hiding 83.92\% of proof-generation time. Session-pinned selection amortizes one model-selection proof across the dialogue and averages 28.79~s per turn, adding 0.11~s (0.38\%) while hiding 98.75\% of proof-generation time. The task-type breakdown also shows the latency stress case, as shorter inference exposes more of the proof path and yields higher relative overhead.

\noindent\textbf{Client verification.}
%
Client verification averages 11.91~ms for a fresh selection, covering both ZK proofs and the serving certificate, and 7.04~ms on later session-pinned turns.

\subsection{RQ4: Scalability}
\label{sec:eval-scalability}
We evaluate scalability with graph size on the control plane and with rewrite-bank capacity and rewrite bounds on the data plane. Proving remains stable for graphs up to 10M nodes and banks up to 16K entries. Increasing the rewrite bounds adds cost only when it expands byte-level fragment processing, while final aggregation proofs remain about 0.15~MiB.

\begin{figure}[b!]
    \centering
    \includegraphics[width=\linewidth]{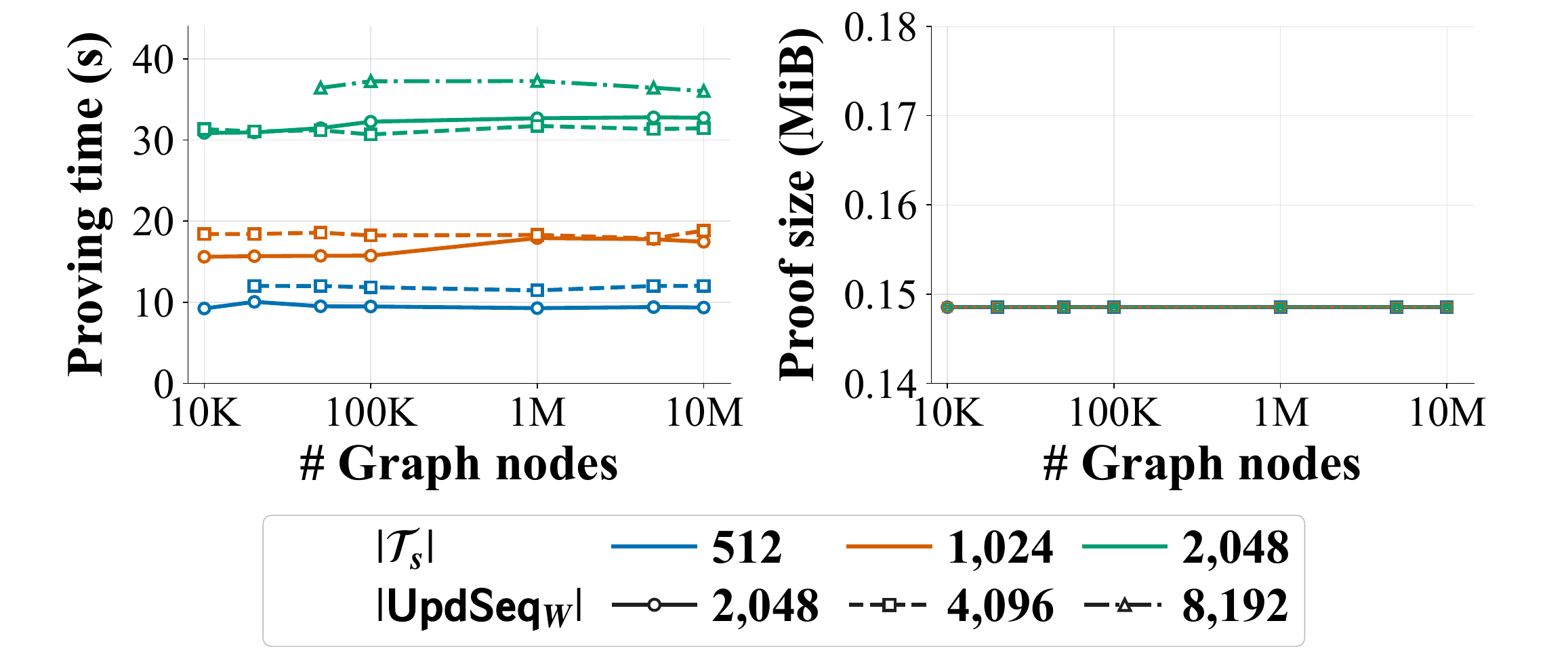}
    \caption{Scalability of model-selection proof with the number of graph nodes. Color denotes $|\dtable|$, while line style denotes $|\WorkUpdTrace|$. At each size, proving time and proof size remain stable as the graph grows to 10M nodes.}
    \label{fig:route-scalability}
\end{figure}

\noindent\textbf{Model-selection proof.} Fig.~\ref{fig:route-scalability} shows that model-selection proving remains stable at each $(|\dtable|,|\WorkUpdTrace|)$ size as the graph grows to 10M nodes. For $(512,2{,}048)$, proving changes from 9.23 to 9.35~s between 10K and 10M nodes. For $(2{,}048,8{,}192)$, it remains within 36.03--37.27~s between 50K and 10M nodes. This stability follows because each proof authenticates a bounded search trace rather than scanning the full graph. Across the trace sizes, the proving-time differences are driven mainly by $|\dtable|$, consistent with Fig.~\ref{fig:eval-subproof}. Every model-selection aggregation proof remains 0.15~MiB.

\begin{figure}[t]
    \centering
    \includegraphics[width=0.9\linewidth]{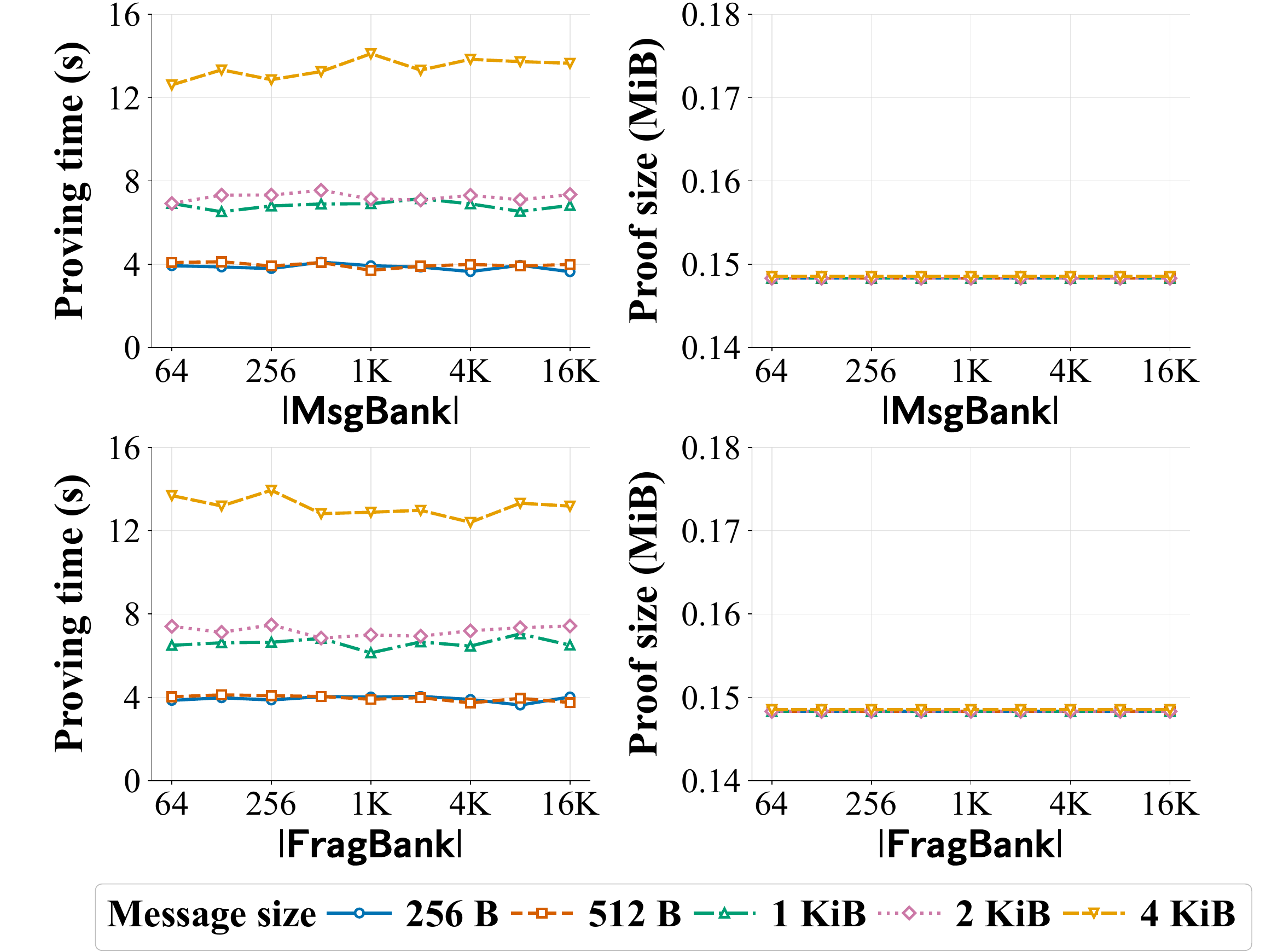}
    \caption{Scalability of rewrite proof with the capacity of rewrite banks $|\gapset|$ and $|\fragset|$. Proving time and proof size remain stable.}
    \label{fig:rewrite-bank-scalability}
\end{figure}

\noindent\textbf{Capacities of rewrite banks.} Fig.~\ref{fig:rewrite-bank-scalability} shows that rewrite proving remains stable as either bank grows from 64 to 16K entries. Across message sizes, the endpoint change is at most 1.04~s for $|\gapset|$ and 0.51~s for $|\fragset|$. This stability follows because the bank-membership checks in Sec.~\ref{sec:method-data-rewrite} authenticate only selected entries instead of processing an entire bank. The rewrite proof remains about 0.15~MiB throughout.

\begin{figure}[b!]
    \centering
    \includegraphics[width=0.9\linewidth]{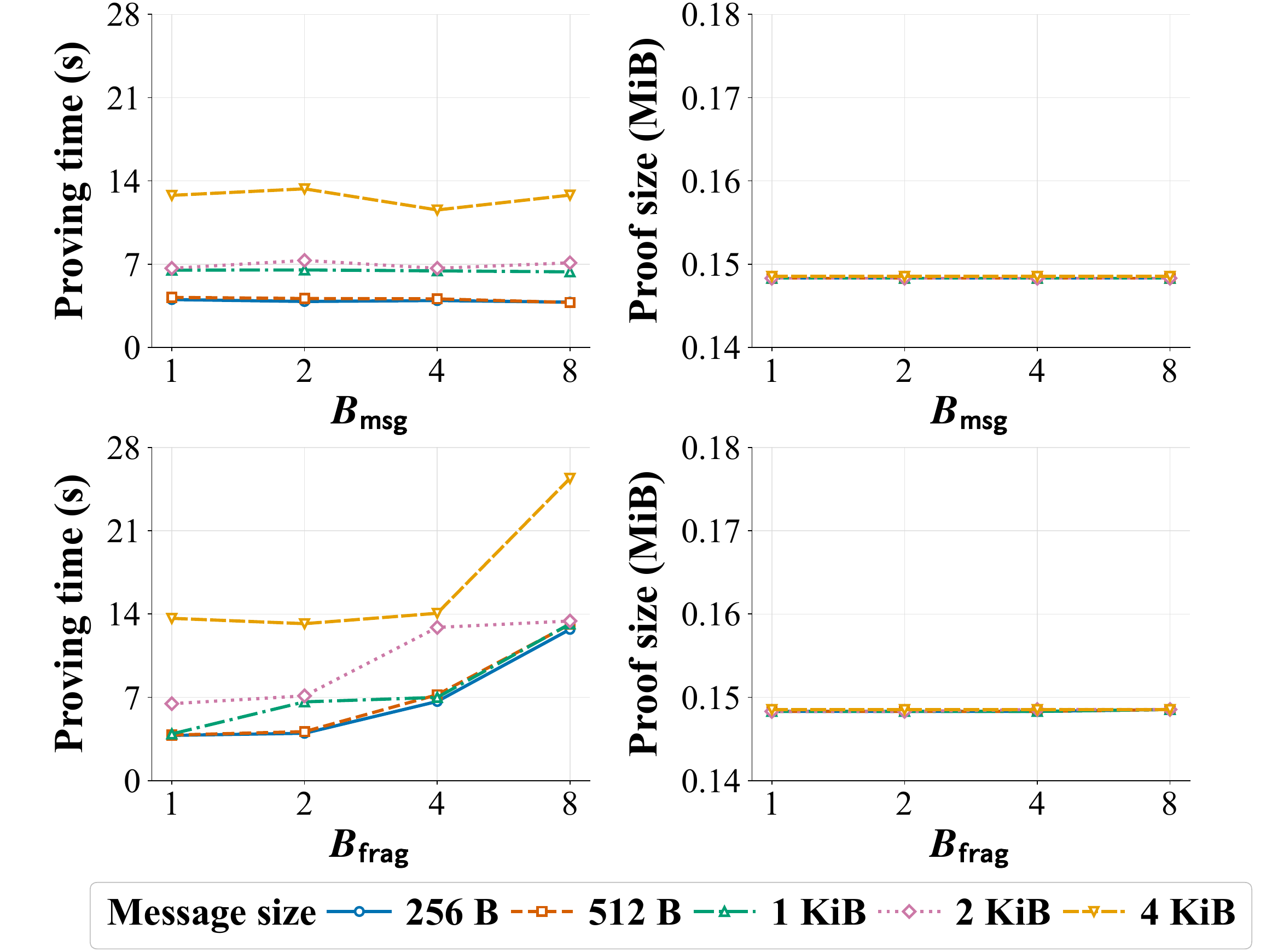}
    \caption{Scalability of rewrite proof with insertion bounds $\MsgBound$ and $\FragBound$. Proving time is stable with $\MsgBound$ and scales predictably with $\FragBound$, while the proof size remains 0.15~MiB.}
    \label{fig:rewrite-bound-scalability}
\end{figure}

\noindent\textbf{Rewrite insertion bounds.}
Fig.~\ref{fig:rewrite-bound-scalability} evaluates scalability under the two rewrite bounds $\MsgBound$ and $\FragBound$. The results reflect the two hashing paths in our rewrite proof. Across message sizes, increasing $\MsgBound$ leaves average proving time essentially unchanged because each complete inserted message contributes an authenticated bank digest directly to the ordered fold of $\ReqHash$. By contrast, increasing $\FragBound$ from 1 to 8 raises average proving time from 6.33 to 15.57~s because a fragment forms part of $\mathsf{wrap}(x_i)$ and must be hashed with the original content to derive the wrapped-message digest. Since consecutive content on the same side can be concatenated and committed as one fragment, eight independent fragments per side is already a generous practical bound. Recursive aggregation keeps the rewrite aggregation proof about 0.15~MiB.

%% file: tables/control-plane-cost.tex
\begin{table}[t]
\centering
\caption{Control-plane model-selection cost. Best completing results are in \textbf{bold}. \toolname reduces average proving time by 1,307$\times$ over SP1, while Jolt and direct Plonky2 replay exceed 256~GiB of memory.}
\label{tab:route-cost}
\begingroup
\setlength{\tabcolsep}{3pt}
\renewcommand{\arraystretch}{0.94}
\resizebox{0.9\columnwidth}{!}{%
\begin{tabular}{cc l cccc}
\toprule
$|\dtable|$ & $|\WorkUpdTrace|$ & \textbf{Method} & \makecell{\textbf{Proving}\\\textbf{time}} & \makecell{\textbf{Verification}\\\textbf{time}} & \makecell{\textbf{Proof}\\\textbf{size}} & \makecell{\textbf{Peak}\\\textbf{RSS}} \\
\midrule
\multirow{8}{*}{$512$} & \multirow{4}{*}{$2,048$} & SP1 & 228.2 min & 36.99 ms & 1.21 MiB & 171.4 GiB \\
 &  & Jolt & \textcolor{oommuted}{OOM} & -- & -- & \textcolor{oommuted}{$>256$ GiB} \\
 &  & Direct Plonky2 & \textcolor{oommuted}{OOM} & -- & -- & \textcolor{oommuted}{$>256$ GiB} \\
 &  & \textbf{\toolname} & \textbf{9.41 s} & \textbf{4.89 ms} & \textbf{0.15 MiB} & \textbf{11.5 GiB} \\
\cmidrule(lr){2-7}
 & \multirow{4}{*}{$4,096$} & SP1 & 230.2 min & 37.16 ms & 1.21 MiB & 154.3 GiB \\
 &  & Jolt & \textcolor{oommuted}{OOM} & -- & -- & \textcolor{oommuted}{$>256$ GiB} \\
 &  & Direct Plonky2 & \textcolor{oommuted}{OOM} & -- & -- & \textcolor{oommuted}{$>256$ GiB} \\
 &  & \textbf{\toolname} & \textbf{12.18 s} & \textbf{4.86 ms} & \textbf{0.15 MiB} & \textbf{15.6 GiB} \\
\midrule
\multirow{8}{*}{$1,024$} & \multirow{4}{*}{$2,048$} & SP1 & 385.7 min & 37.05 ms & 1.21 MiB & 175.4 GiB \\
 &  & Jolt & \textcolor{oommuted}{OOM} & -- & -- & \textcolor{oommuted}{$>256$ GiB} \\
 &  & Direct Plonky2 & \textcolor{oommuted}{OOM} & -- & -- & \textcolor{oommuted}{$>256$ GiB} \\
 &  & \textbf{\toolname} & \textbf{15.07 s} & \textbf{4.86 ms} & \textbf{0.15 MiB} & \textbf{19.9 GiB} \\
\cmidrule(lr){2-7}
 & \multirow{4}{*}{$4,096$} & SP1 & 387.9 min & 37.01 ms & 1.21 MiB & 171.8 GiB \\
 &  & Jolt & \textcolor{oommuted}{OOM} & -- & -- & \textcolor{oommuted}{$>256$ GiB} \\
 &  & Direct Plonky2 & \textcolor{oommuted}{OOM} & -- & -- & \textcolor{oommuted}{$>256$ GiB} \\
 &  & \textbf{\toolname} & \textbf{17.78 s} & \textbf{4.85 ms} & \textbf{0.15 MiB} & \textbf{24.3 GiB} \\
\midrule
\multirow{12}{*}{$2,048$} & \multirow{4}{*}{$2,048$} & SP1 & 700.3 min & 37.09 ms & 1.21 MiB & 179.1 GiB \\
 &  & Jolt & \textcolor{oommuted}{OOM} & -- & -- & \textcolor{oommuted}{$>256$ GiB} \\
 &  & Direct Plonky2 & \textcolor{oommuted}{OOM} & -- & -- & \textcolor{oommuted}{$>256$ GiB} \\
 &  & \textbf{\toolname} & \textbf{30.60 s} & \textbf{4.87 ms} & \textbf{0.15 MiB} & \textbf{34.7 GiB} \\
\cmidrule(lr){2-7}
 & \multirow{4}{*}{$4,096$} & SP1 & 702.8 min & 37.53 ms & 1.21 MiB & 176.1 GiB \\
 &  & Jolt & \textcolor{oommuted}{OOM} & -- & -- & \textcolor{oommuted}{$>256$ GiB} \\
 &  & Direct Plonky2 & \textcolor{oommuted}{OOM} & -- & -- & \textcolor{oommuted}{$>256$ GiB} \\
 &  & \textbf{\toolname} & \textbf{31.46 s} & \textbf{4.89 ms} & \textbf{0.15 MiB} & \textbf{40.3 GiB} \\
\cmidrule(lr){2-7}
 & \multirow{4}{*}{$8,192$} & SP1 & 706.8 min & 37.20 ms & 1.21 MiB & 176.2 GiB \\
 &  & Jolt & \textcolor{oommuted}{OOM} & -- & -- & \textcolor{oommuted}{$>256$ GiB} \\
 &  & Direct Plonky2 & \textcolor{oommuted}{OOM} & -- & -- & \textcolor{oommuted}{$>256$ GiB} \\
 &  & \textbf{\toolname} & \textbf{36.92 s} & \textbf{4.88 ms} & \textbf{0.15 MiB} & \textbf{47.3 GiB} \\
\midrule
\multicolumn{2}{c}{\multirow{4}{*}{\textbf{Average}}} & SP1 & 477.4 min & 37.15 ms & 1.21 MiB & 172.0 GiB \\
 & & Jolt & \textcolor{oommuted}{OOM} & -- & -- & \textcolor{oommuted}{$>256$ GiB} \\
 & & Direct Plonky2 & \textcolor{oommuted}{OOM} & -- & -- & \textcolor{oommuted}{$>256$ GiB} \\
 & & \textbf{\toolname} & \textbf{21.92 s} & \textbf{4.87 ms} & \textbf{0.15 MiB} & \textbf{27.7 GiB} \\
\bottomrule
\end{tabular}%
}
\endgroup
\end{table}

%% file: tables/pq_cost.tex
\begin{table}[b!]
\centering
\caption{Standalone priority-queue proving cost on operation traces. Best completing results are in \textbf{bold}. Our checker is 17.1--52.0$\times$ faster than the fastest completing replay baseline.}
\label{tab:pq-cost}
\resizebox{0.75\columnwidth}{!}{%
\begin{tabular}{rcccc}
\toprule
\textbf{\#Operations} & \textbf{SP1} & \textbf{Jolt} & \textbf{Direct Plonky2} & \textbf{Ours} \\
\midrule
1,000 & 1.1 min & 25.89 s & 29.52 s & \textbf{0.95 s} \\
5,000 & 3.4 min & 1.3 min & \textcolor{oommuted}{OOM} & \textbf{3.68 s} \\
10,000 & 6.2 min & 2.4 min & \textcolor{oommuted}{OOM} & \textbf{7.33 s} \\
50,000 & 28.7 min & 9.2 min & \textcolor{oommuted}{OOM} & \textbf{32.30 s} \\
100,000 & 56.9 min & \textcolor{oommuted}{OOM} & \textcolor{oommuted}{OOM} & \textbf{65.70 s} \\
\bottomrule
\end{tabular}%
}
\end{table}

%% file: tables/rewrite-cost.tex
\begin{table}[t]
\centering
\caption{Rewrite-proof cost. Best completing results are in \textbf{bold}. \toolname reduces average proving time by 100.6$\times$ over Jolt, the fastest general-purpose baseline.}
\label{tab:rewrite-cost}
\resizebox{0.9\columnwidth}{!}{%
\begin{tabular}{clrrrr}
\toprule
\textbf{Message size} & \textbf{Method} & \makecell{\textbf{Proving}\\\textbf{time}} & \makecell{\textbf{Verification}\\\textbf{time}} & \makecell{\textbf{Proof}\\\textbf{size}} & \makecell{\textbf{Peak}\\\textbf{RSS}} \\
\midrule
\multirow{3}{*}{256 B} & SP1 & 13.7 min & 36.7 ms & 1.21 MiB & 138.2 GiB \\
 & Jolt & 4.9 min & 264.9 ms & \textbf{0.09 MiB} & 105.4 GiB \\
 & \toolname & \textbf{2.26 s} & \textbf{4.55 ms} & 0.14 MiB & \textbf{5.4 GiB} \\
\cmidrule(lr){1-6}
\multirow{3}{*}{512 B} & SP1 & 16.9 min & 36.6 ms & 1.21 MiB & 139.6 GiB \\
 & Jolt & 7.0 min & 271.3 ms & \textbf{0.09 MiB} & 175.0 GiB \\
 & \textbf{\toolname} & \textbf{2.72 s} & \textbf{4.63 ms} & 0.14 MiB & \textbf{5.3 GiB} \\
\cmidrule(lr){1-6}
\multirow{3}{*}{1 KiB} & SP1 & 17.4 min & 36.7 ms & 1.21 MiB & 141.3 GiB \\
 & Jolt & 7.0 min & 272.1 ms & \textbf{0.09 MiB} & 171.1 GiB \\
 & \textbf{\toolname} & \textbf{3.94 s} & \textbf{4.58 ms} & 0.14 MiB & \textbf{8.3 GiB} \\
\cmidrule(lr){1-6}
\multirow{3}{*}{2 KiB} & SP1 & 19.7 min & 36.8 ms & 1.21 MiB & 150.7 GiB \\
 & Jolt & 8.3 min & 271.3 ms & \textbf{0.10 MiB} & 189.8 GiB \\
 & \textbf{\toolname} & \textbf{4.23 s} & \textbf{4.57 ms} & 0.14 MiB & \textbf{8.0 GiB} \\
\cmidrule(lr){1-6}
\multirow{3}{*}{4 KiB} & SP1 & 20.2 min & 36.8 ms & 1.21 MiB & 122.1 GiB \\
 & Jolt & 8.3 min & 274.9 ms & \textbf{0.10 MiB} & 173.8 GiB \\
 & \textbf{\toolname} & \textbf{8.06 s} & \textbf{4.57 ms} & 0.14 MiB & \textbf{8.5 GiB} \\
\midrule
\multirow{3}{*}{Average} & SP1 & 17.6 min & 36.70 ms & 1.21 MiB & 138.4 GiB \\
 & Jolt & 7.1 min & 270.89 ms & \textbf{0.09 MiB} & 163.0 GiB \\
 & \textbf{\toolname} & \textbf{4.24 s} & \textbf{4.58 ms} & 0.14 MiB & \textbf{7.1 GiB} \\
\bottomrule
\end{tabular}%
}
\end{table}

%% file: tables/serving-certificate-cost.tex
\begin{table}[b!]
\centering
\caption{Serving-certificate cost. Provider generation remains below 4.2~ms, and every certificate is only 0.145~KiB.}
\label{tab:serving-certificate-cost}
\begingroup
\setlength{\tabcolsep}{5pt}
\renewcommand{\arraystretch}{1.05}
\resizebox{0.86\columnwidth}{!}{%
\begin{tabular}{rccc}
\toprule
\textbf{Message size} &
\textbf{Provider generation} &
\textbf{Client verification} &
\textbf{Certificate size} \\
\midrule
256 B & 3.46 ms & 2.46 ms & \multirow{6}{*}{0.145 KiB} \\
512 B & 3.65 ms & 2.48 ms & \\
1 KiB & 3.94 ms & 2.46 ms & \\
2 KiB & 4.15 ms & 2.47 ms & \\
4 KiB & 3.99 ms & 2.49 ms & \\
\cmidrule(lr){1-3}
\textbf{Average} & \textbf{3.84 ms} & \textbf{2.47 ms} & \\
\bottomrule
\end{tabular}%
}
\endgroup
\end{table}

%% file: section/discussion.tex
\section{Discussion and Related Work}
\label{sec:discussion-main}
\label{sec:related_work}

\noindent\textbf{Guarantee and scope.}
\toolname certifies that model selection and request rewriting conform to the semantic index and rewrite banks committed in the public context $\mathsf{ctx}$, but does not assess their policy quality. An independent auditor with access to the private specifications can validate and endorse these policies~\cite{schnabl2025attestable}, e.g., ensuring that rewrite banks exclude disallowed sponsored content as in Fig.~\ref{fig:intro-threat}. The current relation covers common additive rewrites, while model-generated rewrites can be handled by composing \toolname with verifiable inference proofs~\cite{chen2024zkml,sun2024zkllm,qu2025zkgpt}. Updates to the semantic index or rewrite banks, including new references or retrieval content, can be incorporated through an updated public context. Inference from repeated  outputs~\cite{tramer2016stealing} and hiding client requests~\cite{lu2023bumblebee,akimoto2023privformer,chowdhury2025pr,huang2026your} are complementary concerns. Further discussion is provided in Appendix~\ref{app:discussion}.

\noindent\textbf{Provider participation.}
Providers only register keys and sign responses (Tab.~\ref{tab:serving-certificate-cost}). Faithful relay assumes router--provider non-collusion, motivated in the third-party setting (Sec.~\ref{sec:bg-routing}) by independent providers' reputational incentives to sign only actual serving events~\cite{zhang2026real,c2pa2025spec}. Joint control of $\router$ and $\provider_\model$ instead requires an independent trust anchor such as remote attestation, trusted execution environment (TEE)-backed auditing, or ZK-verifiable inference~\cite{schnabl2025attestable,jin2026trusting,chen2024zkml,sun2024zkllm,qu2025zkgpt,guo2026immaculate}.

\noindent\textbf{Zero-knowledge proofs and verifiable ML.}
General-purpose ZK systems certify programs by replaying control flow~\cite{wahby2014efficient,heath20202,franzese2021constant}, whereas \toolname uses specialized routing proofs. ZK techniques for machine learning (ZKML) and LLMs (ZKLLM) certify learned computation~\cite{chen2024zkml,liu2021zkcnn,hao2024scalable,sun2024zkllm,qu2025zkgpt,weng2021mystique}, including full LLM inference and authenticated agent transcripts~\cite{gailly2026deepprove,wang2026zkagent}. Composing them with \toolname can verify selector or provider inference.

\noindent\textbf{Verifiable ANN search.}
Prior graph-ANN verification schemes disclose query-specific index information to the client~\cite{lu2024annproof,wang2024verifiable}. Recent ZK techniques prove IVF-PQ~\cite{qiu2026v3db,lin2026verirag}, whereas \toolname certifies data-dependent traversal over a committed semantic-search graph.

\noindent\textbf{LLM routing security.}
Prior work studies robust model selection~\cite{shafran2025rerouting,lin2025life,kassem2026robust,zhang2026rerouteguard,tang2026route}, malicious API intermediaries~\cite{liu2026your}, failure bounds under honest execution~\cite{uddin2026conformal}, and routing-input privacy under semi-honest secure multiparty computation (MPC)~\cite{wu2026privacy}. TEE-based gateway provenance attests routing and fallback constraints~\cite{wang2026evidence}. \toolname instead makes LLM routing verifiable even when the router is untrusted.

%% file: section/conclusion.tex
\section{Conclusion}
\label{sec:conclusion}

We presented \toolname, a verifiable LLM routing framework that establishes control-plane and data-plane integrity under an endorsed public context while hiding the router's private data. On real-world workloads, it adds only 4.08~s or 0.11~s to serving latency under per-turn or session-pinned selection, respectively.

%% file: section/appendix.tex
\section{Appendix}

\subsection{Algorithms of Semantic Search}
\label{app:hnsw-algorithms}

Algs.~\ref{alg:hnsw-upper} and~\ref{alg:hnsw-efsearch} show upper-layer greedy descent and level-$0$ \textsf{efSearch}. Both rank a larger score first and retain the first encountered node on a tie, but they encode ties at different scopes. An upper-layer scan visits the current node and then its neighbors in stored order; its strict score comparison leaves the first maximum unchanged. At level $0$, each node receives a persistent first-encounter position $\tau_u$, and the search compares the lexicographic key $(\score_u,-\tau_u)$. We use $<$ and $>$ under the applicable order in each algorithm. If the entry node lies on layer $0$, Alg.~\ref{alg:hnsw-upper} is skipped and $\EntryNode=u_{\mathsf{ent}}$.

\begin{algorithm}[t!]
\caption{Upper-layer greedy descent}
\label{alg:hnsw-upper}
\begin{algorithmic}[1]
\Require entry node $u_{\mathsf{ent}}$ at top layer $L_{\max}\geq1$, query embedding $e_q$
\State $\mathsf{cur} \gets u_{\mathsf{ent}}$
\For{$\ell = L_{\max},\, L_{\max}-1,\, \dots,\, 1$}
    \While{\textbf{true}}
        \State $\mathsf{best} \gets \mathsf{cur}$
        \State $\mathsf{stop} \gets 1$
               \Comment{$\mathsf{stop}=1$ iff no improvement found}
        \For{$u$ in the stored order of $\mathsf{Nbr}_\ell(\mathsf{cur})$}
            \If{$u>\mathsf{best}$}
                \State $\mathsf{best} \gets u$
                \State $\mathsf{stop} \gets 0$
            \EndIf
        \EndFor
        \If{$\mathsf{stop} = 1$}
            \State \textbf{break}
                   \Comment{no improvement on layer $\ell$, descend to layer $\ell-1$}
        \Else
            \State $\mathsf{cur} \gets \mathsf{best}$
                   \Comment{advance to best neighbor and continue on layer $\ell$}
        \EndIf
    \EndWhile
\EndFor
\State \Return $\mathsf{cur}$
       \Comment{level-$0$ entry node: endpoint of upper-layer greedy descent}
\end{algorithmic}
\end{algorithm}

\begin{algorithm}[t!]
\caption{Level-0 \textsf{efSearch}}
\label{alg:hnsw-efsearch}
\begin{algorithmic}[1]
\Require level-$0$ entry node $\EntryNode$, query embedding $e_q$, beam width $\EfCap$, output size $k$ with $1\leq k\leq\EfCap$
\State initialize visited set $V \gets \{\EntryNode\}$
\State initialize candidate queue $\CandQ \gets \{\EntryNode\}$
       \Comment{max-heap by search key, \textsf{PopMax} returns the largest-key element}
\State initialize working queue $\WorkQ \gets \{\EntryNode\}$
       \Comment{min-heap by search key, keeps the top-$\EfCap$ nodes}
\While{$\CandQ \neq \emptyset$}
    \State $\mathsf{cur} \gets \CandQ.\mathsf{PopMax}()$
    \If{$|\WorkQ|=\EfCap$ \textbf{and} $\mathsf{cur}<\WorkQ.\mathsf{Worst}()$}
        \State \textbf{break} \Comment{best candidate is below working threshold, stop}
    \EndIf
    \For{$u$ in the stored order of $\mathsf{Nbr}_0(\mathsf{cur})$}
        \If{$u \notin V$}
            \State $V \gets V \cup \{u\}$
            \If{$|\WorkQ|<\EfCap$ \textbf{or} $u>\WorkQ.\mathsf{Worst}()$}

                \State $\WorkQ.\mathsf{Push}(u)$
                \State $\CandQ.\mathsf{Push}(u)$
                \If{$|\WorkQ| > \EfCap$}
                    \State $\WorkQ.\mathsf{PopMin}()$
                          \Comment{evict worst element to maintain capacity $\EfCap$}
                \EndIf
            \EndIf
        \EndIf
    \EndFor
\EndWhile
\State \Return the best $k$ nodes in $\WorkQ$
\end{algorithmic}
\end{algorithm}

\subsection{Detailed ZK Protocols}

\subsubsection{Request Hash Instantiation}
\label{app:hash-instantiation}

We instantiate the structured request hash $\ReqHash$ with Poseidon~\cite{grassi2021poseidon}. It hashes content bytes, binds each content digest to its message role and length, folds the message digests in request order, and binds the message count in a final hash. The ordered fold also supports composition over consecutive message ranges. {Every evaluation of $H_{\mathsf{zk}}$ has a fixed input length, because Poseidon without padding is collision resistant only on fixed-length inputs. The request-hash contexts below are separated by distinct tags: a tag is the first input element in sponge mode, and for the two-digest fold step it occupies a capacity element of a single Poseidon permutation. We assume that $H_{\mathsf{zk}}$ is collision resistant in each of these contexts.} Roles, lengths, and bytes use fixed, injective field encodings.

\noindent\textbf{Content digest.}
For a byte string $c=(a_1,\ldots,a_{N_c})$, set $b_0(c)=H_{\mathsf{zk}}(\mathsf{tag}_{\mathsf{byte\text{-}init}}\|N_c\|0)$ and $b_j(c)=H_{\mathsf{zk}}(\mathsf{tag}_{\mathsf{byte\text{-}step}}\|b_{j-1}(c)\|1\|a_j)$ for $j\in[N_c]$, and define $b(c)=b_{N_c}(c)$.

\noindent\textbf{Message digest.}
For a message $y=(\mathsf{role},c)$, define $m(y)=H_{\mathsf{zk}}(\mathsf{tag}_{\mathsf{msg}}\|\mathsf{role}\||c|\|b(c))$.

\noindent\textbf{Ordered fold.}
Let $\mathsf{IV}$ be the protocol-fixed initial fold value. For a digest sequence $\mu=(\mu_1,\ldots,\mu_{N_\mu})$ and starting fold value $s_{\mathsf{in}}$, set $s_0=s_{\mathsf{in}}$ and $s_j=H_{\mathsf{zk}}(\mathsf{tag}_{\mathsf{req}}\|s_{j-1}\|\mu_j)$ for $j\in[N_\mu]$, and define $\ReqFold(s_{\mathsf{in}};\mu)=s_{N_\mu}$. A full request starts from $s_{\mathsf{in}}=\mathsf{IV}$. When consecutive message ranges are folded separately, each later range starts from the preceding range's output.

\noindent\textbf{Final request hash.}
For a request $R=[y_1,\ldots,y_{N_R}]$, define
\begin{equation}
\begin{aligned}
\mu_R&=\bigl(m(y_1),\ldots,m(y_{N_R})\bigr),\\[-2pt]
\ReqHash(R)&=H_{\mathsf{zk}}\bigl(\mathsf{tag}_{\mathsf{req\text{-}final}}\|N_R\|\ReqFold(\mathsf{IV};\mu_R)\bigr).
\end{aligned}
\label{eq:structured-request-hash}
\end{equation}

\noindent\textbf{Rewrite-bank digests.}
The rewrite banks reuse these digest functions. The canonical message digest table $\MsgDigestTable$ contains $m(g)$ for every complete message $g\in\gapset$, and the canonical fragment digest table $\FragDigestTable$ contains $b(f)$ for every fragment $f\in\fragset$. These digest tables instantiate the rewrite-bank commitments of Sec.~\ref{sec:method-initial}. A committed lookup against $\MsgCommit$ authenticates $m(g)$, which the rewrite circuit places directly in the ordered message-digest sequence. For a selected fragment, the circuit recomputes $b(f)$ from its bytes, authenticates it against $\FragCommit$, and reuses the bytes in $m(\mathsf{wrap}(x_i))$.

\begin{lemma}[Request-hash collision resistance]
\label{lem:reqhash-cr}
If $H_{\mathsf{zk}}$ is collision resistant {in each context above}, then $\ReqHash$ is collision resistant for requests.
\end{lemma}

\begin{proof}
Suppose $\ReqHash(R)=\ReqHash(R')$. Unless some compared pair of $H_{\mathsf{zk}}$ calls already forms a collision, equality at the final call fixes the message count and final fold value. Working backward through the fold fixes every message digest, and working backward through each message and content hash fixes every role, content length, and byte. Hence $R=R'$, so distinct requests with the same hash yield a collision in $H_{\mathsf{zk}}$.
\end{proof}

The same backward argument, using the length-bound initialization $b_0(\cdot)$, also shows that the fragment digest $b(\cdot)$ is collision resistant.

\subsubsection{Table Digests}

The approved profile fixes each table's leaf encoding, hash functions, binary-tree height and padding, yielding $d_{\mathcal T}=\mathsf{MerkleRoot}(\mathcal T)$~\cite{merkle1989certified,miller2014authenticated}. {Every opening uses an authentication path of exactly the fixed height, so a leaf is never interpreted as an internal node, and comparing two distinct fixed trees with the same root yields a hash collision. The reference and graph tables apply Poseidon to the concatenated child digests and duplicate the last node of an odd level, so every position reachable by a fixed-height path holds a copy of an endorsed row. The rewrite-bank digest tables use tagged compression nodes and pad to a full tree with zero leaves; the profile records the actual bank lengths, and enabled lookup indices are constrained to these ranges, so padding leaves are never opened.}

The private per-request score table instead uses an internal full-table digest. Let $\ScoreTableCap=|\dtable|$ be its public capacity and $\overline{\dtable}$ its $\ScoreTableCap$-row serialization, with active rows $(u,\score_u,\mathsf{mid}(u))$ in strictly increasing node-identifier order. Every remaining row equals the protocol-fixed inactive encoding $\bot_{\mathsf{score}}$, which is disjoint from every real score-table row. Define
\begin{equation}
h_{\dtable}=H_{\mathsf{zk}}(\mathsf{tag}_{\mathsf{score}}\|\ScoreTableCap\|\overline{\dtable}).
\label{eq:flat-score-hash}
\end{equation}
Each model-selection subproof recomputes this digest from its private table copy, and aggregation requires the same $h_{\dtable}$ in all three subproof statements. Collision resistance thereby binds the copies to one table, while the aggregation proof hides both the digest and the table.

\subsubsection{Commitment Instantiation}
\label{app:committed-lookup}

We instantiate every commitment with $H_{\mathsf{zk}}$. For a committed object $A$, let $d_A$ be its fixed-length digest: $d_A=\mathsf{MerkleRoot}(A)$ for a Merkle-authenticated table and $d_A=A$ when $A$ is already a digest:
\begin{equation}
\Commit(A;\rho)=H_{\mathsf{zk}}(\mathsf{tag}_{\mathsf{com}}\|d_A\|\rho).
\label{eq:commitment-instantiation}
\end{equation}
We encode $(\mathsf{tag}_{\mathsf{com}},d_A,\rho)$ injectively and sample $\rho$ as four fresh, independent, secret field elements with total min-entropy $h_\rho\geq252$ bits. Lemma~\ref{lem:reqhash-cr}, the fixed Merkle encoding and collision resistance make Eq.~\eqref{eq:commitment-instantiation} computationally binding. {Modeling $H_{\mathsf{zk}}$ as a random oracle, the commitment is computationally hiding because $\rho$ is fresh and has at least $252$ bits of min-entropy.} \toolname uses this construction for semantic tables, rewrite-bank digest tables and requests.

For queried entries $\mathcal Q$ of a committed table $\mathcal T$, let $c_{\mathcal T}=\Commit(\mathcal T;\rho_{\mathcal T})$. The witness $\mathcal O_{\mathcal Q}$ contains $\rho_{\mathcal T}$, each queried row's leaf index, and its authentication path. The circuit enforces $\LUT(\mathcal Q,c_{\mathcal T};\mathcal O_{\mathcal Q})=1$, which checks that every queried leaf has the requested row and reconstructs the same root, and that Eq.~\eqref{eq:commitment-instantiation} yields $c_{\mathcal T}$.

{\noindent\textbf{Private lookups.} The monotone-stack checker looks up $\mathcal Q_{\mathsf{parent}}$ in $\ParentTab$, and the priority-queue checker looks up $\mathcal Q_{\mathsf{parentOut}}$ in $\RemovalTab$. Neither table is committed, and the circuits realize both lookups with multiset equalities. Every row of $\ParentTab$ is created by exactly one push, so one multiset equality matches the push queries with the rows. A pop query $(\PoppedItem{i}{j},\StackTop{i}{j})$ asks for the row created by the push of $\PoppedItem{i}{j}$. The circuit lists $\FinStack$ from the top together with each item's parent and enforces that the rows of $\ParentTab$, as edges $(u,\ParentOf{u})$, equal as a multiset the pop queries together with these final edges. This equality implies the pop lookups, and its projection to first components is Eq.~\eqref{eq:stk-accounting}; an honest execution satisfies it because every pushed item is either popped once or remains in $\FinStack$. Finally, each $\OutSeq$ item carries its removal position as a payload of its stack row, so the parent row returned for $\OutItem{r}$ also returns $\ParentOutPos{r}$, which answers the lookup into $\RemovalTab$. Lemma~\ref{lem:multiset-error} bounds the error of all these multiset equalities.}

\subsubsection{Parallel Dot-product Gate}
\label{app:score-opening-gate}

The score-opening subproof evaluates many high-dimensional node scores against one query. Our parallel dot-product gate shares the query coordinates across several node scores per row, reducing score-computation rows.

\noindent{\bf Score decomposition.} {Our implementation instantiates the deterministic score function from Sec.~\ref{sec:bg-routing} with standard 16-bit affine quantization~\cite{jacob2018quantization}. The public parameters $\mathsf{pp}_{\mathsf{sel}}$ fix the quantizer $\mathsf{Q}_{\mathsf{sel}}$, its scale and zero point $z_q$, and the score bias. The client computes $q=\mathsf{Q}_{\mathsf{sel}}(e_q)$ and binds its digest in the verified statement. The score-opening circuit authenticates this $q$. Each endorsed reference row stores $v_u=\mathsf{Q}_{\mathsf{sel}}(\encoder(R_u^{\mathsf{ref}}))$.} Let $D$ be the embedding dimension. The certified score is
\begin{equation}
\label{eq:centered-score-opening}
\begin{aligned}
\score_u
&=\sum_{i=0}^{D-1}(q_i-z_q)(v_{u,i}-z_q)+\mathsf{bias}\\
&=\sum_{i=0}^{D-1}(q_i-z_q)v_{u,i}+\left(Dz_q^2-z_q\sum_{i=0}^{D-1}q_i+\mathsf{bias}\right).
\end{aligned}
\end{equation}
Because the query $q$ and zero point $z_q$ are the same for every node, the circuit computes the centered query coordinates $\widehat q_i=q_i-z_q$ and the query-only correction $\Delta_q=Dz_q^2-z_q\sum_iq_i+\mathsf{bias}$ once. Eq.~\eqref{eq:centered-score-opening} then becomes $\score_u=\sum_i\widehat q_iv_{u,i}+\Delta_q$ for every node $u$.

\noindent{\bf Gate organization.} Each row processes $B_{\mathsf{coord}}$ coordinates for $N_{\mathsf{lane}}$ nodes, with one lane carrying each node's running sum. It stores the $B_{\mathsf{coord}}$ query coordinates once and, per lane, the $B_{\mathsf{coord}}$ node coordinates and incoming and outgoing sums. Thus it uses $B_{\mathsf{coord}}+N_{\mathsf{lane}}(B_{\mathsf{coord}}+2)$ routed wires and $N_{\mathsf{lane}}$ quadratic equations for $N_{\mathsf{lane}}B_{\mathsf{coord}}$ products.

The Plonky2 recursion configuration provides $80$ routed wires. Maximizing $N_{\mathsf{lane}}B_{\mathsf{coord}}$ subject to $B_{\mathsf{coord}}+N_{\mathsf{lane}}(B_{\mathsf{coord}}+2)\leq80$ gives $N_{\mathsf{lane}}B_{\mathsf{coord}}=56$. Among the maximizing layouts, we choose $(N_{\mathsf{lane}},B_{\mathsf{coord}})=(4,14)$ because it uses the fewest lane constraints. One row therefore processes four nodes over $14$ coordinates using $78$ wires and four quadratic equations.

The circuit pads $\dtable$ to its public size and groups its rows in fours. Let $\CoordBlockCount=\lceil D/14\rceil$. For node group $g$, lane $a\in\{0,\ldots,3\}$ contains node $u_{g,a}$. For each dimension block $b\in\{0,\ldots,\CoordBlockCount-1\}$, the gate enforces
\begin{equation}
\label{eq:four-lane-dot-gate}
\mathsf{acc}_{\mathsf{out},g,b,a}=\mathsf{acc}_{\mathsf{in},g,b,a}+\sum_{j=0}^{13}\widehat q_{14b+j}v_{u_{g,a},14b+j}.
\end{equation}

\noindent{\bf Accumulator chaining.} Consecutive gate rows process successive coordinate groups. Each lane starts with $\mathsf{acc}_{\mathsf{in},g,0,a}=\Delta_q$, and the circuit enforces $\mathsf{acc}_{\mathsf{in},g,b+1,a}=\mathsf{acc}_{\mathsf{out},g,b,a}$ for $0\leq b<\CoordBlockCount-1$.
After the final group, the circuit enforces $\chi_{g,a}(\score_{u_{g,a}}-\mathsf{acc}_{\mathsf{out},g,\CoordBlockCount-1,a})=0$, where $\chi_{g,a}\in\{0,1\}$ is the row's activity bit. An inactive row is fixed to its padding value and requires no reference-table opening. Coordinates beyond $D$ are constrained to zero. {The circuit range-checks $q_i,z_q\in[0,2^{16}-1]$ and the resulting scores to 48 bits. Endorsement checks this coordinate range for the full reference table. Authenticated openings transfer it to each active $v_u$. Setup requires $D(2^{16}-1)^2<2^{47}$ and fixes $\mathsf{bias}=2^{47}$, so scores lie in $[0,2^{48})$ without base-field wraparound and later comparisons use integer order.}

\noindent{\bf Cost.} For the padded score table $\dtable$, the dot-product body uses $\mathsf{Rows}_{\mathsf{dot}}(|\dtable|,D)=\lceil|\dtable|/4\rceil\lceil D/14\rceil$ rows; the query terms and the final active-row equalities add $O(D+|\dtable|)$ arithmetic operations. Under the same $80$-wire configuration, a generic Plonky2 arithmetic gate checks $20$ products per row, compared with $56$ for our gate. At $|\dtable|=1024$ and $D=1536$, the generic realization uses $78{,}848$ rows for the products, while our dot-product body uses $28{,}160$ rows, about a $64\%$ reduction.

\noindent{\bf Gate correctness.} For an active lane, unrolling Eq.~\eqref{eq:four-lane-dot-gate} from its initial value $\Delta_q$ and applying the chaining constraints gives $\mathsf{acc}_{\mathsf{out},g,\CoordBlockCount-1,a}=\Delta_q+\sum_{i=0}^{D-1}\widehat q_i v_{u_{g,a},i}$. The final active-row constraint therefore gives $\score_{u_{g,a}}=\sum_i\widehat q_i v_{u_{g,a},i}+\Delta_q$, which is Eq.~\eqref{eq:centered-score-opening}. The circuit authenticates the client-computed $q$, constrains $\widehat q_i=q_i-z_q$ and $\Delta_q=Dz_q^2-z_q\sum_iq_i+\mathsf{bias}$, and authenticates each active $v_u$. Thus sharing the query coordinates across lanes changes only the gate layout, not the certified score relation.

\subsubsection{Level-0 \textsf{efSearch} Subproof Details}
\label{app:wc-instantiation}

This subsection gives the trace layout and constraints by which $\cef$ certifies Alg.~\ref{alg:hnsw-efsearch}, organized around the three conditions in Sec.~\ref{sec:method-control-ef}: queue-trace validity, cross-queue consistency, and the stopping rule and output.

\noindent\textbf{Shared trace layout.}
Let $N_{\mathsf{round}}^{0,\max}$ be the public bound on candidate-expansion rounds and let $B_0$ be the public level-$0$ neighbor-row width. The padded trace length is $|\WorkUpdTrace|=T=1+N_{\mathsf{round}}^{0,\max}B_0$, and the maximum reserved position in $\CandQ$'s operation order is $P=1+N_{\mathsf{round}}^{0,\max}(B_0+1)$.
Let $\CandPopLen\leq N_{\mathsf{round}}^{0,\max}$ be the number of certified candidate pops, with $\CandPopItem{a}$ and $\CandPopTime{a}$ denoting the item and logical position of pop $a$. As defined in Sec.~\ref{sec:method-control-ef}, the padded $\WorkUpdTrace$ has rows $(n_i,\tau_i,\mathsf{En}_i,\score_i,\AdmitBit{i})$ indexed by $i\in[T]$.
Row $i$ is aligned with the node $\WorkEvictItem{i}$ popped by its potential $\WorkQ$ update. The sequence starts with $n_1=\EntryNode$ and continues with the padded level-$0$ neighbor rows exposed by $\CandPopItem{1},\ldots,\CandPopItem{\CandPopLen}$ in this order. Inactive round blocks set $n_i=\bot_{\mathsf{pad}}$, a protocol-fixed disabled node encoding disjoint from every real node identifier. Repeated nodes and padding remain in the sequence.

For $a\in[N_{\mathsf{round}}^{0,\max}]$ and $p\in[B_0]$, let $i_{a,p}=1+(a-1)B_0+p$ be the row index for neighbor position $p$ in round $a$. The circuit assigns
\begin{equation}
\begin{aligned}
\tau_1&=1,\\[-2pt]
\CandPopTime{a}&=2+(a-1)(B_0+1)
  &&(a\in[\CandPopLen]),\\[-2pt]
\tau_{i_{a,p}}&=2+(a-1)(B_0+1)+p.
\end{aligned}
\label{eq:efsearch-fixed-positions}
\end{equation}
Here $\tau_i$ is the position in $\CandQ$'s operation order reserved for row $i$. Only an admitted row uses it for an actual push. Rejected, repeated, padded, and inactive rows leave holes. All positions are circuit constants, lie in $[P]$, and strictly increase with row order. Together with the fixed $\WorkQ$ positions in Eq.~\eqref{eq:workq-fixed-positions}, this makes $\mathsf{TraceConsistency}$ hold by construction for both queues.

\noindent\textbf{(i) Queue-trace validity.}
\uline{\emph{Priority encodings.}} Let $\score_{\max}^{\mathsf{enc}}$ be a public upper bound on every encoded score. For every real queue item $u$, let $i$ be its unique enabled row under the first-visit check below. Set $\tau_u:=\tau_i$ and $\mathsf{key}(u)=(\score_u,-\tau_u)$.
Because the reserved positions strictly increase with row order, this key favors the first encounter on equal scores. The candidate checker packs this order using $P$ as the position radix:
\begin{equation}
\kappa_C(u)=(\score_{\max}^{\mathsf{enc}}-\score_u)P+(\tau_u-1),
\label{eq:efsearch-packed-c-key}
\end{equation}
so its minimum is the largest-score, first-encountered candidate. For distinct sentinels $\DummyTok^{(j)}$, $j\in[\EfCap]$, representing unused capacity, the working checker uses
\begin{equation}
\begin{aligned}
\kappa_W(\DummyTok^{(j)})&=j-1,\\[-2pt]
\kappa_W(u)&=\EfCap+\score_uP+(P-\tau_u),
\end{aligned}
\label{eq:efsearch-packed-w-key}
\end{equation}
so every unused sentinel has smaller $\kappa_W$ than every real item and the minimum real item is the smallest-score, latest-encountered node. The instantiated priority relations are $u\CandHigher v\Longleftrightarrow\kappa_C(u)<\kappa_C(v)$ and $u\WorkHigher v\Longleftrightarrow\kappa_W(u)<\kappa_W(v)$.
Setup requires $2^{48}P+\EfCap\leq2^{62}$, so both packed keys are integers below $2^{62}$ and range-checked comparisons preserve their integer order without field wraparound. Distinct reserved positions and distinct sentinel ordinals make every queue item's packed key unique.

Every real queue item carries $\widehat{u}=(u,\score_u,\tau_u)$. We suppress the score and position in node notation when clear. The $\CandQ$ relation checks the position-annotated projections $\CandPushTrace$ and $\CandPopTrace$. Its public push, pop, and final-sequence capacities are $T$, $N_{\mathsf{round}}^{0,\max}$, and $T$, respectively.

\uline{\emph{Aligned $\WorkQ$ trace.}} The working trace first inserts the $\EfCap$ sentinels $\DummyTok^{(1)},\ldots,\DummyTok^{(\EfCap)}$. Real identifiers, sentinels, $\NoParent$, and $\bot_{\mathsf{pad}}$ have disjoint encodings, with $\bot_{\mathsf{pad}}$ disabled. Under $\WorkHigher$, the generic checker pops the $\kappa_W$-minimum.

Every row retains $n_i$ and its reserved position $\tau_i$. If $\mathsf{En}_i=0$, its remaining fields and aligned queue targets are fixed zeros and it contributes no operation. Let $I_{\mathsf{en}}=\{i\in[T]:\mathsf{En}_i=1\}$. The enabled item projection $\WorkPushTrace=\bigl((n_i,\score_i,\tau_i)\bigr)_{i\in I_{\mathsf{en}}}$ is $\WorkQ$'s real-node push sequence and hence its logical update sequence. After inserting each sentinel exactly once, the $\WorkQ$ relation checks the complete position-annotated projections
\begin{equation}
\begin{aligned}
\PushTrace_{\WorkQ}
&=((j,\DummyTok^{(j)}))_{j\in[\EfCap]}
  \mathbin\Vert
  ((\EfCap+2i-1,n_i))_{i\in I_{\mathsf{en}}},\\[-2pt]
\PopTrace_{\WorkQ}
&=((\EfCap+2i,\WorkEvictItem{i}))_{i\in I_{\mathsf{en}}}.
\end{aligned}
\label{eq:workq-fixed-positions}
\end{equation}
The checker capacities are fixed in $\mathsf{pp}_{\mathsf{sel}}$.

\uline{\emph{Inlined queue checkers.}} The circuit instantiates the generic checker from Sec.~\ref{sec:zk-priority-queue} for both $\CandQ$ and $\WorkQ$, with respective parameter--witness pairs $(x_{\mathsf{pq}}^{\CandQ},w_{\mathsf{pq}}^{\CandQ})$ and $(x_{\mathsf{pq}}^{\WorkQ},w_{\mathsf{pq}}^{\WorkQ})$. Denote their certified final sequences by $\mathsf{Final}_{\CandQ}$ and $\mathsf{Final}_{\WorkQ}$. These sequences and the annotated traces are shared with $\cef$, and the circuit enforces
\begin{equation}
\begin{aligned}
\Cpq(x_{\mathsf{pq}}^{\CandQ};w_{\mathsf{pq}}^{\CandQ})&=1,\\[-2pt]
\Cpq(x_{\mathsf{pq}}^{\WorkQ};w_{\mathsf{pq}}^{\WorkQ})&=1.
\end{aligned}
\label{eq:efsearch-queue-checkers}
\end{equation}
Each real entry contains its identifier, score-table-bound score, and reserved position, with priority computed by Eq.~\eqref{eq:efsearch-packed-c-key} or Eq.~\eqref{eq:efsearch-packed-w-key}.

\noindent\textbf{(ii) Cross-queue consistency.}
First, \emph{node-order authentication} binds the node sequence in $\WorkUpdTrace$ to the $\CandQ$ pop trace and the committed graph. For every $a\in[\CandPopLen]$, the circuit authenticates the padded row $(0,\CandPopItem{a},\mathsf{Pad}_{B_0}(\mathsf{Nbr}_0(\CandPopItem{a})))$ against $\GraphCommit$ and equates its $p$-th slot with $n_{i_{a,p}}$. Together with $n_1=\EntryNode$, this makes every real node in $\WorkUpdTrace$ either the initial entry or a neighbor of a node popped from $\CandQ$.

Second, the \emph{first-visit filter} sorts a copy of $(n_i,i,\mathsf{En}_i)_{i\in[T]}$ by node and original row index and binds it by multiset equality. Adjacent-row checks mark only the earliest row of each real node, enforcing, for every $i\in[T]$,
\begin{equation}
\mathsf{En}_i=1
\Longleftrightarrow
n_i\text{ is real}\ \land\ \forall j<i,\ n_j\ne n_i.
\label{eq:first-visit-stable-filter}
\end{equation}

Whenever a row is enabled, the circuit binds $(n_i,\score_i)$ to the node-and-score projection of the unique active row for $n_i$ in the level-$0$ subproof's padded score table and uses $(n_i,\score_i,\tau_i)$ for its $\WorkQ$ update. Since $n_1=\EntryNode$ is real and has no preceding row, Eq.~\eqref{eq:first-visit-stable-filter} enables the first row and supplies the initial entry update. Thus, $\WorkPushTrace$ is exactly the first-visit item sequence in row order.

Third, the \emph{admission filter} links $\WorkQ$ admissions to candidate pushes. For every row, the circuit enforces
\begin{equation}
\begin{aligned}
\AdmitBit{i}=1
&\Longleftrightarrow
\mathsf{En}_i=1\ \land\ \WorkEvictItem{i}\ne n_i
&& (i\in[T]),\\[-2pt]
\CandPushTrace
&=\bigl((\tau_i,n_i):\AdmitBit{i}=1\bigr)_{i\in[T]}.
\end{aligned}
\label{eq:efsearch-direct-projection}
\end{equation}
This reuses exactly the admitted rows in $\CandPushTrace$ with the same node, score-table-bound score, and reserved position. Let $J_a=\{i_{a,p}:p\in[B_0]\}$ be the row indices assigned to round $a$. Equations~\eqref{eq:efsearch-fixed-positions} and~\eqref{eq:efsearch-direct-projection} place its admitted first visits in the $B_0$ positions following the round-$a$ pop.

\noindent\textbf{(iii) Stopping rule.}
The entry update precedes all expansion rounds. For round $a$, if the round has an admitted row, let $i_a^\star=\min\{i\in J_a:\AdmitBit{i}=1\}$. The circuit derives the one-hot first-admitted flag in $\WorkUpdTrace$ row order and uses it to select $\WorkEvictItem{i_a^\star}$. It enforces
\begin{equation}
\WorkEvictItem{i_a^\star}\text{ is real}
\quad\Longrightarrow\quad
\mathsf{key}(\CandPopItem{a})\geq\mathsf{key}(\WorkEvictItem{i_a^\star}).
\label{eq:efsearch-round-threshold}
\end{equation}
When such a row exists, rows preceding $i_a^\star$ in $J_a$ are rejected and leave $\WorkQ$ unchanged. By the capacity invariant, the first admission evicts a real item exactly when the round starts with $\EfCap$ real nodes, and that item is the round-start minimum. Equation~\eqref{eq:efsearch-round-threshold} therefore checks the required threshold using this eviction. A round with no admitted row leaves $\WorkQ$ unchanged.

Let $C_{\mathsf{fin}}$ and $W_{\mathsf{fin}}$ be the real-node sets in $\mathsf{Final}_{\CandQ}$ and $\mathsf{Final}_{\WorkQ}$. The terminal check enforces
\begin{equation}
\begin{aligned}
C_{\mathsf{fin}}=\emptyset\quad\lor\quad\bigl(&
|W_{\mathsf{fin}}|=\EfCap\\[-2pt]
&{}\land\max_{u\in C_{\mathsf{fin}}}\mathsf{key}(u)
<\min_{u\in W_{\mathsf{fin}}}\mathsf{key}(u)\bigr).
\end{aligned}
\label{eq:efsearch-terminal}
\end{equation}

The sequence $\mathsf{Final}_{\WorkQ}$ is certified by the $\WorkQ$ checker. Because every sentinel has smaller $\kappa_W$ than every real node, the sentinels form a prefix. Let $(u_1^{\WorkQ,\mathsf{fin}},\ldots,u_{N_W}^{\WorkQ,\mathsf{fin}})$ be the real-node suffix. The final-sequence order gives $\mathsf{key}(u_1^{\WorkQ,\mathsf{fin}})<\cdots<\mathsf{key}(u_{N_W}^{\WorkQ,\mathsf{fin}})$. The circuit requires $N_W\geq k$ and enforces
\begin{equation}
F^{\mathsf{topk}}_{\WorkQ}=\{u_{N_W-k+1}^{\WorkQ,\mathsf{fin}},\ldots,u_{N_W}^{\WorkQ,\mathsf{fin}}\}.
\label{eq:efsearch-topk-output}
\end{equation}
Thus $F^{\mathsf{topk}}_{\WorkQ}$ contains exactly the $k$ largest-key reference indices retained in the final $\WorkQ$. The circuit also enforces $\mathsf{mid}(u)=\model$ for some $u\in F^{\mathsf{topk}}_{\WorkQ}$.

\subsubsection{Model-selection Aggregation Protocol}
\label{app:zk-control-agg-protocol}

Fig.~\ref{fig:zk-control-plane-aggregation} gives the model-selection protocol, which recursively aggregates the subproofs. The public statement is $x_{\mathsf{sel}}=(e_q,\model,c_{\mathsf{idx}},\mathsf{pp}_{\mathsf{sel}})${, where $\mathsf{pp}_{\mathsf{sel}}$ includes the entry node $u_{\mathsf{ent}}$, its layer $L_{\max}$, $\EfCap$, and $k$}. Each subproof circuit privately receives its own copy of $\overline{\dtable}$ and exposes its digest $h_{\dtable}$ only to the recursive parent. Let $x_{\mathsf{shared}}=(e_q,c_{\mathsf{idx}},\mathsf{pp}_{\mathsf{sel}},h_{\dtable})$. The three subproof statements are
\begin{equation}
\begin{aligned}
x_{\mathsf{score}}=x_{\mathsf{shared}},\\[-2pt]
x_{\mathsf{upper}}=(x_{\mathsf{shared}},{\EntryNode}),\\[-2pt]
x_{\mathsf{ef}}=(x_{\mathsf{shared}},{\EntryNode},\model).
\end{aligned}
\label{eq:control-subproof-interfaces}
\end{equation}

For recursive verification, the aggregation circuit is the parent and the proofs it verifies are its children. At circuit construction, the parent fixes each child's approved verifier data: \texttt{circuit\_digest} identifies the child circuit, and \texttt{constants\_sigmas\_cap} authenticates its constant and permutation polynomials. The parent uses these fixed values to verify the child proof, binding it to the intended relation.

\begin{figure}[t!]
\centering
\begingroup
\small
\setlength{\fboxsep}{3pt}
\setlength{\abovedisplayskip}{2pt}
\setlength{\belowdisplayskip}{2pt}
\fbox{
\begin{minipage}{0.95\linewidth}
\raggedright
\textbf{Public statement:} $x_{\mathsf{sel}}=(e_q,\model,c_{\mathsf{idx}},\mathsf{pp}_{\mathsf{sel}})$.

\textbf{Subproof witnesses:} Each of $w_{\mathsf{score}}$, $w_{\mathsf{upper}}$, and $w_{\mathsf{ef}}$ contains a private copy of $\overline{\dtable}$ and the openings and execution data required by its relation; $w_{\mathsf{ef}}$ also contains the two inlined queue-checker witnesses and the retained-row model-membership witness.

\textbf{Prover:}
\begin{enumerate}[leftmargin=1.25em,itemsep=1pt,topsep=2pt,parsep=0pt]
  \item \textbf{Subproofs.} Generate
  \[
  \begin{aligned}
  \pi_{\mathsf{score}}&\leftarrow\ZKProve(\cscore;x_{\mathsf{score}};w_{\mathsf{score}}),\\[-2pt]
  \pi_{\mathsf{upper}}&\leftarrow\ZKProve(\cupper;x_{\mathsf{upper}};w_{\mathsf{upper}}),\\[-2pt]
  \pi_{\mathsf{ef}}&\leftarrow\ZKProve(\cef;x_{\mathsf{ef}};w_{\mathsf{ef}}).
  \end{aligned}
  \]

  \item \textbf{Aggregation.} Set
  \[
  w_{\mathsf{sel}}^{\mathsf{agg}}=(h_{\dtable},\EntryNode,\pi_{\mathsf{score}},\pi_{\mathsf{upper}},\pi_{\mathsf{ef}}).
  \]
  The aggregation circuit reconstructs the three statements from $x_{\mathsf{sel}}$ and $w_{\mathsf{sel}}^{\mathsf{agg}}$, verifies exactly the three subproofs above, enforces Eq.~\eqref{eq:control-subproof-interfaces}, and outputs
  \[
  \begin{aligned}
  \pi_{\mathsf{sel}}\leftarrow\ZKProve\bigl(\mathsf{C}_{\mathsf{sel}};x_{\mathsf{sel}};w_{\mathsf{sel}}^{\mathsf{agg}}\bigr).
  \end{aligned}
  \]
\end{enumerate}

\textbf{Verifier:} $\ZKVerify(\mathsf{C}_{\mathsf{sel}};x_{\mathsf{sel}};\pi_{\mathsf{sel}})\in\{\mathtt{accept},\mathtt{reject}\}$.
\end{minipage}
}
\caption{Protocol for the model-selection aggregation proof.}
\label{fig:zk-control-plane-aggregation}
\endgroup
\end{figure}

The score-opening subproof certifies $\dtable$ by enforcing strictly increasing node-identifier order and fixed padding, authenticating every active row against $\RefCommit$, and computing its score. This performs each active row's high-dimensional score computation and reference-table opening once.

The traversal subproofs then reuse the certified scores through lower-cost table-consistency checks. Each traversal subproof represents every consumed score use by a tuple $(u,\score_u,\mathsf{mid}(u))$. It binds a sorted copy $((u_r,\score_r,\mathsf{mid}_r))_r$ to all such tuples by multiset equality. For every adjacent pair of real-node tuples, the circuit enforces $u_r=u_{r-1}\Longrightarrow(\score_r,\mathsf{mid}_r)=(\score_{r-1},\mathsf{mid}_{r-1})$.
Thus, all uses of the same node have the same score and model identifier. On each traversal's sorted copy, adjacent-row constraints mark the first tuple of every real-$u$ run, retaining one representative solely for table accounting. Write $\mathsf{Active}(\dtable)$ for the multiset of active rows. Let $\mathcal U_j$ be the multiset of these representatives and let $\mathcal U_j^{\mathsf{rest}}$ contain the remaining active rows. The circuit enforces
\begin{equation}
\mathcal U_j\uplus\mathcal U_j^{\mathsf{rest}}=_{\mathsf{ms}}\mathsf{Active}(\dtable)
\qquad(j\in\{\mathsf{upper},\mathsf{ef}\}).
\label{eq:subproof-score-table-consistency}
\end{equation}
Eq.~\eqref{eq:subproof-score-table-consistency} matches each representative to its unique active row in $\dtable$, with $\mathcal U_j^{\mathsf{rest}}$ accounting for the other rows. The repeated-use checks extend this binding to every occurrence. The level-$0$ subproof also binds its final top-$k$ references to these tuples and enforces $\mathsf{mid}(u)=\model$ for at least one $u\in F^{\mathsf{topk}}_{\WorkQ}$.

\subsubsection{Rewrite Aggregation Protocol}
\label{app:zk-rewrite-protocol}

Let $N_{\mathsf{cmsg}}=|\clientreq|\geq1$. The public parameters $\mathsf{pp}_{\mathsf{rw}}$ fix a positive number $N_{\mathsf{leaf}}$ of leaf-subproof slots, the maximum range width, the insertion bounds $\MsgBound$ and $\FragBound$, and all padding encodings. For a supported request, let $N_{\mathsf{leaf}}^{\mathsf{act}}\in[N_{\mathsf{leaf}}]$ and set $z_a=1$ exactly for the active prefix $a\leq N_{\mathsf{leaf}}^{\mathsf{act}}$. The active slots have nonempty ranges $I_a=[\mathsf{lo}_a,\mathsf{hi}_a]$ that partition $[N_{\mathsf{cmsg}}]$ consecutively; every disabled slot has $I_a=\emptyset$.

For a range $I_a$, the private layout witness $\alpha_a$ contains, in order, the enabled message-digest entries representing each $G_{i-1}$, the enabled prefix and suffix fragment bytes around each $x_i$ for $i\in I_a$, and the entries representing $G_{N_{\mathsf{cmsg}}}$ only when $\mathsf{hi}_a=N_{\mathsf{cmsg}}$. The relation $\mathsf{DigestLayout}_{I_a,N_{\mathsf{cmsg}}}$ enforces all message and fragment count bounds; preserves the role and original content of every $x_i$; recomputes each wrapped-message digest from its ordered prefix fragments, original content, and ordered suffix fragments; and forms $\mu^{(a)}$ in the exact order
``$G_{\mathsf{lo}_a-1}$, $\mathsf{wrap}(x_{\mathsf{lo}_a})$, \ldots, $G_{\mathsf{hi}_a-1}$, $\mathsf{wrap}(x_{\mathsf{hi}_a})$,'' followed by $G_{N_{\mathsf{cmsg}}}$ exactly for the last active range. It also derives $\mathcal Q^{(a)}_{\mathsf{msg}}$ and $\mathcal Q^{(a)}_{\mathsf{frag}}$ from exactly the enabled additions used in this sequence. Every disabled insertion slot has its fixed padding encoding and contributes to neither the digest sequence nor a lookup query. A disabled slot contributes the logical no-op $\mu^{(a)}=()$, $s^{(a)}=s^{(a-1)}$, and $N_a^{\mathsf{dig}}=0$. Its proof slot contains a valid duplicate proof. Each active slot realizes the corresponding range of Eqs.~\eqref{eq:gformat} and~\eqref{eq:fformat}.

Fig.~\ref{fig:zk-rewrite-proof} presents the aggregation relation in leaf order. The leaf proofs are combined through an aggregation tree. Each parent verifies its immediate child proofs using the fixed verifier data above, binds their bank commitments and request-randomizer digest, and joins consecutive message ranges and matching request-hash/count boundaries. The trusted public circuit profile fixes this tree and a constant upper bound on its depth.

\begin{figure}[t!]
\centering
\small
\begingroup
\setlength{\fboxsep}{3pt}
\setlength{\abovedisplayskip}{2pt}
\setlength{\belowdisplayskip}{2pt}
\fbox{
\begin{minipage}{0.95\linewidth}
\raggedright
\textbf{Public statement:} $x_{\mathsf{rw}}=(\clientreq,\ReqCommit,c_{\mathsf{rw}},\mathsf{pp}_{\mathsf{rw}})$, where $c_{\mathsf{rw}}=(\MsgCommit,\FragCommit)$.

\textbf{Prover:}
\begin{enumerate}[leftmargin=1.25em,itemsep=1pt,topsep=2pt,parsep=0pt]
  \item \textbf{Leaf subproofs.} For every active slot $a\in[N_{\mathsf{leaf}}]$, let $\alpha_a$ contain the selected additions and positions, $\mu^{(a)}$ the resulting digest sequence, and $\mathcal O_a=(\mathcal O^{(a)}_{\mathsf{msg}},\mathcal O^{(a)}_{\mathsf{frag}})$ their bank openings. Let $s^{(a-1)}$ and $s^{(a)}$ be its input and output fold values, and let $N_a^{\mathsf{dig}}$ be its claimed digest count. Set
  \[
  x^{(a)}_{\mathsf{leaf}}=(z_a,I_a,N_{\mathsf{cmsg}},\clientreq|_{I_a},s^{(a-1)},s^{(a)},N_a^{\mathsf{dig}},c_{\mathsf{rw}},\mathsf{pp}_{\mathsf{rw}}).
  \]
  Let $\mathcal Q^{(a)}_{\mathsf{msg}}$ contain the enabled message digests and let $\mathcal Q^{(a)}_{\mathsf{frag}}=\{b(f):f\text{ is an enabled fragment in }\alpha_a\}$. The leaf-subproof circuit $\mathsf{C}_{\mathsf{leaf}}$ enforces
  \begin{equation}
  \begin{aligned}
  &\mathsf{DigestLayout}_{I_a,N_{\mathsf{cmsg}}}(\clientreq|_{I_a},\alpha_a,\mu^{(a)})\\[-2pt]
  &{}\land\ s^{(a)}=\ReqFold(s^{(a-1)};\mu^{(a)})\land N_a^{\mathsf{dig}}=|\mu^{(a)}|\\[-2pt]
  &{}\land\ \LUT(\mathcal Q^{(a)}_{\mathsf{msg}},\MsgCommit;\mathcal O^{(a)}_{\mathsf{msg}})\\[-2pt]
  &{}\land\ \LUT(\mathcal Q^{(a)}_{\mathsf{frag}},\FragCommit;\mathcal O^{(a)}_{\mathsf{frag}}).
  \end{aligned}
  \label{eq:rewrite-leaf-relation}
  \end{equation}
  These statement values remain internal to aggregation. Generate
  \[
  \pi^{(a)}_{\mathsf{leaf}}\leftarrow\ZKProve\bigl(\mathsf{C}_{\mathsf{leaf}};x^{(a)}_{\mathsf{leaf}};(\mu^{(a)},\alpha_a,\mathcal O_a)\bigr).
  \]
  Fill each disabled proof slot with a valid duplicate. Its logical contribution is the no-op defined above.
  \item \textbf{Aggregation.} The aggregation circuit $\mathsf{C}_{\mathsf{rw}}$ recursively verifies all $N_{\mathsf{leaf}}$ leaf subproofs. Setting $N_{\mathsf{cmsg}}=|\clientreq|$, it requires $(z_a)_a$ to select a nonempty prefix whose ranges consecutively partition the public $\clientreq$ and every remaining slot to contribute unchanged fold boundaries and zero count. It binds the shared subproof inputs to $x_{\mathsf{rw}}$, chains their fold boundaries from $s^{(0)}=\mathsf{IV}$, and enforces
  \[
  \begin{aligned}
  N_{\mathsf{rw}}^{\mathsf{dig}}&=\textstyle\sum_{a=1}^{N_{\mathsf{leaf}}}N_a^{\mathsf{dig}},\\[-2pt]
  d_{\mathsf{rw}}&=H_{\mathsf{zk}}(\mathsf{tag}_{\mathsf{req\text{-}final}}\|N_{\mathsf{rw}}^{\mathsf{dig}}\|s^{(N_{\mathsf{leaf}})}),\\[-2pt]
  \ReqCommit&=\Commit(d_{\mathsf{rw}};\ReqRand).
  \end{aligned}
  \]
  {The count $N_{\mathsf{rw}}^{\mathsf{dig}}$ is private and bounded only by the leaf capacities.} With $w_{\mathsf{rw}}^{\mathsf{agg}}=(\ReqRand,N_{\mathsf{rw}}^{\mathsf{dig}},\{z_a,I_a,s^{(a-1)},s^{(a)},N_a^{\mathsf{dig}},\pi^{(a)}_{\mathsf{leaf}}\}_{a=1}^{N_{\mathsf{leaf}}})$, output
  \[
  \begin{aligned}
  \pi_{\mathsf{rw}}\leftarrow\ZKProve\bigl(\mathsf{C}_{\mathsf{rw}};x_{\mathsf{rw}};w_{\mathsf{rw}}^{\mathsf{agg}}\bigr).
  \end{aligned}
  \]
\end{enumerate}

\textbf{Verifier:} $\ZKVerify(\mathsf{C}_{\mathsf{rw}};x_{\mathsf{rw}};\pi_{\mathsf{rw}})\in\{\mathtt{accept},\mathtt{reject}\}$.
\end{minipage}
}
\caption{Protocol for the rewrite aggregation proof.}
\label{fig:zk-rewrite-proof}
\endgroup
\end{figure}

\subsubsection{End-to-end Protocol}
\label{app:e2e-protocol}

Fig.~\ref{fig:zk-verirouter-e2e-protocol} gives the end-to-end protocol, which composes the model-selection proof, rewrite proof, and serving certificate by sharing $\model$, $\ReqCommit$, and $\rid$ across their checks. All structured inputs to $H_{\mathsf{sig}}$ use fixed, injective, length-delimited encodings. The client obtains the approved root verifier keys and circuit profile through the authenticated context. The figures name the bound query and request as $e_q$ and $\clientreq$. The implementation represents them by client-computed digests binding the dimension/zero point and message count, respectively.

\begin{figure}[t!]
\centering
\small
\begingroup
\setlength{\fboxsep}{3pt}
\setlength{\abovedisplayskip}{2pt}
\setlength{\belowdisplayskip}{2pt}
\fbox{
\begin{minipage}{0.95\linewidth}
\raggedright
\textbf{Public context:} $\ctx$ from Eq.~\eqref{eq:public-context}.

\textbf{Selection regime:} At session setup, $\client$ chooses per-turn or session-pinned selection, which determines $\mmode$ for each turn.

\textbf{Request:} $\client$ chooses a fresh $\rid$ and sends $(\rid,\clientreq)$ to $\router$.

\textbf{Serving and proof generation:}
\begin{enumerate}[leftmargin=1.25em,itemsep=1pt,topsep=2pt,parsep=0pt]
  \item On a fresh-selection turn, $\router$ selects $\model$; on a pinned turn, it reuses $\model_{\mathsf{pin}}$. It then constructs $\receivedreq\in\trans(\clientreq)$, samples fresh $\ReqRand$, sets $\ReqCommit=\Commit(\ReqHash(\receivedreq);\ReqRand)$, and forwards $(\rid,\model,\receivedreq,\ReqRand)$ to $\provider_\model$.
  \item On a fresh-selection turn, $\router$ runs the model-selection aggregation protocol in Fig.~\ref{fig:zk-control-plane-aggregation} to generate $\pi_{\mathsf{sel}}$. On every turn, it runs the rewrite aggregation protocol in Fig.~\ref{fig:zk-rewrite-proof} to generate $\pi_{\mathsf{rw}}$.
  \item $\provider_\model$ recomputes $\ReqCommit$, evaluates $Y=\model(\receivedreq)$, and returns $Y$ and $\mathsf{SC}=(\rid,\ReqCommit,h_Y,\sigma_\model)$ as defined in Eq.~\eqref{eq:serving-certificate}. $\router$ sends $(Y,\model,\mathsf{SC},\pi_{\mathsf{rw}},\pi_{\mathsf{sel}})$ to $\client$, omitting $\pi_{\mathsf{sel}}$ on pinned turns.
\end{enumerate}

\textbf{Client verification:}
\begin{enumerate}[leftmargin=1.25em,itemsep=1pt,topsep=2pt,parsep=0pt]
  \item Parse $\mathsf{SC}=(\rid,\ReqCommit,h_Y,\sigma_\model)$, obtain $pk_\model=\mathsf{Reg}[\model]$, require that $\rid$ equals the unused identifier issued for this request, check $h_Y=H_{\mathsf{sig}}(Y)$, and verify
  \[
  \mathsf{Verify}_{pk_\model}\bigl(H_{\mathsf{sig}}(\rid\|\model\|\ReqCommit\|h_Y),\sigma_\model\bigr)=1.
  \]
  \item Verify $\ZKVerify(\mathsf{C}_{\mathsf{rw}};x_{\mathsf{rw}};\pi_{\mathsf{rw}})=\mathtt{accept}$ using the same $\clientreq$ and $\ReqCommit$.
  \item On a fresh-selection turn, derive $e_q=\encoder(\clientreq)$ and verify $\ZKVerify(\mathsf{C}_{\mathsf{sel}};x_{\mathsf{sel}};\pi_{\mathsf{sel}})=\mathtt{accept}$ under the approved root key. Require the model identifier in the verified $x_{\mathsf{sel}}$ to equal $\model$ used in the serving-certificate check. On a pinned turn, instead require $\model=\model_{\mathsf{pin}}$ from the accepted initial fresh-selection turn.
  \item Accept $(Y,\model)$ only if all checks required for $\mmode$ accept. On the initial accepted turn of a session-pinned regime, set $\model_{\mathsf{pin}}\gets\model$; then mark $\rid$ completed.
\end{enumerate}
\end{minipage}
}
\caption{End-to-end serving and verification protocol.}
\label{fig:zk-verirouter-e2e-protocol}
\endgroup
\end{figure}

\subsection{Constraint Complexity}
\label{app:constraint-complexity}

This section analyzes the constraint complexity of the protocols.

All bounds below count constraints over the public padded capacities. Fixed-width local checks cost $O(1)$ per processed entry, while a Merkle opening costs constraints linear in its tree height. The batched $\mathsf{LUT}$ and $\mathsf{MSetEq}$ arguments from Sec.~\ref{sec:bg-zkp} use constraints linear in their total table, query, or sequence lengths. Ordered copies are supplied as witnesses, bound by multiset equality, and checked locally, so they require no sorting network in the circuit.

\noindent\textbf{Data-structure checkers.}
The stack checker processes padded event blocks and its final sequence. The priority-queue checker processes push/pop projections and an inlined stack trace over $\OutSeq$. Let their total padded checker lengths be
\[
\begin{aligned}
\mathsf{Len}_{\mathsf{stk}}
&:=\StackPushLenMax(\NumPopsMax+1)+\FinLenMax,\\
\mathsf{Len}_{\mathsf{pq}}
&:=\PushLenMax+\PopLenMax\\[-2pt]
&\qquad +(\PopLenMax+\PQFinLenMax)(\NumPopsMax+1)+\FinLenMax.
\end{aligned}
\]
In $\mathsf{Len}_{\mathsf{pq}}$, $\NumPopsMax$ and $\FinLenMax$ come from the inlined stack parameters $\xstkPQ$.

For the stack checker, $\mathsf{StackTransition}$ applies a constant number of local constraints to each padded event. Its parent lookup has one table row per push and one query per push or pop, while $\mathsf{StackAccounting}$ compares the push sequence with the popped and final items. Each component therefore processes $O(\mathsf{Len}_{\mathsf{stk}})$ entries, giving $O(\mathsf{Len}_{\mathsf{stk}})$ constraints.

For the priority-queue checker, $\mathsf{TraceConsistency}$ scans $\OpPosVec$ and the push and pop projections, including their position multiset equality. $\mathsf{QueueAccounting}$ processes the push projection and $\OutSeq$. $\mathsf{QueueOrder}$ adds the inlined stack checker, one removal-position lookup, and constant local inequalities per $\OutSeq$ item, while $\mathsf{FinalSequenceOrder}$ scans the final sequence once. All non-stack inputs have total length $O(\PushLenMax+\PopLenMax+\PQFinLenMax)$ and are absorbed by $\mathsf{Len}_{\mathsf{pq}}$; its remaining term is exactly the padded inlined-stack length. Hence the checker uses $O(\mathsf{Len}_{\mathsf{pq}})$ constraints.

\noindent\textbf{Model-selection proof.}
Let $\mathsf{ht}_{\mathsf{ref}}$ and $\mathsf{ht}_{\mathsf{graph}}$ be the reference- and graph-table Merkle heights. The padded score table has $\ScoreTableCap$ rows, and each reference vector has dimension $D$. The upper-layer subproof has $\ustep$ padded steps; writing $B_{\mathsf{upper}}=\max_{\ell\geq1}B_\ell$, define its padded node workload as $\mathsf{Len}_{\mathsf{upper}}:=\ustep(B_{\mathsf{upper}}+1)$. At level $0$, $N_{\mathsf{round}}^{0,\max}$ is the padded expansion bound and $T=1+N_{\mathsf{round}}^{0,\max}B_0$ is the length of $\WorkUpdTrace$. Let $\mathsf{Len}_{\mathsf{pq}}^{\CandQ}$ and $\mathsf{Len}_{\mathsf{pq}}^{\WorkQ}$ be the priority-queue checker lengths above instantiated for the two queues, and set $\mathsf{Len}_{\mathsf{ef}}:=T+\mathsf{Len}_{\mathsf{pq}}^{\CandQ}+\mathsf{Len}_{\mathsf{pq}}^{\WorkQ}$.
This includes the level-$0$ update rows and both inlined queue checkers. The filters and multiset-equality checks make only a constant number of linear passes over these rows.

The score-opening subproof computes and authenticates at most $\ScoreTableCap$ $D$-dimensional scores. Each traversal subproof makes another linear pass over the padded score table, absorbed by the same $\ScoreTableCap$ term. Upper-layer traversal opens at most one graph row per padded step, and level-$0$ traversal opens at most one per padded expansion. Let $\mathsf{VfyCost}_{\mathsf{sel}}$ be the total cost of recursively verifying the three subproofs in aggregation, and write $|x_{\mathsf{sel}}|$ for the number of field elements in the public statement. The full model-selection proof therefore uses
\[
\begin{aligned}
O\bigl(&\ScoreTableCap(D+\mathsf{ht}_{\mathsf{ref}}+1)
+\mathsf{Len}_{\mathsf{upper}}+\mathsf{Len}_{\mathsf{ef}}\\[-2pt]
&+(\ustep+N_{\mathsf{round}}^{0,\max})\mathsf{ht}_{\mathsf{graph}}
+\mathsf{VfyCost}_{\mathsf{sel}}+|x_{\mathsf{sel}}|\bigr).
\end{aligned}
\]
Apart from recursive verification, this is linear in the displayed sum of padded score and traversal workloads. The full committed-index size enters only through the Merkle heights.

\noindent\textbf{Rewrite proof.}
Let $\mathsf{Rows}_{\mathsf{rw}}$ be the total number of padded local rows processed across the rewrite leaf subproofs and their aggregation, excluding Merkle paths and recursive proof verification. It includes byte-hash steps, digest-layout and fold slots, selected-addition and range checks, and public-request bytes matched to subproof slices. Let $\mathsf{OpenSlots}_{\mathsf{msg}}$ and $\mathsf{OpenSlots}_{\mathsf{frag}}$ be the total padded message- and fragment-bank opening slots, with Merkle heights $\mathsf{ht}_{\mathsf{msg}}$ and $\mathsf{ht}_{\mathsf{frag}}$. For $N_{\mathsf{leaf}}$ leaf subproofs, let $\mathsf{VfyCost}_{\mathsf{leaf}}$ be the cost of recursively verifying one such subproof. The full rewrite proof uses
\[
\begin{aligned}
O\bigl(&\mathsf{Rows}_{\mathsf{rw}}
+\mathsf{OpenSlots}_{\mathsf{msg}}\mathsf{ht}_{\mathsf{msg}}
+\mathsf{OpenSlots}_{\mathsf{frag}}\mathsf{ht}_{\mathsf{frag}}\\[-2pt]
&+N_{\mathsf{leaf}}\mathsf{VfyCost}_{\mathsf{leaf}}\bigr)
\end{aligned}
\]
constraints. A complete inserted message contributes an authenticated digest and a fold slot but no byte hashing because the message bank supplies its digest. Fragment bytes contribute both the bank-digest and wrapped-message hash passes to $\mathsf{Rows}_{\mathsf{rw}}$. For fixed Merkle heights and recursive verifier, the bound is linear in $\mathsf{Rows}_{\mathsf{rw}}+\mathsf{OpenSlots}_{\mathsf{msg}}+\mathsf{OpenSlots}_{\mathsf{frag}}+N_{\mathsf{leaf}}$. With binary Merkle banks, varying bank capacity affects it only through logarithmic Merkle heights.

\subsection{Algorithmic Correctness}
\label{app:algorithmic-correctness}

We prove the stack, priority-queue and level-$0$ \textsf{efSearch} characterizations with exact multiset equality.

\subsubsection{Monotone-stack Checker Correctness}

\begin{theorem}[Equivalent checker for monotone-stack traces]
\label{thm:stack-checker}
For any strict total order $\StackOrd$, Eq.~\eqref{eq:stk-condition} holds if and only if $\Trace_{\mathsf{stk}}$ is a valid monotone-stack execution on $\StackPushSeq$ and $\FinStack$ contains exactly the items remaining after that execution.
\end{theorem}

\begin{proof}
If $\StackPushLen=0$, $\mathsf{StackTransition}$ is vacuous and $\mathsf{StackAccounting}$ holds exactly when $\FinStack$ is empty. This is precisely the empty monotone-stack execution. Assume $\StackPushLen>0$ below.

\noindent\textbf{Necessity.} Execute the monotone-stack algorithm and let $S_{i,j}$ be the stack after the $j$-th event of $\Trace_{\mathsf{stk},i}$, with $S_{i,0}$ denoting the stack before processing $\StackPushItem{i}$. Set $\StackTop{i}{j}=\mathsf{top}(S_{i,j})$, using $\NoParent$ for an empty stack. For every pushed item $u$, set $\ParentOf{u}$ to the top immediately before its push. A nonempty parent is therefore already on the stack when $u$ is pushed. A pop of $u_{\mathsf{pop}}=\PoppedItem{i}{j}$ then gives $u_{\mathsf{pop}}=\StackTop{i}{j-1}$, $u_{\mathsf{pop}}\StackOrd\StackPushItem{i}$, and $\StackTop{i}{j}=\ParentOf{u_{\mathsf{pop}}}$. After the last pop, $\ParentOf{\StackPushItem{i}}=\StackTop{i}{\NumPops{i}}$, this top is $\NoParent$ or fails the pop condition, and the push sets $\StackTop{i}{\NumPops{i}+1}=\StackPushItem{i}$. The initial and boundary equalities are $\StackTop{1}{0}=\NoParent$ and $\StackTop{i+1}{0}=\StackTop{i}{\NumPops{i}+1}$. Hence $\mathsf{StackTransition}$ holds. Each insertion instance is popped once or remains in the final stack, so $\{\StackPushItem{i}\}_{i\in[\StackPushLen]}=_{\mathsf{ms}}\{\PoppedItem{i}{j}\}_{i\in[\StackPushLen],\,j\in[\NumPops{i}]}\uplus\FinStack$, establishing $\mathsf{StackAccounting}$.

\noindent\textbf{Sufficiency.} Assume Eq.~\eqref{eq:stk-condition}. Process the verified events in trace order and let $\widehat S_{i,j}$ denote the reconstructed stack after event $j$ of $\Trace_{\mathsf{stk},i}$, with $\widehat S_{i,0}$ denoting the stack before processing $\StackPushItem{i}$. We maintain the invariant that $\StackTop{i}{j}$ is the first item of $\widehat S_{i,j}$, or $\NoParent$ when it is empty, and that $\ParentOf{u}$ is the next item below every $u$ in $\widehat S_{i,j}$, or $\NoParent$ for its bottom item. The base case follows from $\StackTop{1}{0}=\NoParent$. For a pop, the lookup $(\PoppedItem{i}{j},\StackTop{i}{j})\in\ParentTab$ makes $\PoppedItem{i}{j}$ a real item{, since every first component of $\ParentTab$ is a pushed item}. Hence the stack is nonempty, and $\PoppedItem{i}{j}=\StackTop{i}{j-1}$ removes its first item $u_{\mathsf{pop}}$, and $\StackTop{i}{j}=\ParentOf{u_{\mathsf{pop}}}$ is the first item of the remaining suffix. For a push of $u=\StackPushItem{i}$, $\ParentOf{u}=\StackTop{i}{\NumPops{i}}$ links $u$ to the current top, and $\StackTop{i}{\NumPops{i}+1}=u$ prepends it to the stack. Since each $u$ is a new insertion instance, this push cannot create a parent cycle. The equality $\StackTop{i+1}{0}=\StackTop{i}{\NumPops{i}+1}$ carries this invariant to the next pushed item.

Every reported pop therefore removes the actual top and satisfies the pop condition. The final pre-push top is $\NoParent$ or fails that condition, so no required pop is omitted. Thus each $\Trace_{\mathsf{stk},i}$, and hence $\Trace_{\mathsf{stk}}$, equals the monotone-stack execution of $\StackPushSeq$. Finally, $\mathsf{StackAccounting}$ identifies $\FinStack$ with exactly the insertion instances not popped by this execution.

\end{proof}

\subsubsection{Priority-queue Checker Correctness}
\label{app:proof-pq-equiv-checker}

\noindent\textbf{Parent-map property.} {The circuit needs no encoding of the value $\infty$: for a final item, $\InPos{r}<\infty$ holds trivially, and by Eq.~\eqref{eq:pq-final-order} no item's latest preceding lower-priority item is final, so no parent payload is $\infty$.} When the monotone stack scans $\OutSeq$ under $\StackOrdPQ$, its map $\ParentMapPQ$ satisfies $\ParentOf{\OutItem{r}}=\OutItem{\PLIndex{r}}$, the latest preceding lower-priority item, or $\NoParent$ if no such index exists. Suppose $\PLIndex{r}>0$. This item cannot have been popped before step $r$, because the later item that popped it would itself be a later lower-priority predecessor of $r$, contradicting the maximality of $\PLIndex{r}$. Every item strictly between $\PLIndex{r}$ and $r$ that remains on the stack has higher priority than $\OutItem{r}$ and is popped while processing $r$, leaving $\OutItem{\PLIndex{r}}$ on top. If $\PLIndex{r}=0$, the same argument pops every remaining predecessor and leaves the stack empty.

\begin{proof}[Proof of Theorem~\ref{thm:pq-equiv-checker}]
\noindent\textbf{Necessity.} Assume a valid execution and list its remaining items in decreasing priority as $\FinalPQ$. Form $\OutSeq$ accordingly. For each $r\in[\OutLen]$, set $\InPos{r}$ to the position at which $\OutItem{r}$ was pushed. Every pushed insertion instance is either popped once or remains in the final queue, so Eq.~\eqref{eq:pq-in-multiset} holds. The order of $\FinalPQ$ gives Eq.~\eqref{eq:pq-final-order}.

Fix $r\in[\OutLen]$. Since an item must be inserted before it is removed, $\InPos{r}<\OutPos{r}$. If $\PLIndex{r}=0$, then $\OutPos{0}=0<\InPos{r}$. Suppose $\PLIndex{r}>0$. We have $\PLIndex{r}<r$ and $\OutItem{r}\HigherPrio\OutItem{\PLIndex{r}}$. The item $\OutItem{\PLIndex{r}}$ cannot be a final survivor. Otherwise $\OutItem{r}$ would occur later in the final suffix, contradicting its decreasing-priority order. Hence $\OutPos{\PLIndex{r}}<\infty$. If $\InPos{r}<\OutPos{\PLIndex{r}}$, then $\PLIndex{r}<r$ also gives $\OutPos{\PLIndex{r}}<\OutPos{r}$, so $\OutItem{r}$ is live when the queue pops the lower-priority $\OutItem{\PLIndex{r}}$, contradicting priority-queue semantics. Equality $\InPos{r}=\OutPos{\PLIndex{r}}$ is impossible because the assigned operation positions are distinct. Therefore $\OutPos{\PLIndex{r}}<\InPos{r}<\OutPos{r}$, proving Eq.~\eqref{eq:pq-prevlower-main}.

\noindent\textbf{Sufficiency.} Assume Eqs.~\eqref{eq:pq-in-multiset}--\eqref{eq:pq-final-order}. By the definition in Sec.~\ref{sec:zk-priority-queue}, ordering the projection entries by position and discarding their labels gives $\Trace_{\mathsf{pq}}$. Eq.~\eqref{eq:pq-in-multiset} associates each item in $\OutSeq$ with exactly one push. For each pop index $j\in[\PopLen]$, Eq.~\eqref{eq:pq-prevlower-main} gives $\InPos{j}<\OutPos{j}=\PopTime{j}$, so $\OutItem{j}=\PopItem{j}$ is live when reported popped.

Suppose that a higher-priority item $\OutItem{r}$ is also live at $\PopTime{j}$. Then $\InPos{r}<\OutPos{j}<\OutPos{r}$. The removal positions are increasing along the pop prefix and equal $\infty$ on the final suffix, so $\OutPos{j}<\OutPos{r}$ implies $j<r$. Thus $j$ is an eligible index in the definition of $\PLIndex{r}$, and maximality of $\PLIndex{r}$ gives $\PLIndex{r}\geq j$. Monotonicity of removal positions then gives $\OutPos{\PLIndex{r}}\geq\OutPos{j}$. Eq.~\eqref{eq:pq-prevlower-main} instead requires $\OutPos{\PLIndex{r}}<\InPos{r}<\OutPos{j}$, a contradiction. Hence every pop removes the highest-priority live item.

After all operations, Eq.~\eqref{eq:pq-in-multiset} makes $\FinalPQ$ exactly the unpopped items, and Eq.~\eqref{eq:pq-final-order} orders them by decreasing priority. Therefore the trace and the claimed final sequence satisfy the theorem.
\end{proof}

To apply this theorem to the circuit, observe that every witness satisfying Eq.~\eqref{eq:pq-condition} fixes one operation order through $\mathsf{TraceConsistency}$. $\mathsf{QueueAccounting}$ and $\mathsf{FinalSequenceOrder}$ give Eqs.~\eqref{eq:pq-in-multiset} and~\eqref{eq:pq-final-order}. Theorem~\ref{thm:stack-checker} and the parent-map property above make $\ParentOf{\OutItem{r}}=\OutItem{\PLIndex{r}}$, so the removal-position lookup and local inequalities give Eq.~\eqref{eq:pq-prevlower-main}. Theorem~\ref{thm:pq-equiv-checker} then gives a valid priority-queue trace and its exact decreasing-priority final sequence. Conversely, its necessity construction and the honest monotone-stack execution supply a satisfying witness for every such trace fitting the declared capacities.

\subsubsection{\texorpdfstring{Level-$0$ \textsf{efSearch} Constraint Correctness}{Level-0 efSearch Constraint Correctness}}
\label{app:proof-efsearch}

\begin{theorem}[Level-$0$ \textsf{efSearch} constraint correctness]
\label{thm:efsearch-proof}
Fix a graph, a score table and an entry node. With exact multiset equalities and lookups into this graph, the constraints in Appendix~\ref{app:wc-instantiation} satisfy:
\begin{enumerate}[label=(\roman*),itemsep=0pt,topsep=1pt]
\item \uline{Completeness.} Every supported execution of Alg.~\ref{alg:hnsw-efsearch} whose top-$k$ output contains a node $u$ with $\mathsf{mid}(u)=\model$ supplies a satisfying assignment with that output.
\item \uline{Soundness.} Every satisfying assignment describes an execution with these inputs. Its claimed set is the execution's top-$k$ output and contains a node $u$ with $\mathsf{mid}(u)=\model$.
\end{enumerate}
\end{theorem}

\begin{proof}
\noindent\textbf{(i) Queue-trace validity.} The circuit $\cef$ enforces both checks in Eq.~\eqref{eq:efsearch-queue-checkers} on its queue variables and the node keys from its fixed score table. The priority-queue argument above therefore gives valid queue traces and exact final sequences under the two instantiated priority orders. It remains to show that the sentinel representation of $\WorkQ$ implements its capacity bound. For a real-node set $\mathcal A$, let $\mathsf{Top}_{\EfCap}(\mathcal A)$ be its $\EfCap$ largest-key members, or all of $\mathcal A$ if $|\mathcal A|<\EfCap$. For each $j\in[|\WorkPushTrace|]$, let $\widehat{u}_j$ be the $j$th item of $\WorkPushTrace$, let $u_j$ be its node component, and let $W_j$ be the real-node projection after the first $j$ logical updates. Let $\mathsf{row}(j)$ be the unique enabled row projecting to $\widehat{u}_j$. Then
\begin{equation}
\begin{aligned}
W_j&=\mathsf{Top}_{\EfCap}(W_{j-1}\cup\{u_j\}),\\[-2pt]
\AdmitBit{\mathsf{row}(j)}=1&\Longleftrightarrow u_j\in W_j.
\end{aligned}
\label{eq:workq-capacity-equivalence}
\end{equation}
We prove Eq.~\eqref{eq:workq-capacity-equivalence} by induction on $j$. The base case is the sentinel-only initialization, whose real-node projection is empty. If $|W_{j-1}|<\EfCap$, the minimum after insertion is a sentinel, so the new node remains. If $|W_{j-1}|=\EfCap$, the update removes the minimum of $W_{j-1}\cup\{u_j\}$. The inserted node is removed exactly when it is outside $\mathsf{Top}_{\EfCap}(W_{j-1}\cup\{u_j\})$. This proves the capacity invariant and admission equivalence in Eq.~\eqref{eq:workq-capacity-equivalence}, matching the conditional update in Alg.~\ref{alg:hnsw-efsearch}.

\noindent\textbf{(ii) Cross-queue consistency.} First, committed graph lookups bind the node sequence in $\WorkUpdTrace$ to the entry node followed by the level-$0$ neighbor lists exposed in $\CandQ$ pop order. Second, Eq.~\eqref{eq:first-visit-stable-filter} makes $\WorkPushTrace$ exactly the first-visit item sequence. Every enabled row uses its reserved $\tau_i$ as the node's persistent tie-breaking position. Equation~\eqref{eq:workq-fixed-positions} makes that row act on the $\WorkQ$ contents produced by its enabled predecessors.

Third, let $\CandPushLen:=|\CandPushTrace|$ and let $i_1<\cdots<i_{\CandPushLen}$ be the admitted rows of $\WorkUpdTrace$. Since Eq.~\eqref{eq:efsearch-direct-projection} directly reuses these rows as the candidate insertion projection,
\begin{equation}
(\CandPushTime{j},\widehat{\CandPushItem{j}})
=\bigl(\tau_{i_j},(n_{i_j},\score_{i_j},\tau_{i_j})\bigr)
\qquad(j\in[\CandPushLen]).
\label{eq:cand-work-sequence-equality}
\end{equation}
Thus $\CandQ$ receives exactly the admitted first visits in their original order and with the same reserved positions. Equation~\eqref{eq:efsearch-fixed-positions} places every round-$a$ push after the pop that exposed its authenticated neighbor row and, when another certified round follows, before the next certified candidate pop. Starting from the entry push at position $1$, induction on $a$ therefore makes the candidate queue at the start of every round before the output-preserving suffix equal the queue obtained from the preceding rounds of Alg.~\ref{alg:hnsw-efsearch}. The stopping argument below handles any later certified suffix. Together, the authenticated neighbor openings, Eq.~\eqref{eq:first-visit-stable-filter}, Eq.~\eqref{eq:cand-work-sequence-equality}, and the assigned operation positions for both queues reproduce the cross-queue operations and stable first-encounter priorities of every nonterminal algorithm round.

\noindent\textbf{(iii) Stopping rule and output.} Index the certified expansion rounds by $a$. Let $C_a$ and $W_a$ be the candidate and retained real-node sets at the start of round $a$, so $\CandPopItem{a}=\max C_a$, and, when $|W_a|=\EfCap$, let $\theta_a=\min_{u\in W_a}\mathsf{key}(u)$. If round $a$ admits an update, every row before its first admitted row is rejected and leaves $W_a$ unchanged. By Eq.~\eqref{eq:workq-capacity-equivalence}, that first admission evicts a real item if and only if $|W_a|=\EfCap$. When real, the eviction has key $\theta_a$. Equation~\eqref{eq:efsearch-round-threshold} therefore enforces $\mathsf{key}(\CandPopItem{a})\geq\theta_a$ exactly when required. Hence, if $\mathsf{key}(\CandPopItem{a})<\theta_a$, the round admits no update: $W_{a+1}=W_a$ and no new candidate is pushed. Therefore $C_{a+1}=C_a\setminus\{\CandPopItem{a}\}$. Repeating this argument gives, for every subsequent nonempty round $b\geq a$,
\begin{equation}
\begin{aligned}
W_{b+1}&=W_a,\\[-2pt]
C_{b+1}&=C_b\setminus\{\CandPopItem{b}\},\\[-2pt]
C_{b+1}\ne\emptyset&\Rightarrow
\max_{u\in C_{b+1}}\mathsf{key}(u)\leq\mathsf{key}(\CandPopItem{b})<\theta_a.
\end{aligned}
\end{equation}
Thus the suffix contains no admitted update and leaves $\WorkQ$ unchanged. Deleting it leaves the pre-pop configuration of round $a$. Appending the algorithm's $\mathsf{PopMax}$ of $\CandPopItem{a}$ and taking its stopping branch yields a valid execution with the same final $\WorkQ$. Equation~\eqref{eq:efsearch-terminal} is the equivalent pre-pop form of the algorithm's stopping rule. If $C_{\mathsf{fin}}$ is nonempty, its largest-key item is the next candidate popped and immediately satisfies the stopping test.

\noindent\textbf{Completeness.} From an honest supported execution, form $\WorkUpdTrace$ from the authenticated neighbor rows, assign its bits and positions as above, and prepend the $\EfCap$ sentinels. Parts (i)--(iii) show that these witnesses satisfy both queue relations, the cross-queue links, and the stopping checks, using empty terminal queues on exhaustion or the equivalent pre-pop queues on early stopping. Equation~\eqref{eq:efsearch-topk-output} and the model-membership constraint bind the claimed output, so the honest execution supplies a satisfying assignment.

\noindent\textbf{Soundness.} Consider a satisfying assignment to these constraints. Parts (i)--(iii) reconstruct an execution of Alg.~\ref{alg:hnsw-efsearch} with the same final $\WorkQ$, after deleting any below-threshold no-admission suffix. Equation~\eqref{eq:efsearch-topk-output} binds its returned top-$k$ set, and the membership constraint binds one returned row's model identifier to $\model$.
\end{proof}

\subsection{Cryptographic Security}
\label{app:cryptographic-security}

{This section proves Theorems~\ref{thm:control-plane},~\ref{thm:data-plane} and~\ref{thm:verirouter-security} using the algorithmic results of Appendix~\ref{app:algorithmic-correctness}. All adversaries are probabilistic polynomial-time (PPT) in the security parameter $\lambda$, and a function of $\lambda$ is negligible if it is eventually smaller than $1/\mathrm{poly}(\lambda)$ for every polynomial. We model $H_{\mathsf{zk}}$ as a random oracle~\cite{bellare1993random}, as is standard for arguments made non-interactive by the Fiat--Shamir transform~\cite{fiat1986prove}.

\subsubsection{Security Preliminaries}
\label{app:security-extraction}

An assignment satisfying the subproof constraints describes a correct execution if it uses the endorsed index and banks and its multiset equalities hold exactly. Appendix~\ref{app:algorithmic-correctness} shows this for model selection, and the proof of Theorem~\ref{thm:rewrite-proof} shows it for rewriting. Soundness therefore needs three facts about an accepted root proof: (i)~satisfying witnesses of all its subproofs can be extracted; (ii)~these witnesses open only endorsed rows; and (iii)~the multiset equalities that the model-selection subproofs check with randomized tests hold exactly. The paragraphs below establish them in turn.

\noindent\textbf{Extraction through recursion.}
An argument is knowledge sound if, for every PPT prover that outputs a statement $x$ and a proof $\pi$ accepted for a circuit $\mathsf C$, a PPT extractor with access to the prover outputs a witness $w$ with $\mathsf C(x;w)=1$, except with negligible probability. Our construction requires only this property, which holds in the random-oracle model for Fiat--Shamir-compiled FRI-based arguments such as Plonky2~\cite{block2023fsfri}. Each model-selection or rewrite proof is a tree of constant depth whose parent circuits verify their children against fixed verifier data. Applying knowledge soundness to the root and then to every extracted child proof yields witnesses of all subproofs. These witnesses satisfy their circuits and agree with the values on which the parents constrain them. Constant depth keeps the composed extractor polynomial-time and its failure probability negligible; this composition uses knowledge soundness with auxiliary input~\cite{bitansky2012recursive}. Since a parent circuit evaluates $H_{\mathsf{zk}}$ itself when it verifies a child proof, we assume, as for every recursive argument, that knowledge soundness also holds for these recursive verifier circuits. We call the result the \emph{extracted witness} of an accepted proof.

\noindent\textbf{Binding to the endorsed context.}
The semantic index $\indexset$ and the rewrite banks bound by $\ctx$ are the values whose commitment openings the auditor reviewed (Appendix~\ref{app:discussion}). Given these openings as auxiliary input, an efficient adversary that opens a context commitment to a different digest breaks binding, and one that opens a different row at a position of an endorsed Merkle tree produces a hash collision. Hence, except with negligible probability, every row that an extracted witness opens against $\ctx$ is a row of the endorsed table, and by the one-row-per-key invariant of Appendix~\ref{app:discussion} it is the endorsed row of its key. In particular, an authenticated message digest equals $m(g)$ for some $g\in\gapset$, and an authenticated fragment digest computed from extracted bytes identifies those bytes as a fragment of $\fragset$, because $b(\cdot)$ is collision resistant.

\noindent\textbf{Randomized multiset checks.}
The model-selection subproof circuits check every $\MSetEq$, including those that realize the private lookups of Appendix~\ref{app:committed-lookup}, with the standard randomized grand-product test~\cite{bayer2012efficient,gabizon2019plonk}. Let $\mathbb F$ be the circuit field. Consider a check $j$ between two sequences of at most $n_j$ tuples in $\mathbb F^{w_j}$. For a challenge pair $(\alpha,\gamma)\in\mathbb F^2$, the circuit maps each tuple $t$ to $\phi_{\alpha,\gamma}(t)=\gamma+\sum_{k=0}^{w_j-1}\alpha^kt_k$ and requires the products over the two sequences to be equal. It repeats this test for $r$ challenge pairs, obtained by hashing a domain-separated transcript that fixes both compared sequences, either directly or through values that the circuit determines from it. Let $\mathcal J$ be the set of these checks in the subproofs of one model-selection proof, let $\delta=\max_{j\in\mathcal J}n_j\max\{1,w_j-1\}<|\mathbb F|$, and let $Q$ be the number of random-oracle queries made by the adversary and the extractor. The number of repetitions $r$ is chosen such that $(\delta/|\mathbb F|)^r\leq2^{-\lambda}$. The rewrite subproofs perform no such check.

\begin{lemma}[Soundness of multiset checks]
\label{lem:multiset-error}
Except with probability at most $(Q+|\mathcal J|)(\delta/|\mathbb F|)^r\leq(Q+|\mathcal J|)\cdot2^{-\lambda}$ plus the probability of a random-oracle collision, which is negligible, every multiset equality checked in the extracted subproof witnesses of an accepted model-selection proof holds exactly.
\end{lemma}
\begin{proof}
Fix a transcript whose two sequences $L$ and $R$ differ as multisets. Each factor $\phi_{\alpha,\gamma}(t)$ is monic of degree one in $\gamma$, and distinct tuples give distinct factors. By unique factorization, $\prod_{t\in L}\phi_{\alpha,\gamma}(t)-\prod_{t\in R}\phi_{\alpha,\gamma}(t)$ is a nonzero polynomial of total degree at most $\delta$, so by the Schwartz--Zippel lemma~\cite{schwartz1980fast,zippel1979probabilistic} it vanishes at a uniform challenge pair with probability at most $\delta/|\mathbb F|$. The $r$ pairs are uniform and independent random-oracle outputs, so all $r$ tests pass with probability at most $(\delta/|\mathbb F|)^r$. Without a collision, the sequences compared in an accepted witness are those fixed by its transcript. Each such transcript is either queried by the adversary or the extractor, or is evaluated for the first time when the extracted witness is checked, which adds at most $|\mathcal J|$ transcripts. A union bound over these transcripts gives the bound, which is negligible because $Q$ and $|\mathcal J|$ are polynomial in $\lambda$.
\end{proof}

\subsubsection{Proof of Control-plane Integrity}
\label{app:proof-control-plane}

\begin{proof}[Proof of Theorem~\ref{thm:control-plane}]
\noindent\textbf{Completeness.}
Consider first a fresh-selection turn. For an honest supported execution, form $\overline{\dtable}$ by sorting the nodes scored by either traversal by identifier and padding to $\ScoreTableCap$. Every active row opens against $\RefCommit$ and satisfies $\score_u=\score(e_q,v_u)$, so the score-opening subproof accepts with the digest $h_{\dtable}$ of Eq.~\eqref{eq:flat-score-hash}. Given the same table, the distinct nodes consumed by each traversal and the remaining active rows satisfy Eq.~\eqref{eq:subproof-score-table-consistency}. The honest upper-layer trace satisfies the authenticated transitions of Alg.~\ref{alg:hnsw-upper} and outputs its level-$0$ endpoint. By Theorem~\ref{thm:efsearch-proof}, the honest level-$0$ trace satisfies $\cef$, and by Eq.~\eqref{eq:efsearch-topk-output} its top-$k$ set is $F^{\mathsf{topk}}_{\WorkQ}=\mathsf{SS}(e_q,\indexset)$. Eligibility supplies a node $u_{\mathsf{sel}}$ with
\begin{equation}
u_{\mathsf{sel}}\in F^{\mathsf{topk}}_{\WorkQ}
\quad\land\quad
\mathsf{mid}(u_{\mathsf{sel}})=\model,
\label{eq:selected-node-witness}
\end{equation}
which the level-$0$ subproof checks. All three subproofs therefore accept, and completeness of the argument gives an accepting $\pi_{\mathsf{sel}}$. On a pinned turn, an honest serving event has $\model=\model_{\mathsf{pin}}$, so the pinned-model check accepts.

\noindent\textbf{Soundness.}
Consider an accepted fresh-selection turn, and let $w$ be the extracted witness of $\pi_{\mathsf{sel}}$. By Lemma~\ref{lem:multiset-error}, every multiset equality in $w$ holds exactly. By binding to the endorsed context, every reference or graph row opened in $w$ is the endorsed row of its key in $\indexset$. The three subproof witnesses hash to the same $h_{\dtable}$, so by collision resistance they contain the same table $\dtable$. The score-opening witness then gives, for every active row of $\dtable$,
\begin{equation}
(u,v_u,\mathsf{mid}(u))\in\rtable
\quad\land\quad
\score_u=\score(e_q,v_u),
\label{eq:authenticated-score-row}
\end{equation}
where $e_q$ is the query in $x_{\mathsf{sel}}$, which the client derives as $e_q=\encoder(\clientreq)$. In each traversal, Eq.~\eqref{eq:subproof-score-table-consistency} binds every distinct consumed node to its active row of $\dtable$, and the repeated-use checks extend this binding to every use of that node. Both traversals therefore use only the scores in Eq.~\eqref{eq:authenticated-score-row}.

By induction over the upper-layer steps, the stable-maximum and transition conditions of Sec.~\ref{sec:method-control-upper} make the upper-layer witness the execution of Alg.~\ref{alg:hnsw-upper} from $(u_{\mathsf{ent}},L_{\max})$ fixed by $\mathsf{pp}_{\mathsf{sel}}$, with output $\EntryNode$. By Theorem~\ref{thm:efsearch-proof}, the level-$0$ witness is an execution of Alg.~\ref{alg:hnsw-efsearch} from $\EntryNode$, and Eq.~\eqref{eq:efsearch-topk-output} gives $F^{\mathsf{topk}}_{\WorkQ}=\mathsf{SS}(e_q,\indexset)$. The model-membership constraint gives Eq.~\eqref{eq:selected-node-witness}, so
\[
\model\in\{\mathsf{mid}(u):u\in\mathsf{SS}(e_q,\indexset)\}
=\Elig(\clientreq,\indexset),
\]
and $\Sel_{\mathsf{fresh}}(\model;\clientreq,\indexset)$ holds. On a pinned turn, acceptance requires $\model=\model_{\mathsf{pin}}$ for the client-held model of the initial accepted fresh-selection turn, so $\Sel_{\mathsf{pinned}}(\model;\clientreq,\indexset)$ holds.

This reasoning fails only if extraction fails, a multiset check passes on unequal multisets, binding is broken, or a hash collision occurs. Each event has negligible probability, the second by Lemma~\ref{lem:multiset-error}, so by a union bound the claim holds except with negligible probability.
\end{proof}

\subsubsection{Serving-certificate Correctness}

\begin{theorem}[Serving-certificate correctness]
\label{thm:serving-cert}
Under the trust and key-registration model of Secs.~\ref{sec:threat} and~\ref{sec:method-initial}, an EUF-CMA signature scheme, and collision resistance of $H_{\mathsf{sig}}$, the serving certificate satisfies:
\begin{enumerate}[label=(\roman*),itemsep=0pt,topsep=1pt]
\item \uline{Completeness.} Every certificate honestly generated for a serving event is accepted for its matching unused $\rid$.
\item \uline{Soundness.} Except with negligible probability, every accepted certificate for $(\rid,\ReqCommit,Y,\model)$ authenticates an event identified by $\rid$ for which there exist $\receivedreq$ and $\rho_{\mathsf{srv}}$ satisfying
\[
Y=\model(\receivedreq)
\quad\land\quad
\ReqCommit=\Commit(\ReqHash(\receivedreq);\rho_{\mathsf{srv}}).
\]
\end{enumerate}
\end{theorem}

\begin{proof}
\noindent\textbf{Completeness.} For an honest event, $\provider_\model$ computes Eq.~\eqref{eq:serving-certificate}. The client recomputes $h_Y$ and $d_{\mathsf{srv}}$, obtains $pk_\model=\mathsf{Reg}[\model]$, and signature correctness gives $\mathsf{Verify}_{pk_\model}(d_{\mathsf{srv}},\sigma_\model)=1$ for the matching $\rid$.

\noindent\textbf{Soundness.} An accepted certificate carries a valid signature under $pk_\model=\mathsf{Reg}[\model]$ on $d_{\mathsf{srv}}=H_{\mathsf{sig}}(\rid\|\model\|\ReqCommit\|h_Y)$ with $h_Y=H_{\mathsf{sig}}(Y)$. The provider of $\model$ is honest: it signs only in a serving event in which it receives some $(\rid',\model,R,\rho)$, computes $c'=\Commit(\ReqHash(R);\rho)$ and $Y'=\model(R)$, and signs $H_{\mathsf{sig}}(\rid'\|\model\|c'\|H_{\mathsf{sig}}(Y'))$. If the provider never signed $d_{\mathsf{srv}}$, the certificate is an EUF-CMA forgery. Otherwise $d_{\mathsf{srv}}$ is the digest of such an event, and since the encoding is injective, $(\rid,\ReqCommit,Y)\neq(\rid',c',Y')$ would give a collision of $H_{\mathsf{sig}}$. Thus the event identified by $\rid$ satisfies the claim with $\receivedreq=R$ and $\rho_{\mathsf{srv}}=\rho$. Forgeries and collisions have negligible probability, so the claim holds except with negligible probability.
\end{proof}

\subsubsection{Rewrite-proof Correctness}

\begin{theorem}[Rewrite-proof correctness]
\label{thm:rewrite-proof}
Under collision resistance of $H_{\mathsf{zk}}$, computational binding of the commitments, and completeness and knowledge soundness of the argument system, the rewrite proof satisfies:
\begin{enumerate}[label=(\roman*),itemsep=0pt,topsep=1pt]
\item \uline{Completeness.} For every endorsed rewrite within the public bounds, an honestly generated $\pi_{\mathsf{rw}}$ is accepted.
\item \uline{Soundness.} For every PPT adversary, there exists a PPT extractor with access to the adversary and the endorsed banks such that, except with negligible probability, whenever $\pi_{\mathsf{rw}}$ is accepted for $x_{\mathsf{rw}}$, the extractor outputs a request $R_{\mathsf{rw}}\in\trans(\clientreq)$ and a randomizer $\rho_{\mathsf{rw}}$ with
\[
\ReqCommit=\Commit(\ReqHash(R_{\mathsf{rw}});\rho_{\mathsf{rw}}).
\]
\end{enumerate}
\end{theorem}

\begin{proof}
\noindent\textbf{Completeness.} Fix $\receivedreq\in\trans(\clientreq)$ within the public bounds. Membership supplies the gaps, fragments and bank openings. Partition $[N_{\mathsf{cmsg}}]$ into consecutive ranges over a nonempty prefix of the leaf slots, assign to each active leaf the digest subsequence $\mu^{(a)}$ of its range in the order of Eq.~\eqref{eq:gformat}, and set $N^{\mathsf{dig}}_a=|\mu^{(a)}|$ and $s^{(a)}=\ReqFold(\mathsf{IV};\mu^{(1)}\mathbin\Vert\cdots\mathbin\Vert\mu^{(a)})$. Every leaf relation then holds, $N^{\mathsf{dig}}_{\mathsf{rw}}=|\receivedreq|$, and Eq.~\eqref{eq:structured-request-hash} gives $d_{\mathsf{rw}}=\ReqHash(\receivedreq)$. The aggregation relation therefore holds for $\ReqCommit=\Commit(d_{\mathsf{rw}};\ReqRand)$, and completeness of the argument yields an accepting $\pi_{\mathsf{rw}}$.

\noindent\textbf{Soundness.} Let $w$ be the extracted witness of an accepted $\pi_{\mathsf{rw}}$. The aggregation constraints make the active leaf ranges partition $[N_{\mathsf{cmsg}}]$ in order, give each active leaf the messages $\clientreq|_{I_a}$ of the public $\clientreq$, chain the leaf fold values from $s^{(0)}=\mathsf{IV}$, and set $N^{\mathsf{dig}}_{\mathsf{rw}}=\sum_aN^{\mathsf{dig}}_a$. By binding to the endorsed context, every authenticated message digest in $w$ is $m(g)$ for some $g\in\gapset$, and every authenticated fragment is a fragment of $\fragset$. Replace each message digest by its message and concatenate the messages of the active leaves in range order. By $\mathsf{DigestLayout}$, the result $R_{\mathsf{rw}}$ has the form of Eqs.~\eqref{eq:gformat} and~\eqref{eq:fformat}, with at most $\MsgBound$ bank messages per gap, at most $\FragBound$ bank fragments on each side of every client message, and the trailing gap exactly once. Hence $R_{\mathsf{rw}}\in\trans(\clientreq)$.

The concatenated leaf sequences $\mu^{(1)}\mathbin\Vert\cdots\mathbin\Vert\mu^{(N_{\mathsf{leaf}})}$ are exactly $\mu_{R_{\mathsf{rw}}}$, so the chained folds give $s^{(N_{\mathsf{leaf}})}=\ReqFold(\mathsf{IV};\mu_{R_{\mathsf{rw}}})$ and $N^{\mathsf{dig}}_{\mathsf{rw}}=|R_{\mathsf{rw}}|$. The root equations then give $d_{\mathsf{rw}}=\ReqHash(R_{\mathsf{rw}})$ and $\ReqCommit=\Commit(\ReqHash(R_{\mathsf{rw}});\rho_{\mathsf{rw}})$ with $\rho_{\mathsf{rw}}=\ReqRand$ from $w$. This reasoning fails only if extraction fails, binding is broken, or a hash collision occurs. Each event has negligible probability, so by a union bound the claim holds except with negligible probability.
\end{proof}

\subsubsection{Proof of Data-plane Integrity}
\label{app:proof-data-plane}

\begin{proof}[Proof of Theorem~\ref{thm:data-plane}]
\noindent\textbf{Completeness.} For an honest endorsed rewrite, the router samples fresh $\ReqRand$ and sets $\ReqCommit=\Commit(\ReqHash(\receivedreq);\ReqRand)$. The honest provider processes $\receivedreq$ and constructs Eq.~\eqref{eq:serving-certificate}, which is accepted by Theorem~\ref{thm:serving-cert}. Since the same $\receivedreq$ lies in $\trans(\clientreq)$, Theorem~\ref{thm:rewrite-proof} gives an accepted $\pi_{\mathsf{rw}}$ under the same $\ReqCommit$.

\noindent\textbf{Soundness.} Suppose the client accepts $\mathsf{SC}$ and $\pi_{\mathsf{rw}}$ under the same $\ReqCommit$. By Theorem~\ref{thm:serving-cert}, the event identified by $\rid$ has a request $R_{\mathsf{srv}}$, the request actually received by $\model$, and a randomizer $\rho_{\mathsf{srv}}$ with $Y=\model(R_{\mathsf{srv}})$. By Theorem~\ref{thm:rewrite-proof}, the extractor outputs $R_{\mathsf{rw}}\in\trans(\clientreq)$ and $\rho_{\mathsf{rw}}$. Both pairs open $\ReqCommit$:
\[
\begin{aligned}
\ReqCommit&=\Commit(\ReqHash(R_{\mathsf{srv}});\rho_{\mathsf{srv}})\\
&=\Commit(\ReqHash(R_{\mathsf{rw}});\rho_{\mathsf{rw}}).
\end{aligned}
\]
If $\ReqHash(R_{\mathsf{srv}})\neq\ReqHash(R_{\mathsf{rw}})$, these are two openings of $\ReqCommit$ to different digests, which breaks binding. If the digests agree but $R_{\mathsf{srv}}\neq R_{\mathsf{rw}}$, they form a collision of $\ReqHash$, which by Lemma~\ref{lem:reqhash-cr} yields a collision of $H_{\mathsf{zk}}$. Otherwise $R_{\mathsf{srv}}=R_{\mathsf{rw}}\in\trans(\clientreq)$ and $Y=\model(R_{\mathsf{srv}})$. This reasoning fails only in the failure events of Theorems~\ref{thm:serving-cert} and~\ref{thm:rewrite-proof}, a binding violation, or a collision. Each event has negligible probability, so by a union bound the claim holds except with negligible probability.
\end{proof}

\subsubsection{Proof of End-to-end Security}
\label{sec:method-e2e}

\begin{proof}[Proof of Theorem~\ref{thm:verirouter-security}]
\noindent\textbf{Completeness.}
Consider a supported honest event. On a fresh-selection turn, Theorems~\ref{thm:control-plane} and~\ref{thm:data-plane} give an accepting $\pi_{\mathsf{sel}}$ and an accepted $(\pi_{\mathsf{rw}},\mathsf{SC})$. The protocol in Fig.~\ref{fig:zk-verirouter-e2e-protocol} uses the same $\model$ in $\pi_{\mathsf{sel}}$ and $\mathsf{Reg}[\model]$, the same $\ReqCommit$ in $\pi_{\mathsf{rw}}$ and $\mathsf{SC}$, and the client-issued $\rid$ in the certificate, so every client check accepts. On a pinned turn, the honest router reuses $\model_{\mathsf{pin}}$, so the pinned-model check accepts, and Theorem~\ref{thm:data-plane} supplies an accepted rewrite proof and serving certificate.

\noindent\textbf{Soundness.}
Suppose the client accepts a turn. On a fresh-selection turn, Theorem~\ref{thm:control-plane} gives the first conjunct of Eq.~\eqref{eq:verification-target}; on a pinned turn, acceptance requires $\model=\model_{\mathsf{pin}}$, which is that conjunct. Theorem~\ref{thm:data-plane} gives one $\receivedreq$ satisfying the other two conjuncts. The client verifies all statements with the same $\model$, obtains $pk_\model=\mathsf{Reg}[\model]$, and checks the certificate against the current $\rid$, so the three conjuncts hold for the same serving event. This reasoning fails only in the failure events of Theorems~\ref{thm:control-plane} and~\ref{thm:data-plane}. Each event has negligible probability, so by a union bound the claim holds except with negligible probability.

\noindent\textbf{Confidentiality.}
Assume adaptive multi-theorem zero knowledge for the root proofs. The simulator receives the public data: the context $\ctx$, the event data $(\rid,\clientreq,\model,Y)$, the selection mode, $\ReqCommit$, and the proof shapes determined by the public capacities. It replaces the fresh-turn selection proof and then the rewrite proof with simulated proofs for their true public statements, keeping the commitments and the serving certificate. Zero knowledge makes each replacement computationally indistinguishable. A pinned turn has no new $\pi_{\mathsf{sel}}$. The score-table digest $h_{\dtable}$ and the model-request message count $N^{\mathsf{dig}}_{\mathsf{rw}}$ are witness values of recursive aggregation, so the simulation also hides them.

Computational hiding of the commitments in Appendix~\ref{app:committed-lookup} hides the semantic index and rewrite banks, and fresh $\ReqRand$ makes $\ReqCommit$ hide $\ReqHash(\receivedreq)$. Apart from the public response $Y$, the serving certificate depends on $\receivedreq$ only through $\ReqCommit$, since it signs $H_{\mathsf{sig}}(\rid\|\model\|\ReqCommit\|h_Y)$. Hence the returned cryptographic evidence reveals no additional information about the router's private data beyond the public data and the fact that the accepted relations hold. This analysis concerns the returned evidence under the accepted $\ctx$; it does not cover inference from deployment-level choices among public capacity configurations or from system-level observables such as proof-generation latency.
\end{proof}
}

\subsection{Endorsement and Confidentiality Boundaries}
\label{app:discussion}
\noindent\textbf{Policy endorsement and context evolution.}
\toolname proves conformance to the specifications committed in the public context $\ctx$. Determining whether these specifications define an acceptable policy, and how they are audited, is outside our scope. One simple realization uses an independent auditor with access to the private semantic index and rewrite banks to validate policy-specific requirements, such as eligible-model assignments and the absence of disallowed sponsored content~\cite{schnabl2025attestable}. The auditor also verifies that the reviewed specifications and their commitment openings reproduce the commitments in $\ctx$, then endorses the context by signing $H_{\mathsf{sig}}(\ctx)$. This review checks the reference table's dimensions and 16-bit coordinate bounds, that node identifiers are unique, that each node and layer has exactly one graph row, and that every graph source and neighbor belongs to the reference table. It also approves the circuit profile and verifier keys used by the client. Such signed endorsements can be recorded in an append-only transparency log, following established designs in certificate, key, and software transparency~\cite{laurie2014certificate}. When the semantic index, rewrite banks, or retrieval content changes, a new context can be audited, endorsed, and appended to the log. Clients can verify the auditor's signature and log inclusion, while consistency proofs bind successive log checkpoints to an append-only history. This mechanism remains orthogonal to \toolname. The external auditor determines which contexts are endorsed, while \toolname efficiently proves that each serving event conforms to the context accepted by the client.

\noindent\textbf{Confidentiality boundaries.}
The confidentiality guarantee of \toolname concerns the returned cryptographic evidence under the accepted public context $\ctx$. Zero knowledge ensures that the model-selection and rewrite proofs reveal no additional information about their private witnesses beyond their public statements and the fact that the certified relations hold. Public statements and parameters can nevertheless expose coarse information. In particular, $\mathsf{pp}_{\mathsf{sel}}$ fixes padded capacities such as the score-table and traversal bounds. If a deployment chooses among multiple public capacity buckets, the selected bucket reveals a coarse bound on the corresponding private workload. Proof-generation latency may provide an additional signal when it correlates with the selected capacity. Using a single fixed capacity eliminates leakage due to capacity selection at additional proving cost, while a small set of buckets provides an explicit efficiency--leakage tradeoff. {Similarly, $\mathsf{pp}_{\mathsf{rw}}$ fixes the bank lengths used to range-check lookup indices; padding each bank to a public capacity hides its exact size.}

Beyond such capacity-related leakage, system-level observations may reveal additional information that is outside the cryptographic guarantee. Visible model identities, responses, and adaptively repeated queries may enable inference about the underlying routing policy. Such inference is complementary to \toolname and can be studied or mitigated separately~\cite{tramer2016stealing}. Likewise, \toolname does not hide the client's request from the router. Confidential-request processing can be provided by complementary techniques~\cite{lu2023bumblebee,akimoto2023privformer,chowdhury2025pr,huang2026your}.

\subsection{Fidelity of the Quantized Routing}
\label{app:quantized-fidelity}

\toolname's cryptographic guarantee is with respect to the committed 16-bit affine-quantized semantic index under its public search configuration. We separately evaluate whether choosing that quantized index materially changes the outputs of a floating-point semantic search. We compare floating-point and quantized semantic search result sets under three label spaces: retrieved node identifiers (IDs), model IDs, and model-tier IDs. Following RouteLLM~\cite{ong2025routellm}, a tier ID groups concrete models into a coarser routing tier. For each query and label space, let $A_k$ and $B_k$ be the sets obtained by projecting the floating-point and quantized top-$k$ outputs, respectively. We define overlap as $|A_k\cap B_k|/\min(|A_k|,|B_k|)$ and Jaccard similarity as $|A_k\cap B_k|/|A_k\cup B_k|$, and report their averages over queries. Floating-point (FP) Top-1 Hit@$k$ is the fraction of queries for which the floating-point top-ranked node's label appears among the quantized top-$k$ labels.

Tab.~\ref{tab:rq2-fidelity} shows high agreement between quantized and floating-point semantic search across all three label spaces. Even at the finest node-ID level, overlap remains between 89.11\% and 90.25\%. More importantly for the verification target, the eligible model set is highly stable: model-ID overlap ranges from 91.62\% to 95.78\%, and the model label of the floating-point top-ranked node appears among the quantized top-$k$ labels for $92.16\%$--$96.21\%$ of queries. Agreement is stronger at the tier level, where overlap reaches 95.91\%--99.52\% and FP Top-1 Hit@$k$ reaches 97.78\%--100.00\%. Thus, 16-bit quantization primarily perturbs individual retrieved references while largely preserving the eligible models associated with them.

\input{tables/fidelity}

%% file: tables/fidelity.tex
\begin{table}[H]
\centering
\small
\setlength{\tabcolsep}{4.4pt}
\caption{Routing-quality comparison of quantized and floating-point HNSW. \normalfont We compare result sets after projection to node IDs, model IDs, and tier IDs; certification and deployment both use the quantized index. The results show high agreement across all label spaces.}
\label{tab:rq2-fidelity}
\resizebox{0.85\columnwidth}{!}{%
\begin{tabular}{llcccc}
\toprule
\textbf{Label space} & \textbf{Metric} & \textbf{Top-5} & \textbf{Top-10} & \textbf{Top-20} & \textbf{Top-50} \\
\midrule
\multirow{3}{*}{Node ID} & Overlap & 89.71\% & 90.12\% & 90.25\% & 89.11\% \\
 & Jaccard & 87.95\% & 87.97\% & 87.83\% & 85.10\% \\
 & {FP Top-1 Hit@$k$} & 91.64\% & 91.70\% & 93.21\% & 93.21\% \\
\midrule
\multirow{3}{*}{Model ID} & Overlap & 91.62\% & 92.14\% & 94.80\% & 95.78\% \\
 & Jaccard & 89.70\% & 89.72\% & 92.16\% & 92.74\% \\
 & {FP Top-1 Hit@$k$} & 92.16\% & 93.01\% & 95.69\% & 96.21\% \\
\midrule
\multirow{3}{*}{Tier ID} & Overlap & 95.91\% & 98.17\% & 98.42\% & 99.52\% \\
 & Jaccard & 94.17\% & 97.37\% & 97.22\% & 97.93\% \\
 & {FP Top-1 Hit@$k$} & 97.78\% & 99.74\% & 99.93\% & 100.00\% \\
\bottomrule
\end{tabular}%
}
\end{table}